\documentclass
[journal,onecolumn]
{IEEEtran}
\usepackage{amsthm}
\usepackage{amsmath}
\usepackage{mathrsfs}
\usepackage{amssymb} 
\usepackage{bbm,dsfont} 
\usepackage{multirow} 
\usepackage{units}
\usepackage{stmaryrd}   
\usepackage{verbatim}
\usepackage{enumerate} 
\usepackage[table]{xcolor}
\usepackage{slashbox} 
\usepackage{ wasysym }
\usepackage{tikz}
\usetikzlibrary{calc}
\usepackage{tabularx}

\usepackage{lipsum}
\usepackage{multicol}
\usepackage{float}

\usepackage{ifthen}

\usepackage{xcolor}
\usepackage{hyperref} 
\usepackage
{hyperref} 
\hypersetup{
    colorlinks,
    linkcolor={blue!80!black},
    citecolor={green!50!black},
    urlcolor={blue!80!black}
}

\newcommand\blfootnote[1]{%
  \begingroup
  \renewcommand\thefootnote{}\footnote{#1}%
  \addtocounter{footnote}{-1}%
  \endgroup
}

\usepackage{graphicx}
\graphicspath{{images/}}

\title{The Arbitrarily Varying Channel with Colored Gaussian Noise}

\author{ \IEEEauthorblockN{Uzi Pereg$^{\scriptstyle \, 1}$ and Yossef Steinberg$^{\scriptstyle \, 2}$}\\ 
\IEEEauthorblockA{
$^{1}$  Institute for Communications Engineering, Technical University of Munich \\
$^{2}$  Department of Electrical Engineering, Technion  \\
Email: {\tt uzi.pereg@tum.de}, {\tt ysteinbe@ee.technion.ac.il}
 }
}

\usepackage[square,numbers]{natbib}
\definecolor{light-gray}{gray}{0.8}

\definecolor{dark-gray}{gray}{0.3}
\usepackage{accents}
\newlength{\dhatheight}

\newcommand{\bieee}{\begin{IEEEeqnarray}{rCl}}
\newcommand{\eieee}{\end{IEEEeqnarray}}
\newcommand{\prob}[1]{\Pr\left(#1\right)}
\newcommand{\given}{\mid}
\newcommand{\cprob}[2]{\Pr\left(#1\given #2\right)}
\newcommand{\E}{\mathbb{E}}

\newcommand{\eps}{\varepsilon}

\newcommand{\norm}[1]{\left\lVert#1\right\rVert}
\newcommand{\diag}{\mathrm{diag}}
\newcommand{\trace}{\mathrm{tr}}

\newcommand{\ie}{\emph{i.e.} }
\newcommand{\eg}{\emph{e.g.} }
\newcommand{\etal}{\emph{et al.} }

\newcommand{\xvec}{\mathbf{x}}
\newcommand{\yvec}{\mathbf{y}}
\newcommand{\zvec}{\mathbf{z}}
\newcommand{\fvec}{\mathbf{f}}

\newcommand{\svec}{\mathbf{s}}

\newcommand{\Xvec}{\mathbf{X}}
\newcommand{\Yvec}{\mathbf{Y}}
\newcommand{\Svec}{\mathbf{S}}

\newcommand{\Vvec}{\mathbf{V}}
\newcommand{\Zvec}{\mathbf{Z}}		
\newcommand{\tm}{\widetilde{m}}	
																				
\newcommand{\tq}{\widetilde{q}}

\newcommand{\tX}{\widetilde{X}}
\newcommand{\tY}{\widetilde{Y}}

\newcommand{\tS}{\widetilde{S}}

\newcommand{\tx}{\tilde{x}}
\newcommand{\ts}{\widetilde{s}}

\newcommand{\tlambda}{\widetilde{\lambda}}
\newcommand{\tLambda}{\widetilde{\Lambda}}
\newcommand{\oS}{\overline{S}}

\newcommand{\os}{\overline{s}}

\newcommand{\oq}{\overline{q}}

\newcommand{\hm}{\hat{m}}

\newcommand{\hq}{\widehat{q}}

\newcommand{\hP}{\hat{P}}
\newcommand{\hM}{\hat{M}}

\newcommand{\Aset}{\mathcal{A}}

\newcommand{\Dset}{\mathcal{D}}
\newcommand{\Fset}{\mathcal{F}}

\newcommand{\Iset}{\mathcal{I}}
\newcommand{\Jset}{\mathcal{J}}

\newcommand{\Qset}{\mathcal{Q}}

\newcommand{\Sset}{\mathcal{S}}
\newcommand{\Wset}{\mathcal{W}}
\newcommand{\Xset}{\mathcal{X}}
\newcommand{\Yset}{\mathcal{Y}}

\newcommand{\Eset}{\mathcal{E}}
\newcommand{\markovC}[1]{%
\begin{tikzpicture}[#1]%
\draw (0,0.3ex) -- (1ex,0.3ex);%
\draw (0.5ex,0.3ex) circle (0.2ex);
\draw[white] (0.2ex,0) -- (0.5ex,0);%
\end{tikzpicture}%
}
\newcommand{\Cbar}{\markovC{scale=2}}

\theoremstyle{remark}	\newtheorem{theorem}{Theorem}
\theoremstyle{remark}	\newtheorem{lemma}[theorem]{Lemma}
\theoremstyle{remark}	\newtheorem{coro}[theorem]{Corollary}
\theoremstyle{remark} \newtheorem{definition}{Definition}
\theoremstyle{remark} \newtheorem{remark}{Remark}
\theoremstyle{remark} \newtheorem{example}{Example}

\newcommand{\channel}{W_{Y|X,S}}
\newcommand{\compound}{\Wset^\Qset} 
\newcommand{\avc}{\Wset}																		
\newcommand{\opC}{\mathbb{C}}																
\newcommand{\inC}{\mathsf{C}}															 	
\newcommand{\inR}{\mathsf{R}}

\newcommand{\pSpace}{\mathcal{P}}														

\newcommand{\dE}{\mathsf{E}}																
	
\newcommand{\dM}{\mathsf{M}}															 	

\newcommand{\enc}{\mathrm{f}}																				
\newcommand{\dec}{g}																			 	
\newcommand{\code}{\mathscr{C}}															
\newcommand{\gcode}{\mathscr{C}^{\,\Gamma}}									

\newcommand{\Gcerr}{P_{e|\svec^d}^{(n)}}													
\newcommand{\err}{P_e^{(n)}}															
\newcommand{\cost}{\phi}																		
\newcommand{\plimit}{\Omega}																			
\newcommand{\tset}{\Aset^{(n)}_{\delta}}													
\newcommand{\Tset}{\mathcal{T}}												
\newcommand{\qn}{q}
\newcommand{\tQ}{\hat{\Qset}_n}														

\newcommand{\encn}{\enc}																			

\newcommand{\Ccompound}{\opC(\compound)}
\newcommand{\Cavc}{\opC(\avc)}

\newcommand{\emp}{\hP}																		  

\newcommand{\rstarC}{																			  
\, \hspace{-0.3cm} \text{ $$ \mbox{  
\hspace{-0.1cm} 
\small $\star$   
} $$ }
\hspace{-0.25cm}}

\newcommand{\rCcompound}{\opC^{\rstarC}\hspace{-0.1cm}(\compound)}

\newcommand{\rCav}{\opC^{\rstarC}\hspace{-0.1cm}(\avc)}

\newcommand{\rICav}{\inC^{\rstarC}}

\newcommand{\LambdaOig}{\tLambda} 			
\newcommand{\apLSpaceS}{\overline{\pSpace}_\Lambda(\Sset|\theta^{\infty})}			
\newcommand{\pLSpaceS}{\overline{\pSpace}_\Lambda(\Sset|\theta^{\infty})}			
\newcommand{\pLSpaceSn}{\pSpace_\Lambda(\Sset^n|\theta^n)}		

\newcommand{\KrCav}{\opC^{\rstarC}\hspace{-0.1cm}(K_Z)} 			
\newcommand{\KCavc}{\opC(K_Z)} 					  																	

\newcommand{\KrICav}{\inC^{\rstarC}_n(K_Z)} 		

\newcommand{\sigmarCav}{\opC^{\rstarC}\hspace{-0.1cm}(\Sigma)} 			
\newcommand{\sigmaCavc}{\opC(\Sigma)} 					  																	

\newcommand{\sigmarICav}{\inC^{\rstarC}\hspace{-0.1cm}(\Sigma)} 		
\newcommand{\sigmaICavc}{\inC(\Sigma)} 																							

\newcommand{\sKrCav}{\opC^{\rstarC}\hspace{-0.1cm}(\Psi_Z)} 			
\newcommand{\sKCavc}{\opC(\Psi_Z)} 					  																	

\newcommand{\sKrICav}{\inC^{\rstarC}\hspace{-0.1cm}(\Psi_Z)} 		

\begin{document}
\maketitle

{}

\begin{abstract} 
%
We address the arbitrarily varying channel (AVC) with colored Gaussian noise. The work consists of three parts. 
First, we study the \emph{general} discrete AVC with fixed parameters, where the channel depends on two state sequences, one arbitrary and the other fixed and known.
 This model can be viewed as a combination of the AVC and the time-varying channel.  We determine both the deterministic code capacity and the random code capacity.
Super-additivity is demonstrated, showing that the deterministic code capacity can be strictly larger than the weighted sum of the parametric capacities. 

In the second part, we consider the arbitrarily varying Gaussian product channel (AVGPC).
Hughes and Narayan characterized 
 the random code capacity through min-max optimization leading to a 
``double" water filling solution.
Here, we establish the deterministic code capacity and 
also discuss the game-theoretic meaning 
and the connection between double water filling and Nash equilibrium. As in the case of the standard Gaussian AVC, the deterministic code capacity is discontinuous in the input constraint, and depends on which of the input or state constraint is higher.
As opposed to Shannon's classic water filling solution,   it is observed that deterministic coding using independent scalar codes  is suboptimal for the AVGPC.

 Finally, we establish the capacity of the AVC with colored Gaussian noise,
where double water filling is performed 
in the frequency domain.
The analysis relies on our preceding results, on the AVC with fixed parameters and the AVGPC.
\end{abstract}


\blfootnote{
This work was supported by the Israel Science Foundation (grant No. 1285/16).
}

\section{Introduction}
A channel with colored Gaussian noise was first studied by Shannon \cite{Shannon:49p}, introducing the water filling optimal power allocation. 
This channel is the spectral counterpart of the Gaussian product channel (see \eg \cite[Section 9.5]{CoverThomas:06b}). 
Those results led to useful algorithms for DSL and OFDM systems,  and were generalized to multiple-input multiple output (MIMO) wireless communication systems as well (see  \eg \cite{Telatar:99p,Foschini:96p,BCCGPP:07b,BiglieriProadisShamai:98p,ShamaiOzarowWyner:91p,GJJV:03p}). 
Furthermore, for some networks,  water filling is performed in multiple stages \cite{ChengVerdu:93p,WCLM:99p,YuCioffi:02p,YuRheeBoydCioffi:04p,LaiElGamal:08p,WangAggarwalWang:15p}. 
A limit formula for the capacity of the general time-varying channel (TVC) is given in \cite{VerduHan:94p} (see also
\cite{CsiszarKorner:82b,Han:13b,Ahlswede:68p,DasNarayan:02p,BarbarossaScaglione:99p,Martone:00p,SaggarPottieDaneshrad:16p,WangOrchard:03p}).
Another relevant setting is that of a finite-state channel, where the state evolves as a Markov chain
\cite{Wolfowitz:12b,LapidothTelatar:98p,BBT:58p,LapidothNarayan:98p,Han:15p,ThomasEckford:16p,SDJR:17p}.
In practice, there is often uncertainty regarding channel statistics, 
due to a variety of causes such as fading in wireless communication \cite{SimonAlouini:05b,ShamaiSteiner:03p,ASAMN:17p,OzarowShamaiWyner:94p,GoldsmithVaraiya:97p,
CaireShamai:99p,HosseinigokiKosut:19p,HosseinigokKosut:19c1}, memory faults in storage \cite{KuznetsovTsybakov:74p,HeegardElGamal:83p,KuzntsovHanVinck:94p,KimKumar:14c}, malicious attacks on  identification 
 systems \cite{GungorKoksalElGamal:13c,IgnatenkoWillems:12n}, and cyber-physical warfare 
\cite{SlayMiller:08c,Langner:11p,VoraKulkarni:19a}. 
The arbitrarily varying channel (AVC) is an appropriate model to describe such a situation \cite{BBT:60p,LapidothNarayan:98p}.

Blackwell \etal  \cite{BBT:60p} determined  the random code capacity of the general AVC, \ie the capacity achieved with shared randomness between the encoder and the decoder. It was also demonstrated in  \cite{BBT:60p}  that the random code capacity is not necessarily  achievable using deterministic codes. 
  A well-known result by Ahlswede \cite{Ahlswede:78p} is the dichotomy property of the AVC, \ie the deterministic code capacity, also referred to as `capacity', either equals the random code capacity or else, it is zero. 
Subsequently, Ericson \cite{Ericson:85p} and Csisz{\'{a}}r and Narayan \cite{CsiszarNarayan:88p}
 have established a simple single-letter condition, namely non-symmetrizability, which is both necessary and sufficient for the 
capacity to be positive.  Schaefer \etal \cite{SchaeferBochePoor:16c} demonstrated the 
 super-additivity phenomenon, \ie when the capacity of a product of 
orthogonal AVCs 
is 
strictly larger than the sum of the capacities of the components.  
Csisz{\'{a}}r and Narayan \cite{CsiszarNarayan:88p1, CsiszarNarayan:88p} also considered the AVC when input and state constraints are imposed on the user and the jammer, respectively, due to their power limitations.
Not only the constrained setting provokes serious technical difficulties analytically, but also, as shown in \cite{CsiszarNarayan:88p}, constraints have a significant effect on the behavior of the 
capacity. Specifically, it is shown in \cite{CsiszarNarayan:88p} that dichotomy in the sense of \cite{Ahlswede:78p} no longer holds when state constraints are imposed on the jammer. That is, the deterministic code capacity of the general AVC can be lower than the random code capacity, and yet non-zero.

The 
Gaussian AVC is specified by the relation $\Yvec=\Xvec+\Svec+\Zvec$, where
$\Xvec$ and $\Yvec$ are the input and output sequences, respectively; 
 $\Svec$ is a state sequence of unknown joint distribution $F_{\Svec}$, not necessarily independent nor stationary;  and the noise sequence $\Zvec$ is 
i.i.d. $\sim\mathcal{N}(0,\sigma^2)$.
The state sequence can be thought of as if generated by an adversary, or a \emph{jammer}, who randomizes the channel states arbitrarily in an attempt to disrupt communication. It is also possible for $\Svec$ to be a deterministic unknown state sequence.
It is assumed that the user and the jammer have power limitations, and are subject to input and state constraints, $\frac{1}{n}\sum_{i=1}^n X_{i}^2\leq\plimit$ and $\frac{1}{n}\sum_{i=1}^n S_i^2\leq\Lambda$, 
respectively, where $n$ is the transmission length.
 In \cite{HughesNarayan:87p},  Hughes and Narayan showed that the
random code capacity 
 is given by $\inC^{\rstarC}_1=\frac{1}{2}\log(1+\frac{\plimit}{\sigma^2+\Lambda})$.
%
%
%
%
Subsequently,  Csisz{\'{a}}r and Narayan \cite{CsiszarNarayan:91p} showed that the deterministic code capacity is given by
\begin{align}
\inC_1=\begin{cases}
\inC^{\rstarC}_1 &\text{if $\Lambda<\plimit$}\,,\\
0 &\text{if $\Lambda\geq\plimit$}\,.
\end{cases}
\label{eq:Gavc1}
\end{align}
It is noted in \cite{CsiszarNarayan:91p} that this result is \emph{not} a straightforward consequence of the elegant Elimination Technique \cite{Ahlswede:78p}, used by Ahlswede to establish dichotomy 
for the AVC
 without constraints. 
Hosseinigoki and Kosut \cite{HosseinigokKosut:19c1} determined the capacity in 
multiple side information scenarios for the Gaussian AVC with fast fading.
Hughes and Narayan \cite{HughesNarayan:88p} determined the random code capacity of  the arbitrarily varying Gaussian product channel (AVGPC), and showed that it 
is obtained as a ``double" water filling solution to an optimization min-max problem, maximizing over input power allocation and minimizing over state power allocation. 
In the solution, the jammer performs water filling first, attempting to whiten the overall noise as much as possible, and then the user performs water filling taking into account the total interference power, 
contributed by both the channel noise and the jamming signal  \cite{HughesNarayan:88p}.
The Gaussian AVC is also considered in \cite{Ahlswede:73c,ThomasHughes:91p,LaAnantharam:04p,SarwateGastpar:06c,SarwateGastpar:12a,HosseinigokiKosut:18c,HosseinigokiKosut:19p}.

Extensive research has been conducted on other AVC models as well, of which we name a few.
Recently, the arbitrarily varying wiretap channel has been extensively studied,  as \eg in \cite{MBL:09c,BocheShaefer:13p,ACD:13b,BocheShaeferPoor:14c,BocheSchaeferPoor:15p,NotzelWieseBoche:16p,HeLuo:17a,AhlswedeAlthoferDepperTamm:19b}, 
including input and state constraints in \cite{BjelakovicBocheSommerfeld:13b,JWNBJ:15c,GoldfeldCuffPermuter:16p}. 
 The capacity region of the arbitrarily varying multiple access channel (MAC) with and without constraints is characterized in 
\cite{PeregSteinberg:19p4,Jahn:81p,AhlswedeCai:96p,AhlswedeCai:99p};
 capacity bounds for the arbitrarily varying broadcast channel are derived in \cite{Jahn:81p,HofBross:06p}; and for the arbitrarily varying relay channel in \cite{PeregSteinberg:19p2E,PeregSteinberg:18c2}.
Additional results on arbitrarily  varying multi-user channels and constraints  are derived \eg in  \cite{WinshtokSteinberg:06c,BudkuleyDeyPrabhakaran:17p,HeKhistiYener:13p,WieseBoche:13p,PeregSteinberg:19p3,KeresztfalviLapidoth:19p}. 
Transmission of an arbitrarily varying Wyner-Ziv source over a Gel'fand-Pinsker channel is considered in \cite{WinshtokSteinberg:06c2,Winshtok:07z}, 
and  related problems were recently presented in 
\cite{BudkuleyDeyPrabhakaran:17p,BudkuleyJaggi:18c,BudkuleyJaggi:18a}.
Various Gaussian AVC networks are studied \eg in \cite{SarwateGastpar:08c,HeYener:11c,BudkuleyDeyPrabhakaran:15c,HosseinigokiKosut:16c,HosseinigokiKosut:17a,PeregSteinberg:19p1,PeregSteinberg:19p2E,PeregSteinberg:19p4,HosseinigokiKosut:19c}.

%

In this paper, we address the AVC with colored Gaussian noise. The body of this manuscript consists of three parts, of which the first and the second can also be viewed as milestones on our path to the main result.
First, we study the \emph{general} discrete AVC with fixed parameters. This model is a combination of the TVC and the AVC, as the channel depends on two state sequences, one arbitrary and the other fixed.
We determine both the deterministic code capacity and the random code capacity.
 Deterministic code super-additivity is demonstrated, showing that the capacity can be strictly larger
 than the weighted sum of the parametric capacities. 
In the second part of this paper, we establish the deterministic code capacity of the AVGPC, where there is \emph{white} Gaussian noise and no parameters.
We also give observations and discuss the game-theoretic interpretation of Hughes and Narayan's random code characterization \cite{HughesNarayan:88p}, and the connection between the double water filling solution and the idea of Nash equilibrium in game theory.
We further examine the connection between the AVGPC and the product MAC \cite{ChengVerdu:93p,LaiElGamal:08p}
 (without a state), pointing out the similarities and differences between the models, results, and interpretation. 
As in the case of the standard Gaussian AVC, the deterministic code capacity is discontinuous in the input constraint, and depends on which of the input or state constraint is higher.
As opposed to Shannon's classic water filling solution \cite{Shannon:49p},   it is observed that deterministic coding using independent scalar codes  is suboptimal for the AVGPC. Finally, we establish the capacity of the AVC with colored Gaussian noise,
where double water filling is performed 
in the frequency domain.

While the results on the AVC with fixed parameters and on the AVGPC stand in their own right, they also play a key role in our proof of the main capacity theorem 
for the AVC with colored Gaussian noise.
In the random code analysis for the AVC with fixed parameters, we modify  Ahlswede's Robustification Technique (RT) \cite{
Ahlswede:86p}. Essentially, the RT uses a reliable code for the compound channel to construct a random code for the AVC  applying random permutations to the codeword symbols. 
 A straightforward application of
 Ahlswede's RT does not work here, since the user cannot apply permutations to the parameter sequence. 
Hence, we give a modified RT which is restricted to permutations that do not affect the parameter sequence, \ie such that the parameter sequence is an eigenvector of all of our permutation matrices. 
The second part of the paper builds on identifying the symmetrizing jamming strategies and minimal symmetrizability costs for the AVGPC.
At last, we use the results on the AVC with fixed parameters and the AVGPC in our proof of the capacity theorem for the AVC with colored Gaussian noise.
 By orthogonalization of the noise covariance, the AVC with colored Gaussian noise is transformed into an AVC with fixed parameters, which are determined by the spectral representation of the noise covariance matrix. 
This in turn yields double water-filling optimization in analogy to the AVGPC.

\section{Channels with Fixed Parameters}
\label{sec:Pchannels}
In this section we consider the AVC with fixed parameters. The results in this section will be used to analyze the AVC with colored Gaussian noise.

\subsection{Notation}
\label{sec:Pnotation}
We use the following notation.
Calligraphic letters $\Xset,\Sset,\Tset,\Yset,...$ are used for finite sets.
Lowercase letters $x,s,t,y,\ldots$  stand for constants and values of random variables, and uppercase letters $X,S,T,Y,\ldots$ stand for random variables.  
 The distribution of a random variable $X$ is specified by a probability mass function (pmf) 
	$P_X(x)=p(x)$ over a finite set $\Xset$. The set of all pmfs over $\Xset$ is denoted by $\pSpace(\Xset)$. The set of all probability kernels $p(x|t)$ is denoted by 
	$\pSpace(\Xset|\Tset)$.
	%
		%
 We use $x^j=(x_1,x_{2},\ldots,x_j)$ to denote  a sequence of letters from $\Xset$. 
 A random sequence $X^n$ and its distribution $P_{X^n}(x^n)=p(x^n)$ are defined accordingly. 
For a pair of integers $i$ and $j$, $1\leq i\leq j$, we define the discrete interval $[i:j]=\{i,i+1,\ldots,j \}$.
  
The type $\hP_{x^n}$ of a given sequence $x^n$ is defined as the empirical distribution $\hP_{x^n}(a)=N(a|x^n)/n$ for $a\in\Xset$, where $N(a|x^n)$ is the number of occurrences of the symbol $a$ in the sequence $x^n$.
A type class is denoted by $\Tset^n(\hP)=\{ x^n \,:\; \hP_{x^n}=\hP \}$.
Similarly, define the joint type 
$\hP_{x^n,y^n}(a,b)=N(a,b|x^n,y^n)/n$ for $a\in\Xset$, 
$b\in\Yset$, where $N(a,b|x^n,y^n)$ is the number of occurrences of the symbol pair $(a,b)$ in the sequence $(x_i,y_i)_{i=1}^n$.
Then, a conditional type is defined as $\hP_{x^n|y^n}(a,b)=\hP_{x^n,y^n}(a,b)/\hP_{y^n}(b)$. Furthermore, we define
 the $\delta$-typical set $\tset(p)$ with respect to a distribution $p(x)$ by
\begin{align}
\tset(p)\triangleq \Big\{ x^n\in\Xset^n :\, \forall \, a\in\Xset \,,\; 
&\left| p(a)-\hP_{x^n}(a)  \right|\leq \delta \; \text{if $p(a)>0$, and} \nonumber\\
&\hP_{x^n}(a)=0 \;\text{if $p(a)=0$} 
\Big\} \,.
\end{align}

		The distribution of a real random variable $Z\in\mathbb{R}$ is represented by a cumulative distribution function (cdf) 
		$F_Z(z)=\prob{Z\leq z}$ over the real line, or alternatively, the probability density function (pdf) $f_Z(z)$,  when it exists.
%
%
The notation $\zvec=(z_1,z_{2},\ldots,z_n)$ is used 
when it is understood from the context that the length of the sequence is $n$, and  the $\ell^2$-norm of $\zvec$ is denoted  by $\norm{\zvec}$. The trace of a matrix $A\in\mathbb{R}^{m\times n}$ is denoted by $\trace(A)$.

	\subsection{Channel Description}
	\label{subsec:Mchannels}
 A state-dependent discrete memoryless channel (DMC) with parameters 
$(\Xset\times\Sset\times\Tset,W_{Y|X,S,T},\Yset)$ consists of  finite input alphabet $\Xset$, state alphabet $\Sset$, 
parameters alphabet $\Tset$, 
output alphabet $\Yset$,  and a 
conditional pmf 
$W_{Y|X,S,T}$ over $\Yset$. The channel is without feedback, and it is memoryless when conditioned on the state and parameter sequences, \ie  
\begin{align}
W_{Y^n|X^n,S^n,T^n}(y^n|x^n,s^n,t^n)= \prod_{i=1}^n W_{Y|X,S,T}(y_i|x_{i},s_{i},t_i) \,.
\end{align} 

The AVC with fixed parameters is a DMC $W_{Y|X,S,T}$ where the parameter sequence is fixed, while the state sequence has an unknown distribution,  not necessarily independent nor stationary. That is, the parameter is sequence is given by
\begin{align}
T^n=\theta^n \,,
\end{align}
where $\theta_1,\theta_2,\ldots$ is a given sequence of letters from $\Tset$, known to the encoder, decoder, and jammer. 
Whereas, the state sequence $S^n\sim \qn(s^n|\theta^n)$ with an unknown joint pmf $\qn(s^n|\theta^n)$ over $\Sset^n$. In particular, $\qn(s^n|\theta^n)$ could give mass $1$ to some state sequence $s^n$. 
The AVC with fixed parameters 
is denoted by $\avc=\{ W_{Y|X,S,T},\theta^{\infty} \}$, where $\theta^{\infty}$ is a short notation for the sequence 
$(\theta_i)_{i=1}^{\infty}$.

The compound channel with fixed parameters is used as a tool in the analysis.  Different models of compound channels are described in the literature \cite{CsiszarKorner:82b}. Here,   the compound channel  with fixed parameters is  a DMC $W_{Y|X,S,T}$  where the state has a conditional product distribution $q(s|t)$ that is not known
 in exact, but rather belongs to a family of conditional distributions $\Qset$, with $\Qset\subseteq \pSpace(\Sset|\Tset)$. 
That is, 
\begin{align}
S^n\sim \prod_{i=1}^n q(s_i|\theta_i)
\end{align}
with an unknown conditional pmf $q(s|t)\in\Qset$.
We note that this differs from the classical definition of the compound channel, as in \cite{CsiszarKorner:82b}, where the state is fixed throughout the transmission.

\begin{remark}
\label{rem:constT}
Note that the special case of a channel $W_{Y|X,S,T=t}$, with a \emph{constant} parameter  $\theta_i=t$ for $i=1,2,\ldots$, reduces to the standard state-dependent DMC. Thereby,  the AVC $\avc_t=\{ W_{Y|X,S,T=t} \}$ with a constant parameter can be regarded as the traditional AVC, as introduced by Blackwell \etal \cite{BBT:60p}.
On the other hand, the special case of a channel $W_{Y|X,S,T}= W_{Y|X,T}$, which does not depend on the state $S$, reduces to a TVC 
\cite{VerduHan:94p}.
\end{remark}

\begin{remark}
The AVC with colored Gaussian noise does \emph{not} fit the description above. Nevertheless,
the fixed parameters model is a crucial tool for our final goal, \ie to determine the capacity of the AVC with colored Gaussian noise. 
\end{remark}

\subsection{Coding}
\label{subsec:Mcoding}
We introduce some preliminary definitions. 

\begin{definition}[Code] 
\label{def:Pcapacity}
A $(2^{nR},n)$ code for the AVC $\avc$ with fixed parameters 
consists of the following;   
a message set $[1:2^{nR}]$, where $2^{nR}$ is assumed to be an integer, an encoding function
 $\enc:[1:2^{nR}]\times\Tset^n\rightarrow \Xset^n$, and a decoding function
$
\dec: \Yset^n\times\Tset^n\rightarrow [1:2^{nR}]  
$. 

Given a message $m\in [1:2^{nR}]$ and and a parameter sequence $\theta^n$, the encoder transmits the codeword $x^n=\enc(m,\theta^n)$. 
The decoder  receives the channel output $y^n$, and finds an estimate of the message $\hm=g(y^n,\theta^n)$.  We denote the code by $\code=\left(\enc(\cdot,\cdot),\dec(\cdot,\cdot) \right)$.
\end{definition}

We proceed now to coding schemes 
when using stochastic-encoder stochastic-decoder pairs with common randomness.

\begin{definition}[Random code]
\label{def:PcorrC} 
A $(2^{nR},n)$ random code for the AVC $\avc$ with fixed parameters consists of a collection of 
$(2^{nR},n)$ codes $\{\code_{\gamma}=(\enc_{\gamma},\dec_\gamma)\}_{\gamma\in\Gamma}$, along with a probability distribution $\mu(\gamma)$ over the code collection $\Gamma$. 
We denote such a code by $\gcode=(\mu,\Gamma,\{\code_{\gamma}\}_{\gamma\in\Gamma})$.
\end{definition}

\subsection{Input and  State Constraints} 
\label{subsec:Pconstraints}
Next, we consider input constraints and state constraint, imposed on the encoder and the jammer, respectively.
We note that the constraints specifications are known to both the user and the jammer in this model.
Let $\cost:\Xset\rightarrow [0,\infty)$, $k=1,2$, and $l:\Sset\rightarrow [0,\infty)$ be some given bounded functions, and define
	\begin{align}
	\cost^n(x^n)=&\frac{1}{n} \sum_{i=1}^n \cost(x_{i}) \,,\; 
	\label{eq:PLInConstraintStrict} \\
	l^n(s^n)=&\frac{1}{n} \sum_{i=1}^n l(s_i) \,.
	\end{align}
Let $\plimit>0$ and $\Lambda>0$. Below, we specify the input constraint $\plimit$ and state constraint $\Lambda$, corresponding to  the functions
$\cost^n(x^n)$  and $l^n(s^n)$, respectively. It is assumed that for some $a\in\Xset$ and $b\in\Sset$, $\cost(a)=l(b)=0$.	

As the parameter sequence $\theta^{\infty}\equiv (\theta_i)_{i=1}^{\infty}$ is fixed and known to the encoder, the decoder and the jammer, the input and state constraints below are specified for a particular sequence.
	%
	Given an input constraint $\plimit$, the encoding function needs to satisfy 
	\begin{align}
	\cost^n(\enc(m,\theta^n))\leq\plimit \,,\;\text{for all $m\in [1:2^{nR}]$} \,.
	\label{eq:PinputCstrict}
	\end{align}
	That is, the input sequence satisfies $\cost^n(X^n)\leq\plimit$  with probability $1$. 
	
	Moving to the state constraint $\Lambda$, we have different definitions for the AVC and for the compound channel.
	The compound channel has a constraint on average, where the state sequence satisfies
	$ \E_q l^n(S^n)\leq\Lambda$, 
	while the AVC has an almost-surely constraint, 
	$l^n(S^n)\leq\Lambda$ with probability (w.p.) $1$. 
	Explicitly, we say that a compound channel is under a state constraint $\Lambda$ if $\Qset\subseteq\pLSpaceS$, where
	\begin{align}
	\pLSpaceS&\triangleq \bigcap_{n=1}^{\infty} \left\{ q(s|t) \,:\;
\frac{1}{n} \sum_{i=1}^n \sum_{s\in\Sset} q(s|\theta_i) l(s)\leq \Lambda \right\} \,.
\label{eq:StateCcompound}
\intertext{ 	
%
As for the AVC $\avc$, 
 it is now assumed that the joint distribution of the state sequence is limited to $q(s^n|\theta^n)\in\pLSpaceSn$, where
}
\pLSpaceSn &\triangleq\{ q(s^n|\theta^n) \in\pSpace(\Sset^n|\Tset^n) \,:\; q(s^n|\theta^n)=0 \;\text{ if $l^n(s^n)>\Lambda$}\, \} \,.
\end{align}
This includes the case of a deterministic unknown state sequence, \ie when $q$ gives probablity $1$ to a particular $s^n\in\Sset^n$ with
$l^n(s^n)\leq \Lambda$.


\subsection{Capacity Under Constraints}
We move to the definition of an  achievable rate and the capacity of the AVC $\avc$ with fixed parameters  under input and state constraints.
Codes 
 over the AVC $\avc$ with fixed parameters are defined as in Definition~\ref{def:Pcapacity}, 
with the additional constraint 
 (\ref{eq:PinputCstrict}) on the codebook.

 Define the conditional probability of error of a code $\code$ given a state sequence $s^n\in\Sset^n$ by  
\begin{subequations}
\begin{align}
\label{eq:Pcerr}
&\err(\code|s^n,\theta^n)\triangleq 
\frac{1}{2^{ nR }}\sum_{m=1}^{2^ {nR}}
\sum_{y^n:\dec(y^n,\theta^n)\neq m} W_{Y^n|X^n,S^n,T^n}(y^n|\enc(m,\theta^n),s^n,\theta^n) \,.
\end{align}
Now, define the average probability of error of $\code$ for some distribution $\qn(s^n|\theta^n)\in\pSpace(\Sset^n)$, 
\begin{align}
\err(\qn,\theta^n,\code)\triangleq 
\sum_{s^n\in\Sset^n} \qn(s^n|\theta^n)\err(\code|s^n,\theta^n) \,.
\end{align}
\end{subequations}

\begin{definition}[
Achievable rate  and capacity  under constraints]
\label{def:PLcapacity}
A code $\code=(\enc,\dec)$ is a  called a
$(2^{nR},n,\eps)$ code for the AVC $\avc$ with fixed parameters under input constraint $\plimit$ and  state constraint $\Lambda$, when (\ref{eq:PinputCstrict})  is satisfied 
 and 
\begin{align}
\label{eq:PLerr}
& \err(q,\theta^n,\code) 
\leq \eps \,,\quad
\text{for all $q\in\pLSpaceSn$} \,,
\end{align}
or, equivalently, $\err(\code|s^n,\theta^n)\leq\eps$ for all $s^n\in\Sset^n$ with $l^n(s^n)\leq\Lambda$.

  We say that a rate  $R\geq 0$ is achievable under	constraints
	if for every $\eps>0$ and sufficiently large $n$, there exists a  $(2^{nR},n,\eps)$ code for the AVC	$\avc$ with fixed parameters under input constraint $\plimit$ and state constraint $\Lambda$. The operational capacity is defined as the supremum of achievable rates, 
	and it is denoted by $\Cavc$. 
 We use the term `capacity' referring to this operational meaning, and in some places we call it the deterministic code capacity in order to emphasize that achievability is measured with respect to  deterministic codes.

Analogously to the deterministic case,   a $(2^{nR},n,\eps)$ random  code  $\code^{\Gamma}$ 
satisfies the requirements
\begin{subequations}
\label{eq:PLrcodeReq}
\begin{align}
&
\sum_{\gamma}\mu(\gamma)   \cost^n(\enc(m,\theta^n))
  \leq \plimit
\,,\; 
\text{for all $m\in [1:2^{nR}]$}  \,,  \label{eq:PcodeInputCr}
\intertext{and} 
&\err(q,\code^{\Gamma})\triangleq \sum_{\gamma\in\Gamma} \mu(\gamma) \err(q,\theta^n,\code_{\gamma})
\leq \eps \,,\quad\text{for all $q\in\pLSpaceSn$} \,.
\label{eq:vPLrerr}				
\end{align}
\end{subequations}
The capacity region achieved by random codes is then denoted by $\rCav$, and it 
 is referred to as the \emph{random code capacity}.
\end{definition}
 


The definitions above are naturally extended to the compound channel with fixed parameters, under input constraints 
$\plimit$ and state constraint $\Lambda$, by limiting the requirements (\ref{eq:PinputCstrict}), (\ref{eq:PLerr}) and (\ref{eq:PLrcodeReq}) to conditionally memoryless state distributions $q\in\Qset$. 
The respective deterministic code capacity  $\Ccompound$ and random code capacity  $\rCcompound$  are defined accordingly.

\section{Main Results -- Channels with Fixed Parameters} 
\label{sec:Pmain}
In this section, we establish the random code capacity of the AVC with fixed parameters.
To this end, we first give an auxiliary result on the compound channel.

\subsection{The Compound Channel with Fixed Parameters}
We begin with the capacity theorem for the compound channel $\compound=\{ W_{Y|X,S,T},\Qset,\theta^{\infty} \}$. 
This is an auxiliary result, obtained by a simple extension of \cite[Exercise 6.8]{CsiszarKorner:82b}. A similar result appears in 
\cite{LapidothTelatar:98p} as well.
Given a parameter squence $\theta^n$ of a fixed length, 
define
\begin{align} 
\inC_n(\compound)= 
\max_{  p(x|t) \,:\;   \E \cost(X)\leq \plimit   } \; \inf_{  q(s|t)\in\Qset } 
 I_q(X;Y|T) \,,
\end{align}
with $(T,S,X)\sim P_{T}(t) p(x|t) q(s|t)$, where $P_{T}$ is the type of the parameter sequence $\theta^n$.
%
%
\begin{lemma}
\label{lemm:PCcompound}
The capacity  of the compound channel $\compound$ with fixed parameters, under input constraint $\plimit$ and state constraint $\Lambda$, is given by
\begin{align}
\Ccompound=\liminf_{n\rightarrow\infty} \inC_n(\compound) \,,
\end{align}
and it is identical to the random code capacity, \ie $\rCcompound=\Ccompound$. 
\end{lemma}
The proof of Lemma~\ref{lemm:PCcompound} is given in Appendix~\ref{app:PCcompound}.

\subsection{The AVC with Fixed Parameters -- Random Code Capacity}
We determine the random code capacity  of the AVC with fixed parameters, $\avc=\{W_{Y|X,S,T},\theta^{\infty}\}$,  under input constraint $\plimit$ and state constraint $\Lambda$.
The random code derivation is based on our result on the compound channel with fixed parameters and a variation of Ahlswede's 
Robustification Technique (RT). 
%
 Define
\begin{align}
\inC_{n}^{\rstarC}\hspace{-0.05cm}(\avc)\triangleq& \inC_{n}(\compound) \Big|_{\Qset=\pLSpaceS} \,.
\label{eq:Cieqiv3}
\end{align}

We begin with a lemma, based on 
 Ahlswede's RT \cite{Ahlswede:86p} 
(see also \cite[Lemma 9]{PeregSteinberg:19p1}). 
We modify it here to include the parameter sequence $\theta^n$ and  the constraint on the family of conditional state distributions 
$q(s|t)$.
\begin{lemma}[Modified RT] 
\label{lemm:LRT}
Let $h:\Sset^n\times\Tset^n \rightarrow [0,1]$ be a given function. If, for some fixed $\alpha_n\in(0,1)$, and for all 
$ q^n(s^n|\theta^n)=\prod_{i=1}^n q(s_i|\theta_i)$, with 
$q\in\apLSpaceS$, 
\begin{align}
\label{eq:RTcondCs}
\sum_{s^n\in\Sset^n} q^n(s^n|\theta^n)h(s^n,\theta^n)\leq \alpha_n \,,
\end{align}
then,
\begin{align}
\label{eq:RTresCs}
\frac{1}{|\Pi(\theta^n)|} \sum_{\pi\in\Pi(\theta^n)} h(\pi s^n,\theta^n)\leq \beta_n \,,\quad\text{for all $s^n\in\Sset^n$ such that $l^n(s^n)\leq\Lambda$} \,,
\end{align}
where $\Pi(\theta^n)$ is the set of all $n$-tuple permutations $\pi:\Sset^n\rightarrow\Sset^n$ such that 
$\pi \theta^n=\theta^n$, and 
$\beta_n=(n+1)^{|\Sset||\Tset|}\alpha_n$. 
\end{lemma}
Originally, Ahlswede's RT is stated so that (\ref{eq:RTcondCs}) holds for any $q(s)\in\pSpace(\Sset)$, without state constraint (see \cite{Ahlswede:86p}), and without conditioning on the parameter sequence $\theta^n$. We give the proof of 
Lemma~\ref{lemm:LRT} in Appendix~\ref{app:LRT}.
Next, we give our random code capacity theorem.
\begin{theorem}
\label{theo:PrCav}
The random code capacity of the AVC $\avc$ with fixed parameters, under input constraint $\plimit$ and state constraint $\Lambda$, is given by 
\begin{align}
\rCav=\liminf_{n\rightarrow\infty} \inC_{n}^{\rstarC}\hspace{-0.05cm}(\avc)  \,.
\end{align}
\end{theorem}
The proof of Theorem~\ref{theo:PrCav} is given in Appendix~\ref{app:PrCav}. The proof is based on our extension of 
Ahlswede's RT above.
Essentially, we use a reliable code for the compound channel to construct a random code for the AVC by  applying random permutations to the codeword symbols. However, here,  we only use permutations that do not affect the parameter sequence $\theta^n$. 
The result above plays a central role in the proof of the capacity theorem in Section~\ref{sec:GaussCol}, where the AVC with colored Gaussian noise is considered.

We also give an equivalent formulation in terms of the random code capacity of the traditional AVC.
As mentioned in Remark~\ref{rem:constT}, the case of an AVC $\{ W_{Y|X,S,T=t} \}$ with a constant parameter $\theta_i=t$ reduces to the traditional AVC under input and state constraints. For this channel, Csisz\'ar and Narayan \cite{CsiszarNarayan:88p1} showed that the random code capacity is given by 
\begin{align}
\inC_t^{\rstarC}\hspace{-0.05cm}(\plimit,\Lambda) \triangleq \min_{q(s)\,:\; \E l(S)\leq \Lambda} \max_{p(x)\,:\; \E \cost(X)\leq \plimit} I_q(X;Y|T=t) 
=\max_{p(x)\,:\; \E \cost(X)\leq \plimit} \min_{q(s)\,:\; \E l(S)\leq \Lambda}  I_q(X;Y|T=t)  
\label{eq:Ctol}
\end{align}
where the last equality is due to the minimax theorem \cite{sion:58p}.
Then, define 
\begin{align}
\inR_{n}^{\rstarC}\hspace{-0.05cm}(\avc)\triangleq& \min_{ \substack{ \lambda_1,\ldots,\lambda_n \,:\; \\ \frac{1}{n} \sum_{i=1}^n \lambda_i \leq \Lambda } } \;
\max_{ \substack{ \omega_1,\ldots,\omega_n \,:\;\\  \frac{1}{n} \sum_{i=1}^n \omega_i\leq \plimit } }
\frac{1}{n} \sum_{i=1}^n  
\inC_{\theta_i}^{\rstarC}\hspace{-0.05cm}(\omega_i,\lambda_i) \,,
\label{eq:Cieqiv2}
\end{align}
\begin{lemma}
\label{lemm:Ciequiv}
\begin{align}
\inR_{n}^{\rstarC}\hspace{-0.05cm}(\avc)=\inC_{n}^{\rstarC}\hspace{-0.05cm}(\avc) \,.
\end{align}
\end{lemma}
The proof of Lemma~\ref{lemm:Ciequiv} is given in Appendix~\ref{app:Ciequiv}. Theorem~\ref{theo:PrCav} and Lemma~\ref{lemm:Ciequiv} yield the following consequence.
\begin{coro}
\label{coro:PrCavE}
The random code capacity of the AVC $\avc$ with fixed parameters, under input constraint $\plimit$ and state constraint $\Lambda$, is given by 
\begin{align}
\rCav=\liminf_{n\rightarrow\infty} \inR_{n}^{\rstarC}\hspace{-0.05cm}(\avc)  \,.
\end{align}
\end{coro}
The corollary will also be useful in our analysis of the AVC with colored Gaussian noise.

\subsection{The AVC with Fixed Parameters -- Deterministic Code Capacity}
We move to  the deterministic code capacity  of the AVC with fixed parameters, $\avc=\{W_{Y|X,S,T},\theta^{\infty}\}$,  under input constraint $\plimit$ and state constraint $\Lambda$.
%



\subsubsection{Capacity Theorem}
Before we state the capacity theorem, we give  a few definitions.
We begin with  symmetrizability of a channel \emph{without} parameters.
\begin{definition}[see {\cite{CsiszarNarayan:88p}}]
\label{def:symmetrizable}
 A state-dependent DMC $V_{Y|X,S}$ is said to be \emph{symmetrizable} if for some conditional distribution $J(s|x)$,
\begin{align}
\label{eq:symmetrizable}
\sum_{s\in\Sset} V_{Y|X,S}(y|x_1,s)J(s|x_2)=\sum_{s\in\Sset} V_{Y|X,S}&(y|x_2,s)J(s|x_1) \,,\; \nonumber\\
&\forall\, x_1,x_2\in\Xset \,,\; y\in\Yset \,.
\end{align}
Equivalently, the channel $\widetilde{V}(y|x_1,x_2)$ $=$ $
\sum_{s\in\Sset} V_{Y|X,S}(y|x_1,s)J(s|x_2)$ is symmetric, \ie $\widetilde{V}(y|x_1,x_2)=\widetilde{V}(y|x_2,x_1)$, for all $x_1,x_2\in\Xset$ and $y\in\Yset$. We say that such a  $J:\Xset\rightarrow\Sset$ symmetrizes $V_{Y|X,S}$. 
\end{definition}
Intuitively, symmetrizability identifies a poor channel, where 
the jammer can impinge the communication scheme by randomizing the state sequence $S^n$ according to $J^n(s^n|x_2^n)=\prod_{i=1}^n J(s_i|x_{2,i})$, for some codeword $x_2^n$. 
 Suppose that the transmitted codeword is $x_1^n$. The codeword $x_2^n$ can be thought of as an impostor sent by the jammer.  Now, since  the ``average channel" $\widetilde{V}$ is symmetric with respect to $x_1^n$ and  $x_2^n$, the two codewords appear to the receiver as equally likely. Indeed, by \cite{Ericson:85p}, if the AVC $\{ V_{Y|X,S} \}$ without parameters and free of constraints is symmetrizable, then its capacity is zero. 

We will assume that either the channels $\channel(\cdot|\cdot,\cdot,\theta_i)$ are all symmetrizable, or the number of
non-symmetrizable channels grows linearly with $n$. That is, 
\begin{subequations}
\label{eq:Symmassumption}
\begin{align}
\text{either } \; |\Iset(n)|=0 \,\; \text{or }\;  |\Iset(n)|=\mathbf{\Omega}(n) \,,
\end{align}
where
\begin{align}
\Iset(n)=\left\{ i\in [1:n]  \,:\;  \channel(\cdot|\cdot,\cdot,\theta_i) \;\text{is non-symmetrizable} \right\} \,.
\end{align}
\end{subequations}
The asymptotic notation $f(n)=\mathbf{\Omega}(n)$ means that there exist $n_0>0$ and $0<\alpha\leq 1$ such that $f(n)\geq \alpha n$ for all $n\geq n_0$.
An intuitive explanantion for this assumption is given in Remark~\ref{rem:symmSuff} below.
Next, we define a symmetrizability cost and threshold for the AVC with fixed parameters. For every $n$ and $p(x|t)$ with 
\begin{align}
\frac{1}{n} \sum_{i=1}^n p(x|\theta_i) \cost(x)\leq\plimit \,,
\end{align}
 define the \emph{minimal symmetrizability cost} by
\begin{align}
\tLambda_n(p)\triangleq\min\, \frac{1}{n}\sum_{i=1}^n \sum_{x\in\Xset}\sum_{s\in\Sset} p(x|\theta_i)J_{\theta_i}(s|x)l(s)
= \min\,  \sum_{t\in\Tset} \sum_{x\in\Xset}\sum_{s\in\Sset} P_T(t) p(x|tJ_{t}(s|x)l(s)
 \,,
\label{eq:LambdaOig}
\end{align}
where the minimization is over the conditional distributions $J_t(s|x)$ that symmetrize 
$W_{Y|X,S,T}(\cdot|\cdot,\cdot,t)$, for $t\in \Tset$
(see Definition~\ref{def:symmetrizable}). We use the convention that a minimum value over an empty set is $+\infty$. 
Note that the last equality in (\ref{eq:LambdaOig}) holds since $ P_T$ is defined as the type of the parameter sequence $\theta^n$, hence averaging over time is the same as averaging according to $P_T$.
In addition, define the \emph{symmetrizability threshold}
\begin{align}
L_n^*\triangleq  \max_{p(x|t)\,:\; \frac{1}{n} \sum_{i=1}^n p(x|\theta_i) \cost(x)\leq\plimit} \tLambda_n(p) \,.
\label{eq:1Lstar}
\end{align}
Intuitively, $\tLambda_n(p)$ is the minimal average state cost which the jammer has to pay to symmetrize the channel at each time instance, for a given conditional input distribution $p(x|t)$. If this minimal state cost violates the state constraint $\Lambda$, then the jammer is prohibited from symmetrizing the channel. Indeed, we will show that 
if there exists an input distribution $p(x|t)$ with $\frac{1}{n} \sum_{i=1}^n p(x|\theta_i) \cost(x)\leq\plimit$ and 
$\tLambda_n(p)>\Lambda$ for large $n$, then the deterministic code capacity is positive.
The symmetrizability threshold $L_n^*$ is the worst symmetrizability cost from the jammer's perspective.

Our capacity result is stated below.
Let
\begin{align}
\inC_n(\avc)& \triangleq 
\begin{cases}
\min\limits_{  q(s|t) \,:\; \E_q l(S)\leq \Lambda } \;
\max\limits_{ \substack{ p(x|t)\,:\; \E\,\cost(X)\leq\plimit \,,\; \\ \tLambda_n(p)\geq \Lambda } } \;   I_q(X;Y|T)  &\text{if $L_n^*> \Lambda$}\,,\\
  0 																 &\text{if $L_n^*\leq \Lambda$}
\end{cases}	
\,,
\label{eq:Cieqiv3Det} 
\end{align}
with $(T,S,X)\sim P_{T}(t) p(x|t) q(s|t)$, where $P_{T}$ is the type of the parameter sequence $\theta^n$ with a fixed length $n$.

\begin{theorem}
\label{theo:PCavc}
Assume that $L_n^*\neq \Lambda$ for sufficiently large $n$ and that (\ref{eq:Symmassumption}) holds.
The capacity of an AVC $\avc$ with fixed parameters, under input constraint $\plimit$ and state constraint $\Lambda$, is given by
\begin{align}
\Cavc=\liminf\limits_{n\rightarrow\infty} \inC_n(\avc) \,.
\end{align}
In particular, if the channels $W_{Y|X,S,T}(\cdot|\cdot,\cdot,t)$, $t\in\Tset$, are non-symmetrizable, then
$
\Cavc=\rCav=\,$ $\liminf\limits_{n\rightarrow\infty} \inC_n^{\rstarC}\hspace{-0.05cm}(\avc) 
$. That is, the deterministic code capacity coincides with the random code capacity. 
\end{theorem}
The proof of Theorem~\ref{theo:PCavc} is given in Appendix~\ref{app:PCavc}. The theorem will also play a central role in the proof of the capacity theorem in Section~\ref{sec:GaussCol}.

\begin{remark}
\label{rem:symmSuff}
Observe that the second part of the theorem implies that for the case where there are no constraints,
\ie $\plimit=\cost_{max}$ and $\Lambda=l_{max}$, non-symmetrizability is a 
sufficient condition for positive capacity. 
%
Specfically, according to the definition of $\tLambda_n(p)$, $L_n^*$ in (\ref{eq:LambdaOig})-(\ref{eq:1Lstar}), if some of the channels $W_{Y|X,S,T}(\cdot|\cdot,\cdot,\theta_i)$ are non-symmetrizable, then the symmetrizability threshold is  $L_n^*=\infty$, hence the capacity is positive. Intuitively, if the number of such channels is constant, i.e. $|\Iset(n)|=c$ for all $n$, it seems that  this assignment of $L_n^*$ does not make sense, since the user cannot achieve positive rates by coding over a negligible fraction of the block.
Yet, our assumption in (\ref{eq:Symmassumption}) excludes this scenario. In particular, if $|\Iset(n)|$ is non-zero, then we assume that $|\Iset(n)|$ grows linealy in $n$, in which case positive rates can be achieved by coding over the part of the block that lies within $\Iset(n)$. 
Furthermore, without constraints, we may replace the linear growth assumption with a poly-logarithmic one, i.e. $|\Iset(n)|=\mathbf{\Omega}((\log n)^a)$, with $a>1$. Indeed,
 based on Ahlswede's elimination technique \cite{Ahlswede:78p}, the random code capacity can be achieved with a code collection of polynomial size, $|\Gamma|=n^2$. Therefore, without state constraints, the random element $\gamma\in\Gamma$  can be reliably sent to the receiver 
 over the sub-block $\Iset(n)$, at rate $\rho_n=\frac{\log|\Gamma|}{(\log n)^a}=2(\log n)^{-(a-1)}$, which tends to zero as $n\rightarrow\infty$, hence the decrease in the overall rate is negligible as well.
We deduce that  if $|\Iset(n)|=\mathbf{\Omega}((\log n)^a)$, then the deterministic code capacity of the AVC with fixed parameters without constraints is the same as the random code capacity, i.e.
\begin{align}
\Cavc=\rCav=\liminf\limits_{n\rightarrow\infty} \min_{q(s|t)} \max_{p(x|t)} I_q(X;Y|T) \,.
\end{align}
\end{remark}

\begin{remark}
\label{rem:boundLstate}
Even in the case where there are no parameters, the boundary case where $L_n^*=\Lambda$ is an open problem. Although, for the traditional AVC, it is conjectured in \cite{CsiszarNarayan:88p} that the capacity is zero in this case. Similarly, we conjecture that the capacity of the AVC with fixed parameters is given by $\Cavc=\liminf\limits_{n\rightarrow\infty} \inC_n(\avc)$ for all values of $\{L_n^*\}_{n\geq 1}$, provided that
 (\ref{eq:Symmassumption}) holds.
There are special cases where we know that this holds, given in the corollary below. The corollary is based on the remark following Theorem 3 in \cite{CsiszarNarayan:88p}.
\end{remark}

\begin{coro}
\label{coro:LCavc01}
Let $\avc$ be an AVC with fixed parameters such that all channels $W_{Y|X,S,T}(\cdot|\cdot,\cdot,t)$, $t\in\Tset$, are symmetrizable.
If the minimum in (\ref{eq:LambdaOig}) is attained by a $0$-$1$ law, for every $n$ and $p(x|t)$ with 
$\frac{1}{n} \sum_{i=1}^n p(x|\theta_i) \cost(x)\leq\plimit$,
then
\begin{align}
\Cavc=\liminf\limits_{n\rightarrow\infty} \inC_n(\avc) \,.
\label{eq:LCavcTheoSp}
\end{align}
\end{coro}
The proof of Corollary~\ref{coro:LCavc01} is given in Appendix~\ref{app:LCavc01}.
In particular, we note that the condition of $0$-$1$ law in Corollary~\ref{coro:LCavc01} holds when the
output $Y$ is a deterministic function of $X$, $S$, and $T$. 
As opposed to Theorem~\ref{theo:PCavc}, the statement in Corollary~\ref{coro:LCavc01} holds for all values of $\{L_n^*\}_{n\geq 1}$.

\subsubsection{
Decoding Rule}
\label{subsec:Dec}
We specify the decoding rule and state the corresponding properties, which are used in the analysis. 
%
To specify the decoding rule, we define the decoding sets $\Dset(m)\subseteq \Yset^n\times\Tset^n$, for $m\in [1:2^{nR}]$, such that 
$g(y^n,\theta^n)=m$ iff $(y^n,\theta^n)\in\Dset(m)$.
\begin{definition}[Decoder]
\label{def:Ldecoder}
Given the codebook $\{ \enc(m,\theta^n) \}_{m\in [1:2^{nR}]}$,  declare that $(y^n,\theta^n)\in \Dset(m)$ if there exists $s^n\in\Sset^n$ with $l^n(s^n)\leq\Lambda$ such that the following hold.
\begin{enumerate}[1)]
\item
For $(T,X,S,Y)$ that is distributed according to the joint type  $\hP_{\theta^n,\encn(m,\theta^n),s^n,y^n}$, we have that 
\begin{align}
D(P_{T,X,S,Y}|| P_T\times P_{X|T}\times P_{S|T} \times W_{Y|X,S,T} )\leq \eta \,.
\end{align}
 
\item
For every $\tm\neq m$ such that for some $\ts^n\in\Sset^n$ with $l^n(\ts^n)\leq\Lambda$, 
\begin{align}
\label{eq:DcompA} 
D(P_{T,\tX,\tS,Y}|| P_T\times P_{\tX|T}\times P_{\tS|T} \times W_{Y|X,S,T})\leq \eta \,,
\end{align} 
where $(T,\tX,\tS,Y)\sim \hP_{\theta^n,\encn(\tm,\theta^n),\ts^n,y^n}$,
 we have that
\begin{align}
I(X,Y;\tX|S,T)\leq \eta  \,.
\end{align}

\end{enumerate}

\end{definition}
We note that in Definition~\ref{def:Ldecoder}, the variables $T,X, \tX,S,\tS,Y$ are dummy random variables, distributed according to the joint type of 
$(\theta^n,\encn(m,\theta^n),$ $\encn(\tm,\theta^n),$ $s^n,\ts^n,y^n)$, where
 $\encn(m,\theta^n)$ is a ``tested" codeword, 
$\encn(\tm,\theta^n)$ is a competing codeword, $s^n$ is a ``tested" state sequence, $\ts^n$ is a competing state sequence, and $y^n$ is the received sequence. None of the sequences are random here.
We may have that the conditional type $P_{Y|X,S,T}$ differs from the actual channel $ W_{Y|X,S,T}$. Therefore, the divergences and mutual informations in Definition~\ref{def:Ldecoder} could be positive.

For the definition above to be proper, the decoding sets need to be disjoint, as stated in the following lemma.
\begin{lemma}[Decoding Disambiguity]
\label{lemm:disDec}
Suppose that in each codebook, all codewords have the same  conditional type, \ie $\hP_{\enc(m,\theta^n)|\theta^n}= p$ for all $m\in [1:2^{nR}]$.
 Assume  (\ref{eq:Symmassumption}) holds, that for some $\delta_0,\delta_1>0$, $P_T(t)\geq \delta_0$,  $p(x|t)\geq\delta_1$, $\forall x\in\Xset$, $t\in\Tset$,  and also
\begin{align}
  \tLambda_n(p)   >\Lambda \,.
\label{eq:decLambda}
\end{align}
 Then, for sufficiently small $\eta>0$, 
\begin{align}
\Dset(m)\cap \Dset(\tm) =\emptyset \,,\;\text{for all $m\neq \tm$} \,.
\label{eq:disDec}
\end{align}
\end{lemma}
The proof of Lemma~\ref{lemm:disDec} is given in Appendix~\ref{app:disDec}.

\subsubsection{
Codebook Generation}
\label{subsec:codebooks}
We now extend Csisz\'{a}r and Narayan's lemma for the codebook generation \cite{CsiszarNarayan:88p}. 
\begin{lemma}[Codebooks Generation] 
\label{lemm:codeBsets}
For every $\eps>0$, sufficiently large $n$, rate $R\geq \eps$ and conditional type $p(x|t)$,  there exist a set of codewords 
$\{x^n(m,\theta^n)\}_{m\in [1:2^{nR}]}$ of conditional type $p$, 
such that for every $a^n\in\Xset^n$ and $s^n\in\Sset^n$ with $l^n(s^n)\leq\Lambda$, and every joint type $P_{T,X,\tX,S}$ with 
$P_{X|T}=P_{\tX|T}=p$, the following hold.
\begin{align}
|\{ \tm \,:\; (\theta^n,a^n,x^n(\tm,\theta^n),s^n)\in\Tset^n(P_{T,X,\tX,S})  \}|
 \leq
2^{n\left( \left[ R-I(\tX;X,S|T) \right]_{+} +\eps \right)} \,, 
\label{eq:11ebn}
\end{align}
\begin{align}
|\{ m \,:\; (\theta^n,x^n(m,\theta^n),s^n)\in\Tset^n(P_{T,X,S})  \}|
 \leq
2^{n\left(  R-\frac{\eps}{2}   \right)}
\,,\;\text{if $I(X;S|T)>\eps$} \,,
\label{eq:12ebn}
\end{align}
and
\begin{multline}
|\{ m \,:\; (\theta^n,x^n(m,\theta^n),x^n(\tm,\theta^n),s^n)\in\Tset^n(P_{T,X,\tX,S})   \,,\; \text{for some $\tm\neq m$}
\}| \\
\leq
2^{n\left(  R-\frac{\eps}{2}   \right)}
\,,\;
\text{if $I(X;\tX,S|T)-\left[ R-I(\tX;S|T) \right]_{+}>\eps$} \,.
\label{eq:13ebn}
\end{multline}

\end{lemma}
The proof of Lemma~\ref{lemm:codeBsets} is given in Appendix~\ref{app:codeBsets}.

\subsection{Super-Additivity}
We also give an equivalent formulation with a sum over $i\in [1:n]$. Here, as opposed to the previous section, the formula
 \emph{cannot} be expressed in terms of the capacities of the constant-parameter AVCs $\{ W_{Y|X,S,T=\theta_i} \}$.
Considering the AVC without constraints, Schaefer \etal \cite{SchaeferBochePoor:16c} showed that the capacity of any product AVC that is composed of a symmetrizable channel and a non-symmetrizable channel is larger than the sum of the individual capacities (see Theorem 6 in \cite{SchaeferBochePoor:16c}).
Similarly, we give an example at the end of this section where the capacity of the AVC \emph{with fixed parameters} is larger than the weighted sum of the capacities of the constant-parameter AVCs $\{ W_{Y|X,S,T=\theta_i} \}$. 
This phenomenon can be viewed as an instance of the  
super-additivity property in \cite{SchaeferBochePoor:16c}. 

We begin with constant-parameter definitions, \ie for a fixed $T=t$.
For every input distribution $p(x)$ with $\E\cost(X)\leq \plimit$, define the constant-parameter minimal symmetrizability cost by
\begin{align}
\tLambda(p,t)\triangleq\min\, \sum_{x\in\Xset}\sum_{s\in\Sset} p(x)J(s|x)l(s) \,,
\label{eq:LambdaOig1}
\end{align}
where the minimization is over the distributions $J(s|x)$ that symmetrize 
$W_{Y|X,S,T}(\cdot|\cdot,\cdot,t)$, where $t\in\Tset$ is fixed 
(see Definition~\ref{def:symmetrizable}). Then, we can write the minimal symmetrizability cost defined in (\ref{eq:LambdaOig}) as 
\begin{align}
\tLambda_n( p(\cdot|\cdot) )=\frac{1}{n} \sum_{i=1}^n \tLambda(p(\cdot|\theta_i),\theta_i) \,.
\label{eq:LambdaOigEq}
\end{align}
%
%
Let
\begin{align}
\inR_n(\avc)& \triangleq 
\begin{cases}
\min\limits_{ \substack{ \lambda_1,\ldots,\lambda_n \,:\; \\ \frac{1}{n} \sum_{i=1}^n \lambda_i \leq \Lambda } } \;
\max\limits_{ \substack{ \omega_1,\ldots,\omega_n, \tlambda_1,\ldots\tlambda_n \,:\;\\  
\frac{1}{n} \sum_{i=1}^n \omega_i\leq \plimit \,,
\frac{1}{n} \sum_{i=1}^n \tlambda_i\geq \Lambda } }
\frac{1}{n} \sum\limits_{i=1}^n
\inC_{\theta_i}(\omega_i,\tlambda_i,\lambda_i)  &\text{if $L_n^*> \Lambda$}\,,\\
  0 																 &\text{if $L_n^*\leq \Lambda$}
\end{cases}	
\,,
\label{eq:Cieqiv2Det}
\intertext{where }
\inC_t(\plimit,\Delta,\Lambda)\triangleq&
\min\limits_{  q(s) \,:\; \E_q l(S)\leq \Lambda } \;
\max\limits_{ \substack{ p(x)\,:\; \E\,\cost(X)\leq\plimit \,,\; \\ \tLambda(p,t)\geq \Delta } } \;   I_q(X;Y|T=t) 
\label{eq:CtolDet}
\end{align}
We note that based on Csisz{\'a}r and Narayan's result in \cite{CsiszarNarayan:88p}, the capacity of the constant-parameter AVC
$\{W_{Y|X,S,T=t}\}$ is given by $\inC_t(\plimit,\Delta,\Lambda)$ with $\Delta=\Lambda$. 
\begin{lemma}
\label{lemm:CiequivDet}
\begin{align}
\inR_{n}(\avc)=\inC_{n}(\avc) \,.
\end{align}
\end{lemma}
The proof of Lemma~\ref{lemm:CiequivDet} is given in Appendix~\ref{app:CiequivDet}. Theorem~\ref{theo:PCavc}, Corollary~\ref{coro:LCavc01}, and Lemma~\ref{lemm:CiequivDet} yield the following consequence.
\begin{coro}
\label{coro:PrCavEDet}
The deterministic code capacity of the AVC $\avc$ with fixed parameters, under input constraint $\plimit$ and state constraint $\Lambda$, is given by 
\begin{align}
\Cavc=\liminf_{n\rightarrow\infty} \inR_{n}(\avc)  \,,\; \text{if $L_n^*\neq\Lambda$ for sufficiently large $n$ and (\ref{eq:Symmassumption}) holds.} \,.
\end{align}
Furthermore, if the minimum in (\ref{eq:LambdaOig1}) is attained by a $0$-$1$ law, for every $p(x)$ with 
$\E \cost(X)\leq\plimit$, and for all $t\in\Tset$,
then
\begin{align}
\Cavc=\liminf_{n\rightarrow\infty} \inR_{n}(\avc)  \,,\;
\end{align}
for all values of $\{L_n \}_{n\geq 1}$.
\end{coro}
The corollary will also be useful in our analysis of the AVC with colored Gaussian noise.

\begin{example}
Consider the arbitrarily varying binary symmetric channel (BSC) with fixed parameters, 
\begin{align}
Y=X+S+Z_T \;\mod 2
\end{align}
with $\Xset=\Sset=\Tset=\{0,1\}$, where $Z_t\sim\text{Bernoulli}(\eps_t)$, for $t=0,1$, $\eps_0<\eps_1<\frac{1}{2}$.
Consider a parameter sequence with an empirical distribution $P_T(0)=P_T(1)=\frac{1}{2}$, say
 $\theta_{2i}=0$ and $\theta_{2i-1}=1$ for $i=1,2,\ldots$. Suppose that the user and the jammer are subject to input constraint $\plimit$ and state constraint
$\Lambda$, respectively, 
 with Hamming weight cost functions, \ie $\cost(x)=x$ and $l(s)=s$.

For the constant-parameter AVC, we have by Definition~\ref{def:symmetrizable} that $W_{Y|X,S,T=t}$ is symmetrized by any symmetric distribution, \ie
with $J(s|1)=1-J(s|0)$. Denoting $\zeta=J(1|1)=1-J(1|0)$, we have that
\begin{align}
\tLambda(P_X,t)=& \min_{0\leq \zeta\leq 1}  [(1-\zeta)P_X(0)+\zeta P_X(1)] =\min(  P_X(0),P_X(1) ) \,.
\end{align}
Based on the analysis by Csisz\'ar and Narayan \cite[Example 1]{CsiszarNarayan:88p}, the capacity of the constant-parameter AVC under input constraint $\omega$ and state constraint $\lambda$ is given by
\begin{align}
\widetilde{\opC}_t(\omega,\lambda)=\begin{cases}
0 &\text{if $\omega<\lambda < \frac{1}{2}$}\\
h(\omega*\lambda*\eps_t)-h(\omega*\lambda*\eps_t) &\text{if $\lambda<\omega < \frac{1}{2}$} \\
1-h(\omega*\lambda*\eps_t) &\text{if $\lambda< \frac{1}{2}\leq \omega$} \\
0 &\text{if $\lambda\geq \frac{1}{2}$}
\end{cases}
\end{align}
where $h(x)=-x\log x-(1-x)\log x$ is the binary entropy function and $a*b=(1-a)b+a(1-b)$.

Suppose that
\begin{align}
\eps_0=\frac{1}{4} \,,\; \eps_1=\frac{5}{12} \,,\; \plimit=\frac{5}{16}  \,,\; \Lambda=\frac{1}{4} \,.
\end{align}
For those values, we have that
\begin{align}
L_n^*=\max_{P_{X|T} \,:\; \frac{1}{2}\E(X|T=0)+\frac{1}{2}\E(X|T=1)\leq\plimit} 
\left[ \frac{1}{2}P_{X|T}(1|0)+ \frac{1}{2}P_{X|T}(1|1) \right]=\plimit=\frac{5}{16} \,.
\end{align}
Thus, by Corollary~\ref{coro:PrCavEDet}, the capacity is given by
\begin{align}
\Cavc= h( \frac{5}{16}*\frac{7}{16} )-h(\frac{7}{16}) 
= \frac{1}{2} \left( h(\omega_0*\lambda_0*\eps_0)-h(\lambda_0*\eps_0)  \right)+
\frac{1}{2} \left( h(\omega_1*\lambda_1*\eps_1)-h(\lambda_1*\eps_1)  \right) 
\end{align}
with $\omega_0=\omega_1=\frac{5}{16}$, $\lambda_0=\frac{3}{8}$ and $\lambda_1=\frac{1}{8}$. 
Whereas, using two separate codes for $W_{Y|X,S,T=0}$ and $W_{Y|X,S,T=1}$ independently, the rate achieved is 
\begin{align}
\frac{1}{2}\widetilde{\opC}_0(\omega_0,\lambda_0)+\frac{1}{2}\widetilde{\opC}_1(\omega_1,\lambda_1)
=0+\frac{1}{2} \left( h(\omega_1*\lambda_1*\eps_1)-h(\lambda_1*\eps_1)  \right)< \Cavc \,.
\end{align}
This can be viewed as an instance of the more general phenomenon of 
super-additivity, that holds for any product AVC which is composed of a symmetrizable AVC and a non-symmetrizable AVC 
\cite[Theorem 6]{SchaeferBochePoor:16c}.
\end{example}

\subsection{Example: Channel with Fadings}
To illustrate our results, we give another example.
\begin{example}
\label{example:Fading}
Consider an arbitrarily varying fading channel,
\begin{align}
Y_i=\theta_i X_i+S_i+Z_i \,,
\end{align}
with a Gaussian noise sequence $Z^n$ that is i.i.d. $\sim \mathcal{N}(0,\sigma^2)$, where
 $\theta_1,\theta_2,\ldots$ is a sequence of fixed fading coefficients.
Recently, Hosseinigoki and Kosut \cite{HosseinigokKosut:19c1} considered this channel with a random memoryless sequence of fading coefficients. Yet, we assume that the fading coefficients are fixed, and belong to a finite set $\Tset$.
Intuitively, the jammer would like to confuse the decoder by sending a state sequence that simulates the sequence
$\theta^n X^n\equiv (\theta_i X_i)_{i=1}^n$. Indeed, as seen below, the deterministic code capacity is positive only if there exists an input distribution such that $\frac{1}{n} \sum_{i=1}^n \theta_i^2 \E X_i^2>\Lambda$, in which case the jammer cannot simulate $\theta^n X^n$ without violating the state constraint.

Although we previously assumed that the alphabets are finite, our results can be extended to the continuous case as well, using standard discretization techniques \cite{BBT:59p,Ahlswede:78p} \cite[Section 3.4.1]{ElGamalKim:11b}. 
By Theorem~\ref{theo:PrCav}, the random code capacity is given by 
\begin{align}
\rCav=\liminf_{n\rightarrow\infty} \inC_{n}^{\rstarC}\hspace{-0.05cm}(\avc)
\,.
\end{align}
Then, we show that 
\begin{align}
\inC_{n}^{\rstarC}\hspace{-0.05cm}(\avc)=&  \min_{  \lambda(t) \,:\;  \E \lambda(T) \leq \Lambda }  \;
\max_{  \omega(t) \,:\;  \E \omega(T) \leq \plimit  }
\E \left[  \frac{1}{2}\log\left( 1+\frac{T^2 \omega(T) }{\lambda(T)+\sigma^2} \right) \right] \,,
\label{eq:Cieqiv2Fading}
\end{align}
with expectation over $T\sim P_T$,
where $P_T$ is the type of the sequence $\theta^n$.

As for the deterministic code capacity, we show that the minimum in (\ref{eq:LambdaOig}) is attained by a $0$-$1$ law that gives probability $1$ to $s=\theta_i^2 x$, hence we can determine the capacity using Corollary~\ref{coro:LCavc01}.
We show that the minimal symmetrizability cost is given by
\begin{align}
\tLambda_n(F_{X|T})=\frac{1}{n} \sum_{i=1}^n \theta_i^2 \E[X^2|T=\theta_i]=\E(T^2 X^2) 
 \,,
\label{eq:tLambdaF}
\end{align}
and deduce that the capacity of the AVC with fixed fading coeffients is given by 
\begin{align}
\Cavc=\liminf_{n\rightarrow\infty} \inC_{n}(\avc)  \,,\; 
\end{align}
with
\begin{align}
\inC_n(\avc)& \triangleq 
\begin{cases}
\min\limits_{  \lambda(t) \,:\; \E \lambda(T)\leq \Lambda } \;
\max\limits_{ \substack{ \omega(t)\,:\; \E\,\omega(T)\leq\plimit \,,\; \\ \E(T^2 \omega(T)) \geq \Lambda } } \;  \E \left[  \frac{1}{2}\log\left( 1+\frac{T^2 \omega(T) }{\lambda(T)+\sigma^2} \right) \right]  &\text{if $\max\limits_{\omega(t) \,:\; \E \omega(T)\leq \plimit } \E(T^2 \omega(T))> \Lambda$}\,,\\
  0 																 &\text{if $\max\limits_{\omega(t) \,:\; \E \omega(T)\leq \plimit } \E(T^2 \omega(T))\leq \Lambda$}
\end{cases}	
\,.
\label{eq:Cieqiv3DetFading} 
\end{align}
The derivation is given in Appendix~\ref{app:Fading}.
We note that the last expression has the same form as the capacity formula established by Hosseinigoki and Kosut \cite{HosseinigokKosut:19c1} for a random memoryless sequence of fading coefficients.

Next, we extend the result above to continuous fading coefficients, where $\Tset=[-t_0,t_0]\subset\mathbb{R}$. 
First, we  observe that the formulas above can also be written as
\begin{align}
\inC_{n}^{\rstarC}\hspace{-0.05cm}(\avc)=&  \min_{ \substack{ \lambda_1,\ldots,\lambda_n \,:\; \\ \frac{1}{n} \sum_{i=1}^n \lambda_i \leq \Lambda } }  \;
\max_{ \substack{ \omega_1,\ldots,\omega_n \,:\; \\ \frac{1}{n} \sum_{i=1}^n \omega_i \leq \plimit } }
 \frac{1}{n} \sum_{i=1}^n \frac{1}{2}\log\left( 1+\frac{\theta_i^2 \omega_i }{\lambda_i+\sigma^2} \right)  \,,
\label{eq:Cieqiv2FadingEq}
\end{align}
and
\begin{align}
\inC_n(\avc)& =
\begin{cases}
\min\limits_{ \substack{ \lambda_1,\ldots,\lambda_n \,:\; \\ \frac{1}{n} \sum_{i=1}^n \lambda_i \leq \Lambda } }  \;
\max\limits_{ \substack{ \omega_1,\ldots,\omega_n \,:\; \\\frac{1}{n} \sum_{i=1}^n \omega_i \leq\plimit \,,\; \\ 
\frac{1}{n} \sum_{i=1}^n \theta_i^2 \omega_i \geq \Lambda } } \; 
 \frac{1}{n} \sum\limits_{i=1}^n   \frac{1}{2}\log\left( 1+\frac{\theta_i^2 \omega_i }{\lambda_i+\sigma^2} \right)   &\text{if $\max\limits_{ \substack{ \omega_1,\ldots,\omega_n \,:\;\\ \frac{1}{n} \sum_{i=1}^n \omega_i\leq \plimit }} 
\frac{1}{n} \sum\limits_{i=1}^n \theta_i^2 \omega_i > \Lambda$}\,,\\
  0 																 &\text{otherwise.}
\end{cases}	
\label{eq:Cieqiv3DetFadingEq} 
\end{align}
This follows from the same considerations as in the proofs of Lemma~\ref{lemm:Ciequiv} and Lemma~\ref{lemm:CiequivDet}.
Now, if the fading coefficients are continuous, then one may perform the discretization procedure in \cite[Section 3.4.1]{ElGamalKim:11b}. Hence, the  deterministic and random code capacities in the continuous case are also given by the limit infimum of the formulas
 (\ref{eq:Cieqiv2FadingEq}) and (\ref{eq:Cieqiv3DetFadingEq}), respectively.

\end{example}

\section{The Arbitrarily Varying Gaussian Product Channel}
\label{sec:def}
From this point on, we consider Gaussian AVCs, without parameters. In this section, we consider the Gaussian product channel.
Our results on the AVC with colored Gaussian noise, in the next section, are based on the capacity theorems of the AVC with fixed parameters, in the previous section, and on the analysis in the current section.

	\subsection{Channel Description}
	\label{subsec:pGchannels}
	The state-dependent Gaussian product channel consists of a set of $d$ parallel channels,
\begin{align}
Y_j=  X_j+S_j+Z_j \,,\; j\in [1:d] \,,
\end{align}
 where $j$ is the channel index, $d$ is the dimension (number of channels), and 
$Z^d$ is a  Gaussian vector with zero mean and covariance matrix $K_Z$. 
Let $\Xvec_j=(X_{j,i})_{i=1}^n$, $\Svec_j=(S_{j,i})_{i=1}^n$ and $\Zvec_j=(Z_{j,i})_{i=1}^n$  denote the input, state and noise sequences associated with the $j$th channel, respectively,
where $i\in [1:n]$ is the time index, and let 
$\Xvec^d=(\Xvec_j)_{j=1}^d$, $\Svec^d=(\Svec_j)_{j=1}^d$  and $\Zvec^d=(\Zvec_j)_{j=1}^d$.  The corresponding  output of the product channel is the vector sequence $\Yvec^d=\Xvec^d+\Svec^d+\Zvec^d$.

The Gaussian arbitrarily varying product channel (AVGPC) is a state-dependent Gaussian product channel 
with $d$ state sequences $(\Svec_1,\ldots,\Svec_d)$ of unknown distribution,  not necessarily independent nor stationary. That is, $(\Svec_1,\ldots,\Svec_d)\sim  F_{\Svec_1,\ldots,\Svec_d}$, 
where $F_{\Svec_1,\ldots,\Svec_d}$ is an unknown  joint cumulative distribution function (cdf) over $\mathbb{R}^{nd}$. In particular, $F_{\Svec_1,\ldots,\Svec_d}$ could give probability mass $1$ to a particular  sequence of state vectors  $(\svec_1,\ldots,\svec_d)\in\mathbb{R}^{nd}$.
The channel is subject to input constraint $\plimit>0$ and state constraint $\Lambda>0$,
\begin{align}
&\sum_{j=1}^d\norm{\Xvec_j}^2 \leq n\plimit \quad\text{w.p. $1$} \,, \nonumber\\
&\sum_{j=1}^d\norm{\Svec_j}^2\leq n\Lambda \quad\text{w.p. $1$}
\,.
\end{align}

\subsection{Coding}
\label{subsec:pGcoding}
We introduce preliminary definitions for the AVGPC. 

\begin{definition}[Code] 
\label{def:pGcapacity}
A $(2^{nR},n)$ code for the AVGPC consists of the following;   
a message set $[1:2^{nR}]$, 
where it is assumed throughout that $2^{nR}$ is an integer,
a sequence of $d$ encoding functions 
$\fvec_j:  [1:2^{nR}] \rightarrow \mathbb{R}^{n}$, for $j\in [1:d]$,
such that
\begin{align}
\sum_{j=1}^d  \norm{\fvec_j(m)}^2 
\leq n\plimit \,,\;\text{for $m\in [1:2^{nR}]$} \,,
\end{align}
  and a decoding function
$
\dec: \mathbb{R}^{nd}\rightarrow [1:2^{nR}]  
$. 
Given a message $m\in [1:2^{nR}]$, the encoder transmits $\xvec_j=\fvec_j(m)$, for $j\in [1:d]$.
 The codeword is then given by 
$
\xvec^d= \fvec^d(m) \triangleq
 \left( \fvec_1(m),\fvec_2(m),\ldots,\fvec_d(m)   \right) 
$. 
The decoder  receives the channel outputs $\yvec^d=(\yvec_1,\ldots,\yvec_d)$, and finds an estimate of the message $\hm=g(\yvec^d)$.  
We denote the code by $\code=\left(\fvec^d,\dec \right)$.

 Define the conditional probability of error of a code $\code$ given the sequence 
 $\svec^d=(\svec_1,\ldots,\svec_d)$ by  
\begin{align}
\label{eq:Gpcerr}
&\Gcerr(\code)\triangleq 
\frac{1}{2^{ nR }}\sum_{m=1}^{2^ {nR}}
\int_{ \yvec^d\in\mathbb{R}^{nd} \,:\; g(\yvec^d)\neq m} 
d\yvec^d  \cdot f_{\Yvec^d|m,\svec^d}(\yvec^d) \,,
\end{align}
where $f_{\Yvec^d|m,\svec^d}(\yvec^d)=\prod_{i=1}^n f_{Z^d}(y^d_i-\mathrm{f}^d_i(m)-s^d_i)$, with
\begin{align}
&f_{Z^d}(z^d)= 
 \frac{1}{ \sqrt{ (2\pi)^{d}  |K_{Z}| } } 
e^{ -\frac{1}{2} z^d K_{Z}^{-1} (z^d)^T } \,.
\end{align}
%
A code $\code=(\fvec^d,\dec)$ is called a
$(2^{nR},n,\eps)$ code for the AVGPC  
if 
\begin{align}
\label{eq:Gperr}
 \Gcerr(\code) 
\leq \eps \,,
\quad\text{for all $\svec^d\in\mathbb{R}^{nd}$ with $\sum_{j=1}^d \norm{\svec_j}^2 \leq n\Lambda$} \,.
\end{align}

  We say that a rate $R$ is achievable 
	if for every $\eps>0$ and sufficiently large $n$, there exists a  $(2^{nR},n,\eps)$ code for the AVGPC.
The operational capacity is defined as the supremum of all achievable rates, and it is denoted by $\KCavc$. 
 We use the term `capacity' referring to this operational
meaning, and in some places we call it the deterministic code capacity to emphasize that achievability is measured with respect to  deterministic codes.   
\end{definition}

We proceed now to coding schemes 
when using stochastic-encoder stochastic-decoder pairs with common randomness.

\begin{definition}[Random code]
\label{def:GPcorrC} 
A $(2^{nR},n)$ random code for the AVGPC consists of a collection of 
$(2^{nR},n)$ codes $\{\code_{\gamma}=(\fvec_{\gamma}^d ,\dec_\gamma)\}_{\gamma\in\Gamma}$, along with a pmf $\mu(\gamma)$ over the code collection $\Gamma$. 
We denote such a code by $\gcode=(\mu,\Gamma,\{\code_{\gamma}\}_{\gamma\in\Gamma})$.
Analogously to the deterministic case,  a $(2^{nR},n,\eps)$ random code 
 for the AVGPC 
satisfies 
\label{eq:GrcodeReq}
\begin{align}
&
\sum_{\gamma\in\Gamma}\mu(\gamma) 
  \sum_{j=1}^d \norm{\fvec_{\gamma,j}(m)}^2    \leq n\plimit 
\,,\; \text{for all $m\in [1:2^{nR}]$}\,,  \label{eq:GcodeInputCr}
\intertext{and} 
&\Gcerr(\gcode)\triangleq \sum_{\gamma\in\Gamma} \mu(\gamma)  \Gcerr(\code_\gamma) 
\leq \eps \text{
for all $\svec^d\in\mathbb{R}^{nd}$ with $\sum_{j=1}^d \norm{\svec_j}^2 \leq n\Lambda$}\,. 
\end{align}
The capacity achieved by random codes is denoted by $\KrCav$, and it 
 is referred to as the \emph{random code capacity}.
\end{definition}


\subsection{Related Work}
Consider the AVGPC with parallel Gaussian channels,  where the covariance matrix of the additive noise is 
\begin{align}
\Sigma=\diag\{\sigma_1^2,\ldots,\sigma_d^2\} \,,
\end{align}
 \ie
 $Z_1,\ldots,Z_d$ are independent and $Z_j\sim\mathcal{N}(0,\sigma_j^2)$. 
Denote the random code capacity of the AVGPC with parallel channels by $\sigmarCav$.
%
Hughes and Narayan \cite{HughesNarayan:88p} have shown that the solution for the random code capacity is given by ``double" water filling, where the jammer performs water filling first, attempting to whiten the overall noise as much as possible, and then the user performs water filling taking into account the total noise power, which is contributed by both the channel and the jammer.
The formal definitions are given below.
Let
\begin{align}
N_j^*=\left[ \beta-\sigma_j^2 \right]_{+} \,,\; j\in [1:d]\,
\label{eq:GPNjdef}
\end{align}
with $[t]_{+}=\max\{0,t\}$, where $\beta\geq 0$ is chosen to satisfy
\begin{align}
\sum_{j=1}^d \left[ \beta-\sigma_j^2 \right]_{+}=\Lambda \,.
\label{eq:GPbetadef}
\end{align}
Next, let
\begin{align}
P_j^*=
\left[ \alpha-(N_j^*+\sigma_j^2) \right]_{+} \,,\; j\in [1:d]\,,
\label{eq:GPPjdef}
\end{align}
 where $\alpha\geq 0$ is chosen to satisfy
\begin{align}
\sum_{j=1}^d \left[ \alpha-(N_j^*+\sigma_j^2) \right]_{+}=\plimit \,.
\label{eq:GPalphadef}
\end{align}
We can now define Hughes and Narayan's capacity formula \cite{HughesNarayan:88p},  
\begin{align}
\sigmarICav\triangleq& \sum_{j=1}^d \frac{1}{2} \log\left( 1+\frac{P_j^*}{N_j^*+\sigma_j^2}  \right) \,.
\label{eq:sigmarICavdef}
\end{align}

 \begin{theorem}[see {\cite{HughesNarayan:88p}}] 
\label{theo:GPavcRand}
The random code capacity of the AVGPC is given by
\begin{align}
\sigmarCav=\sigmarICav \,.
\end{align}
\end{theorem}

\subsection{Observations on The Water Filling Game}
We give further observations on the results by Hughes and Narayan \cite{HughesNarayan:88p}, which will be useful in the sequel.

\subsubsection{Game Theoretic Interpretation}
By \cite[Theorem 3]{HughesNarayan:88p}, the random code capacity is the solution of the following optimization problem,
\begin{align}
\min \max \sum_{j=1}^d \frac{1}{2} \log\left( 1+\frac{P_j}{N_j+\sigma^2}  \right) \,,
\label{eq:GPoptimi}
\end{align}
 where the minimization is over the simplex
 $\Fset_{\text{state}}=\{ (N_1,\ldots,N_d) \,:\; \sum_{j=1}^d N_j\leq\Lambda \}$,
and the maximization is over the simplex
 $\Fset_{\text{input}}=\{ (P_1,\ldots,P_d) \,:\; \sum_{j=1}^d P_j\leq\plimit \}$.

The optimization problem is thus interpreted as a two-player zero-sum simultaneous game, played by the user and the jammer, where  $\Fset_{\text{input}}$ and $\Fset_{\text{state}}$ are the respective action sets.
The payoff function $v:\Fset_{\text{input}}\times \Fset_{\text{state}}\rightarrow \mathbb{R}$ is defined such that, given a profile $(P_1,\ldots,P_d,N_1,\ldots,N_d)$,
\begin{align}
v(P_1,\ldots,P_d,N_1,\ldots,N_d)\triangleq \sum_{j=1}^d \frac{1}{2} \log\left( 1+\frac{P_j}{N_j+\sigma^2}  \right) \,.
\end{align} 
We have defined a game with pure strategies, \ie the players' actions are deterministic.
In the communication model, the optimal coding and jamming scheme are random in general, yet the capacity can be achieved with deterministic power allocations, as in the game. 

The optimal power allocation has a water filling analogy (see \eg \cite[Section 9.4]{CoverThomas:06b}), 
where the jammer pours water of volume $\Lambda$ to a vessel,
 and then the encoder pours more water of volume $\plimit$.
The shape of the bottom of the vessel is determined by the noise variances $\sigma_1^2,\ldots,$$\sigma_d^2$.
 The jammer brings the water level to $\beta$, and then the encoder brings the water level to $\alpha$.
Water filling for the AVGPC is illustrated in Figure~\ref{fig:WaterF}, for $\plimit=13$, $\Lambda=8$, $d=10$, 
$(\sigma_j^2)_{j=1}^{10}=(5,8,3,1.5,2.5,1.8,3.2,9,4.5,5.5)$. 
The light shade ``fluid" is the jammer's water filling and the dark shade ``fluid" is the transmitter's.
 The resulting ``water levels" are $\beta=4$ and $\alpha=6$.
Then, substituting into  (\ref{eq:GPNjdef}) and (\ref{eq:GPPjdef}) yields the power allocations $(N_j^*)_{j=1}^{10}=(0,0,1,2.5,1.5,2.2,0.8,0,0,0)$ for the jammer and $(P_j^*)_{j=1}^{10}=(1,0,2,2,2,2,2,1.5,0.5)$ for the transmitter.

\begin{center}
\begin{figure}[htb]
        \centering
        \includegraphics[scale=0.5]
				{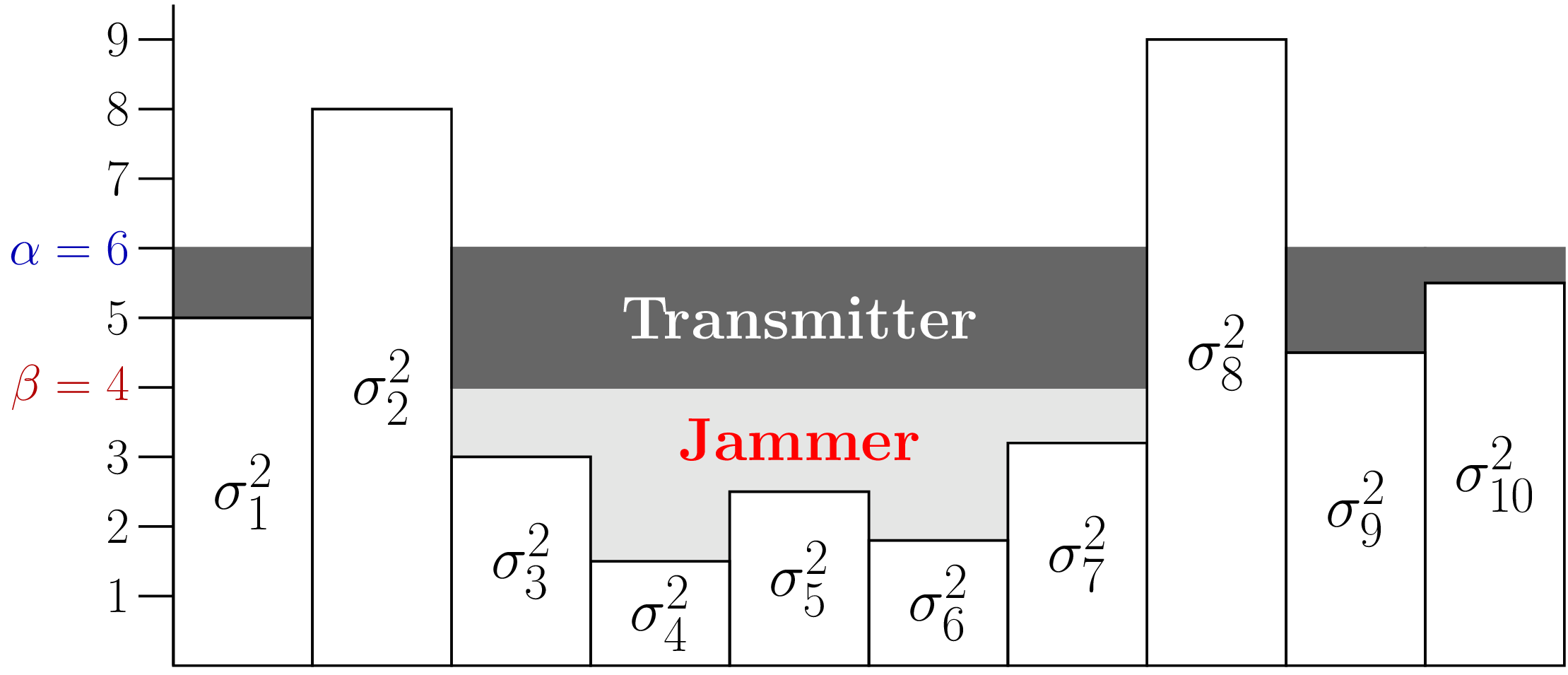}
        
\caption{Water filling for the AVGPC, for $\plimit=13$, $\Lambda=8$, $d=10$, 
$(\sigma_j^2)_{j=1}^{10}=(5,8,3,1.5,2.5,1.8,3.2,9,4.5,5.5)$. 
The light shade ``fluid" is the jammer's water filling and the dark shade ``fluid" is the transmitter's.
 The resulting ``water levels" are $\beta=4$ and $\alpha=6$, hence
$(N_j^*)_{j=1}^{10}=(0,0,1,2.5,1.5,2.2,0.8,0,0,0)$ and $(P_j^*)_{j=1}^{10}=(1,0,2,2,2,2,2,1.5,0.5)$.
  }
\label{fig:WaterF}
\end{figure}
\end{center}

One can easily prove the following properties of the random code capacity characterization.
\begin{lemma}
\label{lemm:WaterProp}
The quantities defined by (\ref{eq:GPNjdef})-(\ref{eq:sigmarICavdef}) 
 satisfy
\begin{align}
\begin{array}{ll}
1) \; \alpha>\beta			
&\quad 2)\;  N_j^*>0 \,\Rightarrow\; P_j^*>0 \; \forall\, j\in [1:d] \vspace{0.25cm}
\\ 
3) \; P_j^*+N_j^*+\sigma_j^2=\max(\alpha,\sigma_j^2)
&\quad 4) \; \sigmarICav= \sum_{j=1}^d \frac{1}{2} \log \frac{\max(\alpha,\sigma_j^2)}{\max(\beta,\sigma_j^2)} 
\,.
\end{array}
\end{align}
\end{lemma}
For completeness, we give the proof of Lemma~\ref{lemm:WaterProp} is given in Appendix~\ref{app:WaterProp}. 
%
 Based on the water filling analogy of the power allocation above, part 1 of Lemma~\ref{lemm:WaterProp} is natural, since $\beta$ is interpreted as the water level after the jammer pours his share, and $\alpha$ is interpreted as the water level after the user pours \emph{additional} water after that (see Figure~\ref{fig:WaterF}). 
Part 3 and part 4 are not surprising either since, as can be seen in Figure~\ref{fig:WaterF}, 
the variance of the combined interference $(Z_j+S_j)$ 
is $\max(\beta,\sigma_j^2)$ and the variance of the channel output $Y_j$ is 
$\max(\alpha,\sigma_j^2)$.

Observe that an equivalent statement of part 2 is the following. If the user discards a channel, \ie assigns $P_j^*=0$ to the $j$th channel, then the jammer does not invest power in this channel either, \ie  $N_j^*=0$. This claim is also intuitive, and from a game theoretic perspective, it is an aspect of the jammer's rationality, as explained below. 
As mentioned above the optimization problem is interpreted as a two-player zero-sum simultaneous game
between the user and the jammer. 
The value of such a game is attained by a pair of strategies 
 which forms a Nash equilibrium \cite{vonNeumannMorgenstern:07b} (see also \cite{Owen:13b}\cite[Theorem 3.1.4]{LeytonBrownShoham:08b}). That is, if the user and the jammer were to agree to use the power allocation strategies $(P_j^*)_{j=1}^d$ and $(N_j^*)_{j=1}^d$, then neither player could profit by deviating from his original strategy, provided that the other player respects the agreement. Now, suppose that for some $j\in [1:d]$, $P_j^*=0$ and $N_j^*>0$. Then, the jammer is wasting energy, and can surely profit from diverging this energy to some other channel $j'$ with $P_{j'}^*>0$. Thus, such strategy profile is irrational and cannot be a Nash equilibrium. 

For a general AVC, a coding scheme which assumes that the jammer is using his optimal strategy would typically fail. The code needs to be robust standing against any state sequence that satisfies the state constraint. 
For example, consider a scalar Gaussian AVC \cite{HughesNarayan:87p}, specified by $\Yvec=\Xvec+\Svec+\Zvec$, under input constraint $\norm{\Xvec}^2\leq n\plimit$ and state constraint $\norm{\Svec}^2\leq n\Lambda$, where the noise sequence $\Zvec$ is i.i.d. $\sim\mathcal{N}(0,\sigma^2)$. Suppose that the receiver is using joint typicality decoding  for a Gaussian channel $\Yvec=\Xvec+\Vvec$, where $\Vvec$ is i.i.d. $\sim\mathcal{N}(0,\Lambda+\sigma^2)$ (see \cite[Section 9.1]{CoverThomas:06b}), corresponding to the optimal jamming strategy. Then, the jammer can fail the decoder by selecting a state sequence such that $\norm{\Svec}^2=\frac{n\Lambda}{2}$, 
for instance. As a result, there is a high probability that
 the square norm of the output sequence is below 
$n(\Lambda+\sigma^2-\delta)$, for small $\delta>0$, in which case the decoder cannot establish  joint typicality and declares an error.
The same principle holds in our problem. The user cannot assume that the jammer is using his optimal power allocation,  and a reliable code must be robust standing against any power allocation of the jammer.

\subsubsection{Multiple Access Channel Analogy}

Water filling in two (or more) stages appears in other settings in the literature, \eg
\cite{ChengVerdu:93p,LaiElGamal:08p,WCLM:99p,YuCioffi:02p}.
Consider a Gaussian product multiple access channel (MAC), where $Y_j=X_{1,j}+X_{2,j}+Z_j$, $j\in [1:d]$, under the input constraints 
$\norm{\Xvec_1^d}^2\leq n\plimit$ and $\norm{\Xvec_2^d}^2\leq n\Lambda$.
This can be viewed as a different variation of the AVGPC where a second transmitter replaces the jammer.
By \cite{ChengVerdu:93p}, 
a corner point of the capacity region can be achieved by applying water filling to the total power in the first step, and then to the power of User 2 in the second step.
Specifically, by \cite[Section III.B.]{ChengVerdu:93p}, the optimal power allocations $(P_j^*)_{j=1}^d$ and $(N_j^*)_{j=1}^d$, for Encoder 1 and Encoder 2, respectively, which achieve a corner point of the capacity region,
 satisfy 
\begin{align}
& P_j^*+N_j^*=\left[ \alpha-\sigma_j^2 \right]_{+} \,,\; j\in [1:d]\,,
\label{eq:GPNjdefMAC}
\intertext{such that $\sum_{j=1}^d (P_j^*+N_j^*)=\plimit+\Lambda$,
and}
& N_j^*=\left[ \beta-\sigma_j^2 \right]_{+} \,,\; j\in [1:d]\,,
\label{eq:GPPjdefMAC}
\end{align}
such that $\sum_{j=1}^d N_j^*=\Lambda$. 
Following part 3 of Lemma~\ref{lemm:WaterProp}, it can be seen that the strategy above is equivalent to 
(\ref{eq:GPNjdef})-(\ref{eq:GPalphadef}). The total power allocation in (\ref{eq:GPNjdefMAC}) seems natural in order to  maximize the sum rate. Though, our presentation in (\ref{eq:GPNjdef})-(\ref{eq:GPalphadef}) is intuitive for the Gaussian product MAC as well. Indeed, 
using successive cancellation decoding, the receiver estimates  the transmission of User 1  while treating the transmission of User 2 as noise, and then subtracts the estimated sequence from the received sequence to decode the transmission of User 2.
Hence, decoding for User 1 is analogous to the decoder in our problem.
Nevertheless, in the next section, we show that the deterministic code capacity in our adversarial problem has a different behavior.
%
%

Another water filling game is described by Lai and El Gamal in  \cite{LaiElGamal:08p}, who considered the flat fading MAC $Y=h_1 X_1+h_2 X_2+Z$ with selfish users, 
where the fading coefficients  are continuous random variables, distributed according to $(h_1,h_2)\sim\mu$.
Suppose that the users are subject to average input constraints, $\E_\mu \norm{\Xvec_1}^2\leq n\plimit$ and 
$\E_\mu \norm{\Xvec_2}^2\leq n\Lambda$.
As shown in  \cite{LaiElGamal:08p}, a maximum sum-rate point on the capacity region boundary is achieved if the users perform water filling treating each other's transmission as noise. 
It is further shown that opportunistic communication is optimal, where User 1 only transmits if his water level times fading coefficient is at least as high as that of User 2, and vice versa.
That is, the power allocations of the users are given by
\begin{align}
P_{h_1,h_2}^*= \begin{cases}
\left[ \beta_1-\sigma^2/h_1 \right]_{+} &
\text{if $\beta_1 h_1 \geq \beta_2 h_2$} \,,
\\ 
0 &\text{otherwise} 
\end{cases}
\,,\quad \nonumber\\
N_{h_1,h_2}^*=\begin{cases}
\left[ \beta_2-\sigma_j^2/h_2 \right]_{+} &
\text{if $\beta_1 h_1 \leq \beta_2 h_2$} \,,
\\
0 &\text{otherwise} 
\end{cases} \,,
\end{align}
where $\beta_1$ and $\beta_2$ are chosen such that 
$\E P_{h_1,h_2}^*=\plimit$ and $\E N_{h_1,h_2}^*=\Lambda$.
This threshold operation resembles the result in the next section, on the deterministic code capacity of the AVGPC, except that the phase transition of the AVGPC depends only on the ``water volumes" $\plimit$ and $\Lambda$ (see Subsection~\ref{sec:AVGPCdisc}). 

\subsection{Results}
\label{sec:AVGPCres}
We give our result on the AVGPC with parallel Gaussian channels,  where the covariance matrix of the additive noise is 
$\Sigma=\diag\{\sigma_1^2,\ldots,\sigma_d^2\}$, \ie
 $Z_1,\ldots,Z_d$ are independent and $Z_j\sim\mathcal{N}(0,\sigma_j^2)$.
The deterministic code capacity of the AVGPC with parallel channels is denoted by $\sigmaCavc$.

We establish the capacity of the AVGPC.
Based on Csisz\'{a}r and Narayan's result in \cite{CsiszarNarayan:88p}, 
 the deterministic code capacity of an AVC under input and state constraints
is given in terms of channel symmetrizability and the minimal state cost for the jammer to symmetrize the channel  (see also \cite{LapidothNarayan:98p} \cite[Definition 5 and Theorem 5]{PeregSteinberg:19p1}). 
By \cite[Definition 2]{CsiszarNarayan:88p}, a AVGPC is symmetrized by 
a conditional pdf
$\varphi(s^d|x^d)$ if 
\begin{align}
\label{eq:GPsymmetrizable}
\int_{-\infty}^\infty\cdots \int_{-\infty}^\infty  \varphi(s^d|x_2^d)f_{Z^d}(y^d-x_1^d-s^d)ds^d=
\int_{-\infty}^\infty\cdots \int_{-\infty}^\infty  \varphi(s^d|x_1^d)f_{Z^d}(y^d-x_2^d-s^d)ds^d
 \,,\; 
\forall\, x_1^d,x_2^d,y^d\in\mathbb{R}^d \,,
\end{align}
where $f_{Z^d}(z^d)=\prod_{j=1}^d\frac{1}{\sqrt{2\pi\sigma_j^2}} e^{-z_j^2/2\sigma_j^2}$. 
%
In particular, observe that (\ref{eq:GPsymmetrizable}) holds for $\varphi(s^d|x^d)=\delta(s^d-x^d)$, where $\delta(\cdot)$ is the Dirac delta function. In other words, the channel is symmetrized by a distribution $\varphi(s^d|x^d)$ which gives probability $1$ to $S^d=x^d$.
For the AVGPC, the minimal state cost for the jammer to symmetrize the channel, for an input distribution 
$f_{X^d}$, is given by
\begin{align}
\LambdaOig(F_{X^d})=\min\, \int_{-\infty}^\infty\cdots \int_{-\infty}^\infty 
  f_{X^d}(x^d)\varphi(s^d|x^d)\norm{s^d}^2 ds^d dx^d
 \,,
\label{eq:GPlambdaT}
\end{align}
where the minimization is over all conditional pdfs $\varphi(s^d|x^d)$ that symmetrize the channel, that is, satisfy (\ref{eq:GPsymmetrizable}).
The following lemma states that the minimal state cost for symmetrizability is the same 
as the input power. The lemma will be used in the achievability proof of the capacity theorem.
\begin{lemma}
\label{lemm:GPscostP}
%
For a zero mean Gaussian vector $X^d\sim\mathcal{N}(\mathbf{0},K_X)$,
\begin{align}
\LambdaOig(F_{X^d})= \trace(K_X) \,.
\end{align}
\end{lemma}
The proof of Lemma~\ref{lemm:GPscostP} is given in Appendix~\ref{app:GPscostP}.
The proof builds on our observation that 
(\ref{eq:GPsymmetrizable}) holds if and only if $\varphi(s^d|x^d)=\varphi(s^d-x^d|0)$. This in turn leads to the conclusion that the minimum in (\ref{eq:GPlambdaT}) is attained by $\varphi_{x^d}(s^d)=\delta(s^d-x^d)$.
Moving to the capacity theorem, define
\begin{align}
\sigmaICavc=
\begin{cases}
\sigmarICav &\text{if $\plimit>\Lambda$}, \\
0						&\text{otherwise}.
\end{cases}
\end{align}

\begin{theorem}
\label{theo:GPavcDet}
The deterministic code capacity of the AVGPC is given by
\begin{align}
\sigmaCavc=\sigmaICavc \,.
\end{align}
\end{theorem}
The proof of Theorem~\ref{theo:GPavcDet} is given in Appendix~\ref{app:GPavcDet}.
Considering the scalar case, 
Csisz{\'a}r and Narayan showed the direct part by providing a coding scheme for the Gaussian AVC \cite{CsiszarNarayan:91p}. While the receiver in their coding scheme uses simple minimum-distance decoding, the analysis is fairly complicated.
Here, on the other hand, we treat the AVGPC using a much simpler approach. To prove direct part, we consider the optimization problem based on the capacity formula of the general AVC under input and state constraints, which is given in terms of symmetrizing state distributions. We use Lemma~\ref{lemm:GPscostP} to show that if $\plimit>\Lambda$, then the transmitter's water filling strategy in (\ref{eq:GPPjdef}) guarantees that $\LambdaOig(F_{x^d})>\Lambda$. Intuitively, this means that the jammer cannot symmetrize the channel without violating the state constraint. In this scenario, the random code capacity can be achieved with deterministic codes as well.

\subsection{Discussion}
\label{sec:AVGPCdisc}
We give a couple of remarks on our result in Theorem~\ref{theo:GPavcDet}.
%
As in the case of the Gaussian scalar AVC \cite{CsiszarNarayan:91p}, the capacity is disconinuous in the input constraint, and has a  phase transition behavior, depending on whether 
$\plimit>\Lambda$ or $\plimit\leq\Lambda$.
We give an intuitive explanation below. 
For the classic Gaussian AVC, reliable communication requires the power of the transmitted signal to be higher than the  power of the jamming signal, otherwise the jammer can confuse the receiver by making the state sequence $\Svec$ ``look like" the input sequence $\Xvec$ \cite{CsiszarNarayan:91p}.
At a first glance at our problem, one might have expected that the input power $P_j$ of the $j$th channel also needs to be higher than the jamming power $N_j$, in order for the output $\Yvec_j$ to be useful. 
This is not the case. Since the decoder has the vector of outputs $(\Yvec_1,\ldots,\Yvec_d)$, even if 
$\Svec_j$ looks like $\Xvec_j$, the receiver could still gain information from $\Yvec_j$ as the other outputs may ``break the symmetry".  

Based on Shannon's classic water filling result \cite{Shannon:49p},
the capacity of the Gaussian product channel, $Y_j=X_j+V_j$, $j\in [1:d]$,  can be achieved by combining $d$ independent encoder-decoder pairs, where the $j$th pair is associated with a capacity achieving code for the scalar Gaussian channel under input constraint $P_j^*$. However,
based on Csisz{\'{a}}r and Narayan's result on the Gaussian single AVC \cite{CsiszarNarayan:91p}, the capacity of the $j$th AVC, $Y_j=X_j+S_j+Z_j$, is zero under input constraint $P_j^*$ and state constraint $N_j^*$ for $P_j^*\leq N_j^*$. This means that, in contrast to the Shannon's Gaussian product channel \cite{Shannon:49p}, using $d$ independent encoder-decoder pairs over the AVGPC is suboptimal in general. This can be viewed as a constrained version of the  
super-additivity phenomenon in \cite{SchaeferBochePoor:16c}. 

\section{Main Results -- AVC with Colored Gaussian Noise}
\label{sec:GaussCol}

\begin{center}
\begin{figure}[htb]
        \centering
        \includegraphics[scale=0.4,trim={1cm 0 0 0},clip] 
				{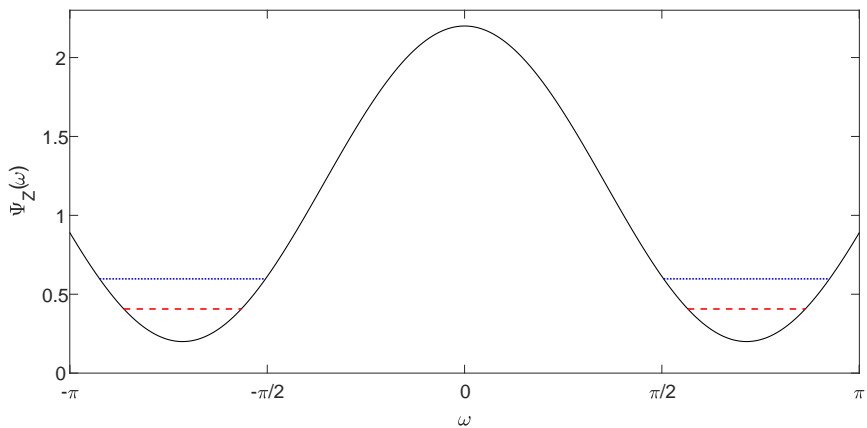}
        
\caption{Water filling in the frequency domain for the AVC with colored Gaussian noise.
The curve depicts the power spectral density $\Psi_Z(\omega)$ of the noise process $Z^n$. 
The red dashed line indicates the ``water level" $\beta$ which corresponds to the jammer's water filling,
and the blue dotted line indicates the ``water level" $\alpha$ which corresponds to the transmitter's water filling.  
 }
\label{fig:WaterFstationary}
\end{figure}
\end{center}

We consider an AVC with colored Gaussian noise, \ie
\begin{align}
\Yvec=\Xvec+\Zvec+\Svec \,,\; 
\label{eq:coloredAVC}
\end{align}
where $\Zvec$ is a zero mean stationary Gaussian process,  with 
%
 power spectral density $\Psi_Z(\omega)$. 
Assume that the power spectral density is bounded and integrable. 
We denote the random code capacity and the deterministic code capacity  of this channel by $\sKrCav$ and
$\sKCavc$, respectively.

We show that the optimal power allocations of the user and the jammer are given by ``double" water filling in the frequency domain.
Define 
\begin{align}
b^*(\omega)=  \left[ \beta-\Psi_Z(\omega) \right]_{+}  \,,\; -\pi\leq\omega\leq\pi\,,
\label{eq:sKGPbidef}
\end{align}
where $\beta\geq 0$ is chosen to satisfy
\begin{align}
\frac{1}{2\pi}\int_{-\pi}^{\pi} \left[ \beta-\Psi_Z(\omega) \right]_{+} \,d\omega=\Lambda \,.
\label{eq:sKGPbetadef}
\end{align}
Next, define 
\begin{align}
a^*(\omega)=\left[ \alpha-(b^*(\omega)+\Psi_Z(\omega)) \right]_{+} \,,\; -\pi\leq\omega\leq\pi\,,
\label{eq:sKGPaidef}
\end{align}
 where $\alpha\geq 0$ is chosen to satisfy
\begin{align}
\frac{1}{2\pi}\int_{-\pi}^{\pi} \left[ \alpha-(b^*(\omega)+\Psi_Z(\omega)) \right]_{+} \,d\omega=\plimit \,.
\label{eq:sKGPalphadef}
\end{align}
Now, let
\begin{align}
\sKrICav\triangleq \frac{1}{2\pi}\int_{-\pi}^{\pi} \frac{1}{2}\log\left(1+\frac{a^*(\omega)}{b^*(\omega)+
\Psi_Z(\omega)}  \right)\, d\omega \,.
\end{align}

\begin{theorem}
\label{theo:sKGPavcRand}
The random code capacity of the AVC with colored Gaussian noise is given by
\begin{align}
\sKrCav=\sKrICav \,,
\end{align}
and the deterministic code capacity is given by 
\begin{align}
\sKCavc= \begin{cases}
\sKrICav   & \text{if $\plimit>\Lambda$}\,,\\
0					& \text{otherwise}\,.
\end{cases}
\end{align}
\end{theorem}
The proof of Theorem~\ref{theo:sKGPavcRand} is given in Appendix~\ref{app:KGPavcRand}, combining our previous results on the AVC with fixed parameters and the AVGPC.
Despite the common belief that the characterization for a channel with colored Gaussian noise easily follows  from the results for the product channel setting, the analysis is more involved. While standard orthogonalization transforms the channel into an equivalent one  with statistically independent noise instances, the noise in the transformed channel is not necessarily white. As the noise variance may change over time, we observe that the transformed channel is in fact an AVC with fixed parameters which represent the sequence of noise variances. Using Corollary~\ref{coro:PrCavE} and Corollary~\ref{coro:PrCavEDet}, we obtain deterministic and random capacity formulas that are analogous to those of the AVGPC, and use Toeplitz matrix properties to express the formulas as integrals in the frequency domain.

The optimal power allocation has a water filling analogy in the frequency domain (see \eg \cite[Section 9.5]{CoverThomas:06b}), where the jammer pours water of volume $\Lambda$ on top of the power spectral density
$\Psi_Z(\omega)$,
 and then the encoder pours more water of volume $\plimit$. The jammer brings the water level to $\beta$, and then the encoder brings the water level to $\alpha$.
The process is illustrated in Figure~\ref{fig:WaterFstationary}.

\begin{appendices}

\section{Proof of Theorem~\ref{lemm:PCcompound}}
\label{app:PCcompound}
Consider the compound channel $\compound$ with fixed parameters under input constraint $\plimit$ and state constraint $\Lambda$. 

\subsection{Achievability Proof}
To show achievability,  we construct a code based on conditional typicality decoding  with respect to a channel state type, which is ``close" to one of the state distributions in $\Qset$. 

Denote the type of the parameter sequence by $P_T=\emp_{\theta^n}$.
Define a set $\tQ$ of conditional state types,
\begin{align}
\tQ= \left\{ \emp_{s^n|\theta^n} \,:\; (\theta^n,s^n)\in\Aset^{(n)}_{\delta_1}(P_T\times q) \,,\;\text{for some $q\in\Qset$}  \right\} \,,
\label{eq:1tQ}
\end{align}
with $(P_T\times q)(t,s)=P_T(t)q(s|t)$, and
\begin{align}
\label{eq:delta1CompNoSI}
\delta_1\triangleq \frac{\delta}{2\cdot |\Sset|} \,,
\end{align}
where $\delta>0$ is arbitrarily small.
In words, $\tQ$ is the set of conditional types $q'(s|t)$, given a parameter sequence $\theta^n$, such that the joint type 
is $\delta_1$-close to $P_T(t) q(s|t)$, for some conditional state distribution $q(s|t)$ in $\Qset$. 
We note that the sets $\Qset$ and $\tQ$ could be disjoint, since $\Qset$ is not limited to conditional empirical distributions.
Nevertheless, for a fixed $\delta>0$ and sufficiently large $n$,  every $q\in\Qset$ can be approximated by 
some $q'\in\tQ$. Indeed, for sufficiently large $n$, there exists a joint type $P'_T(t) q'(s|t)$ such that 
$|P'_T(t) q'(s|t)-P_T(t)q(s|t)|\leq \delta_1/|\Sset|$, hence $|P'_T(t)-P_T(t)|\leq \delta_1$ and
$|P_T(t) q'(s|t)-P_T(t)q(s|t)|\leq \delta_1 q'(s|t)\leq \delta_1 $.
%
Now, a code is constructed as follows.

\emph{Codebook Generation:} Fix $P_{X|T}$ such that $\E \cost(X)\leq \plimit-\eps$, where
\begin{align}
\E \cost(X)=\sum_{t\in\Tset} P_T(t) \E(\cost(X)|T=t)=
\frac{1}{n} \sum_{i=1}^n \sum_{x\in\Xset} P_{X|T}(x|\theta_i) \cost(x) \,.
\end{align}
 Generate $2^{nR}$ independent sequences at random,
$
x^{n}(m,\theta^n) \sim  \prod_{i=1}^{n} P_{X|T}(x_{i}|\theta_i) 
$, 
for $m\in[1:2^{n R}]$.

\emph{Encoding}: To send a message $m$, if $\cost^n(x^n(m,\theta^n))\leq \plimit$, transmit $x^n(m,\theta^n)$. Otherwise, transmit an idle sequence 
$x^n=(a,a,\ldots,a)$ with $\cost(a)=0$. 

\emph{Decoding}: Find a unique $\hm\in [1:2^{nR}]$ for which there exists $q\in\tQ$ such that
 $(\theta^n,x^{n}(\hm,\theta^n),y^{n})\in\tset( P_T P^{q}_{X,Y|T})$, where 
 \begin{align}
\label{eq:UchannelY}
P^{q}_{X,Y|T}(x,y|t)= 
 P_{X|T}(x|t)  \sum_{s\in\Sset} q(s|t)  W_{Y|X,S,T}(y|x,s,t) \,.
 \end{align}
If there is none, or more than one such $\hm$,  declare an error. We note that using the set of types $\tQ$ instead of the original set of state distributions $\Qset$ alleviates the analysis, since $\Qset$ is not necessarily finite nor countable.


\emph{Analysis of Probability of Error}:
Assume without loss of generality that the user sent $M=1$. 
By the union of events bound, we have that
$
 \prob{\hM\neq 1}\leq \prob{ \Eset_{1} }+ \cprob{ \Eset_{2} }{ \Eset_1^c }+ 
\cprob{ \Eset_{3} }{ \Eset_1^c } 
$, 
where
\begin{align}
\Eset_1=& \{ (\theta^n,X^{n}(1,\theta^n))\notin \tset( P_T P_{X|T} ) \} \,, \nonumber\\
\Eset_{2} =&\{ (\theta^n,X^{n}(1,\theta^n),Y^{n})\notin \tset( P_T P_{X|T} P^{q'}_{Y|X,T}) \;\text{for all $q'\in\tQ$ } \} \,, \nonumber\\
\Eset_{3} =&\{ (\theta^n,X^{n}(m,\theta^n),Y^{n})\in \tset( P_T P_{X|T} P^{q'}_{Y|X,T}) 
\;\text{ for some  $m\neq 1$,\, $q'\in\tQ$ } \} \,.
\end{align}
The first term tends to zero exponentially by the law of large numbers and Chernoff's bound (see \eg \cite[Theorem 1.2]{Kramer:08n}).
Now, suppose that the event $\Eset_1^c$ occurs.
Then, for sufficiently small $\delta$, we have that $\cost^n(X^n(1,\theta^n))\leq \plimit$, since 
$\E \cost(X)\leq \plimit-\eps$. Hence, $X^n(1,\theta^n)$ is the channel input.

Next, we claim that the second error event implies that $(\theta^n,X^{n}(1,\theta^n),Y^{n})\notin \Aset^{(n)}_{\nicefrac{\delta}{2}}
(P_T P_{X|T} P^{q}_{Y|X,T})$, where $q(s|t)$ is the \emph{actual} state distribution chosen by the jammer. 	Assume to the contrary that $\Eset_{2}$ holds, but   
	$(\theta^n,X^{n}(1,\theta^n),Y^{n})\in \Aset^{(n)}_{\nicefrac{\delta}{2}}(P_T P_{X|T} P^{q}_{Y|X,T})$. 
	For sufficiently large $n$, there exists a conditional type $q'\in\tQ$ that approximates $q$  in the sense that 
	$
	|P_T(t) q'(s|t)-P_T(t) q(s|t)|\leq \delta_1 
	$ for all $s\in\Sset$ and $t\in\Tset$, hence  
	\begin{align}
	|P_T(t) P_{Y|X,T}^{q'}(y|x,t)-P_T(t) P_{Y|X,T}^{q}(y|x,t)|\leq |\Sset|\cdot\delta_1=\frac{\delta}{2} \,,
	\label{eq:qqPcompD}
	\end{align}
 for all $x\in\Xset$, $t\in\Tset$, $y\in\Yset$	 (see  (\ref{eq:delta1CompNoSI})-(\ref{eq:UchannelY})).
	To show $\delta$-typicality with respect to $q'(s|t)$, 
	we observe that
	\begin{align}
	&|\hP_{\theta^n,X^{n}(1,\theta^n),Y^{n}}(t,x,y) - P_T(t) P_{X|T}(x|t) P_{Y|X,T}^{q'}(y|x,t)| \nonumber\\
	=& \Big|\hP_{\theta^n,X^{n}(1,\theta^n),Y^{n}}(t,x,y) - P_T(t) P_{X|T}(x|t) P_{Y|X,T}^{q}(y|x,t) 	+ P_T(t) P_{X|T}(x|t) P_{Y|X,T}^{q}(y|x,t) \nonumber\\&- P_T(t) P_{X|T}(x|t) P_{Y|X,T}^{q'}(y|x,t) \Big| \nonumber\\
	\leq& |\hP_{\theta^n,X^{n}(1,\theta^n),Y^{n}}(t,x,y) - P_T(t) P_{X|T}(x|t) P_{Y|X,T}^{q}(y|x,t)| \nonumber\\&+ |P_T(t) P_{X|T}(x|t) P_{Y|X,T}^{q}(y|x,t)- P_T(t) P_{X|T}(x|t) P_{Y|X,T}^{q'}(y|x,t)| \nonumber\\
	\leq& \frac{\delta}{2}+\frac{\delta}{2}P_{X|T}(x|t) \leq \delta \,,
	\end{align}
	where the first inequality is due to the triangle inequality, and the second inequality follows from (\ref{eq:qqPcompD}) and the assumption that $(\theta^n,X^{n}(1,\theta^n),Y^{n})\in \Aset^{(n)}_{\nicefrac{\delta}{2}}(P_T P_{X|T} P^{q}_{Y|X,T})$.
	It follows that $(\theta^n,X^{n}(1,\theta^n),Y^{n})\in \Aset^{(n)}_{\delta}( P_T P_{X|T} P^{q'}_{Y|X,T})$, and
	  $\Eset_2$ does not hold. 
	Thus,
	\begin{align} 
	\label{eq:llnRL}
 \cprob{ \Eset_{2} }{\Eset_1^c} 
	\leq& \prob{(\theta^n,X^{n}(1,\theta^n),Y^{n})\notin \Aset^{(n)}_{\nicefrac{\delta}{2}} (P_T P_{X|T} P^{q}_{Y|X,T}) } \,.
	\end{align}
 This tends to zero exponentially as $n\rightarrow\infty$ by the law of large numbers and Chernoff's bound (see \eg \cite[Theorem 1.2]{Kramer:08n}).
	
	Moving to the third error event, as the number of type classes in $\Sset^n$ is bounded by $(n+1)^{|\Sset|}$, 
 we have that  
\begin{align}
\cprob{\Eset_{3}}{\Eset_1^c}
\leq (n+1)^{|\Sset|} \cdot \sup_{q'\in\tQ} 
 \prob{
(\theta^n,X^{n}(m,\theta^n),Y^{n})\in \tset( P_T P_{X|T} P^{q}_{Y|X,T})  \;\text{ for some  $m\neq 1$} 
}.
\label{eq:E2poly}
\end{align}
For every $m\neq 1$, $X^{n}(m,\theta^n)$ is independent of $Y^{n}$, hence
\begin{align}
&
\prob{  (\theta^n,X^{n}(m),Y^{n})\in \tset( P_T P_{X|T} P^{q}_{Y|X,T}) } 
\nonumber\\
=&  \sum_{x^n\in\Xset^n} P_{X^{n}|T^{n}}(x^{n}|\theta^n)  \sum_{y^{n} \,:\; (\theta^n,x^n,y^n)\in \tset(P_T P_{X|T} P^{q'}_{Y|X,T})} 
P_{Y^{n}|T^{n}}^q(y^{n}|\theta^{n}) \,.
\label{eq:E2bound0} 
\end{align}
 Let 
 $(\theta^n,x^n,y^n)\in \tset( P_T P_{X|T} P^{q'}_{Y|X,T})$. Then, $\, (\theta^n,y^n)\in\Aset^{(n)}_{\delta_2}( P_T P_{Y|T}^{q'})$ with $\delta_2\triangleq |\Xset|\cdot\delta$. By Lemmas 2.6-2.7 in \cite{CsiszarKorner:82b},
\begin{align}
\label{eq:pYbound}
P_{Y^{n}|T^n}^{q}(y^{n}|\theta^n)=2^{-n\left(  H(\hP_{y^{n_t}|\theta^n})+D(\hP_{y^{n}|\theta^n}||P_{Y|T}) \right)}\leq 2^{-n 
H(\hP_{y^{n}|\theta^n})}
\leq 2^{-n \left( H_{q'}(Y|T) -\eps_1(\delta) \right)} \,,
\end{align}
where $\eps_1(\delta)\rightarrow 0$ as $\delta\rightarrow 0$. Therefore, by (\ref{eq:E2poly})$-$(\ref{eq:pYbound}),
\begin{align}
 &
\prob{\Eset_{2}}           								\leq
 (n+1)^{|\Sset|} \cdot \sup_{q'\in\tQ}    2^{n R}  \nonumber\\ &  
\sum_{x^{n}\in\Xset^{n}} P_{X^{n}|T^{n}}(x^{n}|\theta^{n})
\cdot |\{y^{n}\,:\; (\theta^n,x^{n},y^{n})\in\tset(P_T P_{X|T} P_{Y|X,T}^{q'})\}| \cdot 
 2^{-n\left( H_{q'}(Y|T) -\eps_1(\delta) \right)}		   												\nonumber\\
\leq& \sup_{q'\in\tQ}  (n+1)^{|\Sset|} 
2^{-n[ I_{q'}(X;Y|T) 
-R-\eps_2(\delta) ]} \label{eq:expCR2} \,,
\end{align}
with $\eps_2(\delta)\rightarrow 0$ as $\delta\rightarrow 0$, 
where the last inequality is due to \cite[Lemma 2.13]{CsiszarKorner:82b}. The RHS of (\ref{eq:expCR2})
  tends to zero exponentially as $n\rightarrow\infty$, provided that $R< I_{q'}(X;Y|T)
-\eps_2(\delta)$.  
The probability of error, averaged over the class of codebooks, exponentially decays to zero  as $n\rightarrow\infty$. Therefore, there must exist a $(2^{nR},n,e^{-an})$ deterministic code, for a sufficiently large $n$.
This completes the proof of the direct part.

\subsection{Converse Proof}
Since the deterministic code capacity is always bounded by the random code capacity,
we consider a sequence of $(2^{nR},n,\alpha_n)$ random codes, where $\alpha_n\rightarrow 0$  as $n\rightarrow\infty$.
Then, let $X^n=f_{\gamma}^n(M,\theta^n)$  be the channel input sequence, and 
$Y^n$ be the corresponding output sequence, where $\gamma\in\Gamma$ is the random element shared between the encoders and the decoder.
 For every $q\in\Qset$, we have by 
Fano's inequality that $H_q(M|Y^n,T^n=\theta^n,\gamma)\leq n\eps_n $, hence
\begin{align}
nR=& H(M|T^n=\theta^n,\gamma)=I_q(M;Y^n|T^n=\theta^n,\gamma)+H(M|Y^n,T^n=\theta^n,\gamma)
\nonumber\\
\leq& I_q(M,\gamma;Y^n|T^n=\theta^n)+n \eps_n =I_q(M,\gamma,X^n;Y^n|T^n=\theta^n)+n \eps_n
\nonumber\\
=& I_q(X^n;Y^n|T^n=\theta^n)+n \eps_n 
\,,
\end{align}
where $\eps_n\rightarrow 0$ as $n\rightarrow\infty$. 
The third equality holds since $X^n$ is a deterministic function of $(M,\gamma,\theta^n)$, and 
the last equality since $(M,\gamma)\Cbar(X^n,T^n)\Cbar Y^n$ form a Markov chain.
%
It follows that
\begin{align}
R-\eps_n\leq&    \frac{1}{n}\sum_{i=1}^n I_q(X_i;Y_i|T_i=\theta_i)= I_q(X;Y|T,K)  \leq I_q(X,K;Y|T)
\end{align}
for all $q\in\Qset$, with $X\equiv X_K$, $Y\equiv Y_K$, $T\equiv T_K=\theta_K$, where the random variable $K$ is uniformly distributed over $[1:n]$,
 and $\eps_n\rightarrow 0$ as $n\rightarrow\infty$.
Observe that the random variable $T$ is distributed according to
\begin{align}
P_T(t)=\prob{ \theta_K=t} =\sum_{i \,:\; \theta_i=t} \prob{K=i}=\frac{1}{n} \cdot N(t|\theta^n) =\hP_{\theta^n}(t) \,,
\label{eq:ConvPT}
\end{align}
where $N(t|\theta^n)$ is the number of occurrences of the symbol $t\in\Tset$ in the sequence $\theta^n$.
Since $K\Cbar (T,X)\Cbar Y$ form a Markov chain, we have that 
\begin{align}
R-\eps_n \leq \inf_{q\in\Qset} I_q(K,X;Y|T)= \inf_{q\in\Qset} I_q(X;Y|T) \,.
\end{align}
\qed

\section{Proof of Lemma~\ref{lemm:LRT}}
\label{app:LRT}
We state the proof of our modified version of Ahlswede's RT \cite{Ahlswede:86p}. The proof follows the lines of \cite[Subsection IV-B]{Ahlswede:86p}, which we modify here  to include a constraint on the family of state distributions $q(s)$ and the parameter sequence 
$\theta^n$. 
Let $\widetilde{s}^{\;n}\in\Sset^n$ such that $l^n(\widetilde{s}^{\;n})\leq\Lambda$. Denote the conditional type of  $\widetilde{s}^{\;n}\in\Sset^n$ given $\theta^n$ by $\hq(s|t) 
$.  
Observe that 
$
\hq\in \apLSpaceS 
$ (see (\ref{eq:StateCcompound})), 
since $\frac{1}{n}\sum_{i=1}^n \sum_{s\in\Sset} q(s|\theta_i) l(s) =l^n(\widetilde{s}^{\;n})$.

Given a permutation $\pi\in\Pi(\theta^n)$,
\begin{align}
\sum_{s^n\in\Sset^n} q^n(s^n|\theta^n) h(s^n,\theta^n)
=\sum_{s^n\in\Sset^n} q^n(\pi s^n|\theta^n) h(\pi s^n,\theta^n)
=\sum_{s^n\in\Sset^n} q^n(\pi s^n|\pi \theta^n) h(\pi s^n,\pi \theta^n)
=\sum_{s^n\in\Sset^n} q^n(s^n|\theta^n) h(\pi s^n,\pi \theta^n) \,,
\end{align}
where the first equality holds since $\pi$ is a bijection, the second equality holds since $\pi \theta^n=\theta^n$
for every $\pi\in\Pi(\theta^n)$, and the last equality holds due to the product form of the conditional distribution $q^n(s^n|t^n)=\prod_{i=1}^n q(s_i|t_i)$.
Hence, taking $q=\hq$,
\begin{align}
\sum_{s^n\in\Sset^n} \hq^{\;n}(s^n|\theta^n) h(s^n,\theta^n)=\frac{1}{|\Pi(\theta^n)|} \sum_{\pi\in\Pi(\theta^n)} \sum_{s^n\in\Sset^n} \hq^{\;n}(s^n|\theta^n) h(\pi s^n,\pi \theta^n) \,,
\end{align}
and by (\ref{eq:RTcondCs}),
\begin{align}
\sum_{s^n\in\Sset^n}  \hq^{\;n}(s^n|\theta^n) \left[\frac{1}{|\Pi(\theta^n)|}\sum_{\pi\in\Pi(\theta^n)} h(\pi s^n,\pi \theta^n)\right] \leq \alpha_n \,.
\end{align}
Thus,
\begin{align}
\sum_{s^n \,:\; \hP_{s^n|\theta^n}=\hq}  \hq^{\;n}(s^n|\theta^n) \left[\frac{1}{|\Pi(\theta^n)|}\sum_{\pi\in\Pi(\theta^n)} h(\pi s^n,\pi \theta^n)\right] \leq \alpha_n \,.
\end{align}
As the expression in the square brackets is identical for all sequences $s^n$ of conditional type $\hq$, we have that 
\begin{align}
\label{eq:rtineq1}
\left[\frac{1}{|\Pi(\theta^n)|}\sum_{\pi\in\Pi(\theta^n)} h(\pi \widetilde{s}^{\;n}, \pi \theta^n)\right]\cdot
\sum_{s^n \,:\; \hP_{s^n|\theta^n}=\hq}  \hq^{\;n}(s^n|\theta^n)  \leq \alpha_n \,.
\end{align}
The second sum is the probability of 
 the conditional type class of $\hq$, hence
\begin{align}
\label{eq:rtineq2}
\sum_{s^n \,:\; \hP_{s^n|\theta^n}=\hq}  \hq^{\;n}(s^n|\theta^n) \geq \frac{1}{(n+1)^{|\Sset| |\Tset|}} \,,
\end{align}
by \cite[Theorem 11.1.4]{CoverThomas:06b}. The proof follows from (\ref{eq:rtineq1}) and (\ref{eq:rtineq2}). 
\qed

\section{Proof of Theorem~\ref{theo:PrCav}}
\label{app:PrCav}
Consider the AVC $\avc$ with fixed parameters under input constraint $\plimit$ and state constraint $\Lambda$.

\subsection{Achievability Proof}
To prove the random code capacity theorem for the AVC with fixed parameters, we use our result on the compound channel along with
  our modified Robustification Technique (RT), \ie Lemma~\ref{lemm:LRT}. 

Let $R<\rICav$.
At first, we consider the compound channel under input constraint $\plimit$, with $\Qset=\pLSpaceS$. 
According to Lemma~\ref{lemm:PCcompound},  
for some $\delta>0$ and sufficiently large $n$,   there exists a  $(2^{nR},n)$  code 
$\code=(\enc(m,\theta^n),$ $\dec(y^n,\theta^n))$ for the compound channel $\avc^{\pLSpaceS}$ with fixed parameters such that 
\begin{align}
\label{eq:PLrAVcosti}
&
\cost^n(\enc(m,\theta^n))  \leq\plimit \,,\; 
\text{for all $m\in [1:2^{nR}]$}\,,
\end{align}
and
\begin{align}
\label{eq:PLrAVerrDirect}
&  \err(q,\theta^n,\code)=\sum_{s^n\in\Sset^n} q(s^n|\theta^n) \err(\code|s^n,\theta^n)  \leq e^{-2\delta n} \,,
\end{align}
for all product state distributions $q(s^n|\theta^n)=\prod_{i=1}^n q(s_i|\theta_i)$, with $q\in\apLSpaceS$.

Therefore, by Lemma~\ref{lemm:LRT}, taking $h_0(s^n,\theta^n)=\err(\code|s^n,\theta^n)$ and $\alpha_n=e^{-2\delta n}$, we have that for a sufficiently large $n$,
\begin{align}
\label{eq:MALdetErrC}
\frac{1}{|\Pi(\theta^n)|} \sum_{\pi\in\Pi(\theta^n)}  \err(\code|\pi s^n,\theta^n)\leq (n+1)^{|\Sset|}e^{-2\delta n} 
\leq e^{-\delta n}  \,,
\end{align}
for all $s^n\in\Sset^n$ with $l^n(s^n)\leq\Lambda$, where the sum is over the set of all $n$-tuple permutations such that 
$\pi \theta^n=\theta^n$. 

On the other hand, for every  $\pi\in\Pi(\theta^n)$,
\begin{align}
\err(\code|\pi s^n,\theta^n) 
  \stackrel{(a)}{=}&
\frac{1}{2^{ nR }}\sum_{m=1}^{2^{nR}}
\sum_{y^n:\dec(y^n,\theta^n)\neq m}  W_{Y^n|X^n,S^n,T^n}(y^n|\enc(m,\theta^n),\pi s^n,\theta^n) \nonumber\\
\stackrel{(b)}{=}& \frac{1}{2^{ nR }}\sum_{m=1}^{2^{nR}}
\sum_{y^n:\dec(\pi y^n, \theta^n)\neq m}  W_{Y^n|X^n,S^n,T^n}(\pi y^n|\enc(m,\theta^n),\pi s^n,\theta^n) \nonumber\\
\stackrel{(c)}{=}& \frac{1}{2^{ nR }}\sum_{m=1}^{2^{nR}}
\sum_{y^n:\dec(\pi y^n, \theta^n)\neq m}  W_{Y^n|X^n,S^n,T^n}( y^n|\pi^{-1}\enc(m,\theta^n), s^n,\pi^{-1}\theta^n) \,,
\label{eq:MLcerrpi1}
\end{align}
where $(a)$ is obtained by plugging 
 $\pi s^n$  in (\ref{eq:Pcerr});
in $(b)$ we substitue $\pi y^n$ instead of $y^n$; and $(c)$ holds because the channel is memoryless. 
Since $\pi \theta^n=\theta^n$ for every $\pi\in\Pi(\theta^n)$, it follows that 
\begin{align}
\err(\code|\pi s^n,\theta^n) =
 \frac{1}{2^{ nR }}\sum_{m=1}^{2^{nR}}
\sum_{y^n:\dec(\pi y^n,\theta^n)\neq m}  W_{Y^n|X^n,S^n,T^n}( y^n|\pi^{-1}\enc(m,\theta^n), s^n,\theta^n) \,.
\label{eq:MLcerrpi}
\end{align}

Then, consider the $(2^{nR},n)$ random code $\code^{\Pi(\theta^n)}$, specified by 
\begin{align}
\label{eq:MLCpi}
&f_{\pi}^n(m,\theta^n)= \pi^{-1} \enc(m,\theta^n) \,,\; 
 g_\pi(y^n,\theta^n)=\dec(\pi y^n,\theta^n)   \,,
\end{align}
with a uniform distribution $\mu(\pi)=\frac{1}{|\Pi(\theta^n)|}$ for $\pi\in\Pi(\theta^n)$. 
As the inputs cost is additive (see (\ref{eq:PLInConstraintStrict})), the permutation does not affect the costs of the codewords, hence the random code satisfies the input constraint $\plimit$.
 From (\ref{eq:MLcerrpi}), we see that 
 $
\err(\code^{\Pi(\theta^n)}|s^n,\theta^n)=\sum_{\pi\in\Pi(\theta^n)} \mu(\pi) \cdot  \err(\code|\pi s^n,\theta^n) 
$,
for all $s^n\in\Sset^n$ with $l^n(s^n)\leq \Lambda$. Therefore, together with (\ref{eq:MALdetErrC}), we have that the probability of error of the random code $\code^{\Pi(\theta^n)}$ is bounded by 
$
\err(\qn,\theta^n,\code^{\Pi(\theta^n)})\leq e^{-\delta n} 
$, 
for every $\qn(s^n|\theta^n)\in\pLSpaceSn$. 
It follows that $\code^{\Pi(\theta^n)}$ is a $(2^{nR},n,e^{-\delta n})$ random 
 code for the AVC $\avc$ with fixed parameters  under input constraint $\plimit$ and state constraint $\Lambda$. 
\qed

\subsection{Converse Proof}
Assume to the contrary that there exists an achievable rate pair 
\begin{align}
R>\inC(\compound) \big|_{\Qset=\overline{\pSpace}_{\Lambda-\delta}(\Sset|\theta^{\infty})} \,,
\label{eq:MrCconverseRate}
\end{align}
using random codes over the AVC $\avc$  under input constraint $\plimit$ and state constraint $\Lambda$, where $\delta>0$ is arbitrarily small. 
That is, for every $\eps>0$ and sufficiently large $n$,
there exists a $(2^{nR},n)$ random code $\code^\Gamma=(\mu,\Gamma,\{\code_\gamma\}_{\gamma\in\Gamma})$ for the AVC $\avc$, such that 
 $\sum_{\gamma\in\Gamma} \mu(\gamma)  \cost^n(\enc_{\gamma}(m,\theta^n))  \leq\plimit$, and 
\begin{align}
& \err(q,\theta^n,\code^\Gamma)\leq\eps \,,
\label{eq:MStateConverse1b}
\end{align}
 for all $m\in [1:2^{nR}]$ and $q(s^n|\theta^n)\in\pLSpaceSn$. 
	In particular, for  distributions $q(\cdot|\theta^n)$ that give mass $1$ to some sequence $s^n\in\Sset^n$ with $l^n(s^n)\leq\Lambda$, we have that
	$
	\err(\code^\Gamma|s^n,\theta^n)\leq\eps 
	$. 
	
	Consider using the random code $\code^\Gamma$ over the compound channel $\avc^{\overline{\pSpace}_{\Lambda-\delta}(\Sset)}$ with fixed parameters under input constraint $\plimit$. Let $\oq(s|t)\in\overline{\pSpace}_{\Lambda-\delta}(\Sset)$ be a given state distribution. Then, 
	 define a sequence of conditionally independent random variables $\oS_1,\ldots,\oS_n\sim \oq(s|t)$.  
	Letting 
	$\oq^n(s^n|\theta^n)\triangleq\prod_{i=1}^n \oq(s_i|\theta_i)$, the probability of error is bounded by
	\begin{align}
	\err(\oq,\theta^n,\code^\Gamma)
	\leq 
	\sum_{s^n\,:\; l^n(s^n)\leq\Lambda} \oq^{n}(s^n|\theta^n) \err(\code^\Gamma|s^n,\theta^n)
	+\prob{l^n(\oS^{n})>\Lambda} . 
	\end{align}
	The first sum is bounded by (\ref{eq:MStateConverse1b}), and the second term vanishes 
	by the law of large numbers, since
	$\oq\in\overline{\pSpace}_{\Lambda-\delta}(\Sset|\theta^{\infty})$. 
	It follows that the random code $\code^\Gamma$ achieves a rate 
	$R$ as in (\ref{eq:MrCconverseRate}) 
	over the compound channel $\avc^{\overline{\pSpace}_{\Lambda-\delta}(\Sset)}$ with fixed parameters under  input constraint 
	$\plimit$, for an arbitrarily small $\delta>0$, in contradiction to Lemma~\ref{lemm:PCcompound}. 
	We deduce that the assumption is false, and $\rCav\leq  \inC(\compound) \big|_{\Qset=
	\overline{\pSpace}_{\Lambda}(\Sset|\theta^{\infty})}=\inC_n^{\rstarC}\hspace{-0.05cm}(\avc)$.
\qed

\section{Proof of Lemma~\ref{lemm:Ciequiv}}
\label{app:Ciequiv}
To prove that $\inR_{n}^{\rstarC}\hspace{-0.05cm}(\avc)=\inC_{n}^{\rstarC}\hspace{-0.05cm}(\avc)$, we begin with the  property in the lemma below.
\begin{lemma}
\label{lemm:Psamet}
Let $\omega_i^*$, $\lambda_i^*$, $i\in [1:n]$, be the parameters that achieve the saddle point in  (\ref{eq:Cieqiv2}), \ie
\begin{align}
\inR_{n}^{\rstarC}\hspace{-0.05cm}(\avc)= 
\frac{1}{n} \sum_{i=1}^n  
\inC_{\theta_i}(\omega_i^*,\lambda_i^*) \,.
\label{eq:Cieqiv21}
\end{align}
Then, for every $i,j\in [1:n]$ such that $\theta_i=\theta_j$, we have that $\omega_i^*=\omega_j^*$ and $\lambda_i^*=\lambda_j^*$.
\end{lemma}

\begin{proof}[Proof of Lemma~\ref{lemm:Psamet}]
For every $i\in [1:n]$, let $p_i,q_i$ denote input and state distributions such that $\E \cost(X_i)\leq \omega_i^*$, $\E l(S_i)\leq \lambda_i^*$ for $X_i\sim p_i$, $S_i\sim q_i$. Now, suppose that $\theta_i=\theta_j=t$, and define
\begin{align}
p'(x)=\frac{1}{2}[ p_i(x)+p_j(x) ] \,,\; q'(s)=\frac{1}{2}[ q_i(s)+q_j(s) ] \,.
\end{align}
Then, $\E \cost(X') = \frac{1}{2}[ \E \cost(X_i)+ \E \cost(X_j) ]$ and $\E l(S') = \frac{1}{2}[ \E l(S_i)+ \E l(S_j) ]$ 
for $X'\sim p'$, $S'\sim q'$. Furthermore, since the mutual information is concave-$\cap$ in the input distribution and convex-$\cup$ in the state distribution, we have that
\begin{align}
& \frac{1}{2}\left[ I_{q'}(X_i;Y_i|T_i=t)+I_{q'}(X_j;Y_j|T_j=t) \right] \leq I_{q'}(X';Y'|T'=t) \nonumber\\
& \frac{1}{2}\left[ I_{q_i}(X';Y'|T'=t)+I_{q_j}(X';Y'|T=t) \right] \geq I_q(X';Y'|T'=t) \,.
\end{align}
Therefore, the saddle point distributions must satisfy $p_i=p_j=p'$ and $q_i=q_j=q'$, hence $\omega_i^*=\omega_j^*$ and $\lambda_i^*=\lambda_j^*$.
\end{proof}

Next, it can be inferred from Lemma~\ref{lemm:Psamet} that
\begin{align}
\inR_{n}^{\rstarC}\hspace{-0.05cm}(\avc)=& \min_{ \substack{ (\lambda_t)_{t\in\Tset} \,:\; \\  \sum_{t\in\Tset} P_T(t) \lambda_t \leq \Lambda } } \;
\max_{ \substack{ (\omega_t)_{t\in\Tset} \,:\;\\   \sum_{t\in\Tset} P_T(t) \omega_t\leq \plimit } }
\sum_{t\in\Tset} P_T(t)  
\inC_{t}(\omega_t,\lambda_t)
\nonumber\\
=& \min_{ \substack{ (\lambda_t)_{t\in\Tset} \,,\; q(s|t) \,:\; \\  
\E_q[ l(S)|T=t ]\leq \lambda_t \\
\sum_{t\in\Tset} P_T(t) \lambda_t \leq \Lambda } }
\max_{ \substack{ (\omega_t)_{t\in\Tset} \,,\; p(x|t) \,:\; \\  
\E[ \cost(X)|T=t ]\leq \omega_t \\
\sum_{t\in\Tset} P_T(t) \omega_t \leq \plimit } } I_q(X;Y|T)
\nonumber\\
=& \min_{  q(s|t) \,:\; \E_q l(S) \leq \Lambda  }
\max_{  p(x|t) \,:\; \E \cost(X)\leq \plimit  } I_q(X;Y|T) =\inC_{n}^{\rstarC}\hspace{-0.05cm}(\avc)
 \,,
\label{eq:Cieqiv31}
\end{align}
where $P_T$ is the type of the parameter sequence $\theta^n$. The second equality follows from the definition of 
$\inC_{t}^{\rstarC}\hspace{-0.05cm}(\omega_t,\lambda_t)$ in (\ref{eq:Ctol}), using the minimax theorem \cite{sion:58p} to switch between the order of the minimum and maximum.
In the third line, we eliminate the slack variables $\lambda_i$ and $\omega_i$ replacing $\E_q l(S_i)$ and $\E\cost(X_i)$, respectively.
The last equality holds by the definition of $\inC_{n}^{\rstarC}\hspace{-0.05cm}(\avc)$ in (\ref{eq:Cieqiv3}).
\qed

\section{Proof of Lemma~\ref{lemm:disDec}}
\label{app:disDec}
Consider the AVC $\avc$ with fixed parameters under input constraint $\plimit$ and state constraint $\Lambda$. Let $\theta^n$ be sequence of fixed parameters for a 
given blocklength. Recall that  $T$ is a random variable that is distributed as the type of $\theta^n$.
We extend the proof in \cite{CsiszarNarayan:88p}.
First, we give an auxiliary lemma, which we also used in \cite{PeregSteinberg:19p4}.
\begin{lemma}[See {\cite
{CsiszarNarayan:88p} 
\cite[Lemma 11]{PeregSteinberg:19p4}} ]
\label{lemm:A2}
For every pair of conditional state distributions $Q(s|x,t)$ and $Q'(s|x,t)$  such that 
\begin{align}
\max\left\{
\sum_{t,x,s} P_T(t)  p(x|t)Q(s|x,t)l(s)  \,,\; 
\sum_{t,x,s} P_T(t)  p(x|t)Q'(s|x,t)l(s) 
\right\}
<\tLambda_n(p) \,,\; 
\label{eq:A2sump} 
\end{align}
 there exists $\xi>0$ such that
\begin{align}
\max_{x,\tx,y} \Big|\sum_{t,s} P_T(t) Q(s|\tx,t) W_{Y|X,S,T}(y|x,s,t) -\sum_{t,s} P_T(t) Q'(s|x,t) W_{Y|X,S,T}(y|\tx,s,t) \Big|
\geq \xi \,.
\label{eq:decCNresA2}
\end{align}
\end{lemma}

\begin{proof}[Proof of Lemma~\ref{lemm:A2}]
Assume to the contrary that the LHS in (\ref{eq:decCNresA2}) is zero, and
define
\begin{align}
Q_A(s|x,t)=\frac{1}{2}\left(Q(s|x,t)+Q'(s|x,t)\right) \,.
\end{align}
Using the symmetry between $Q$ and $Q'$,  we have that 
\begin{align}
0=&\max_{x,\tx,y} \Big|\sum_{t\in\Tset} \sum_{s\in\Sset} P_T(t)  Q(s|\tx,t)W_{Y|X,S,T}(y|x,s,t)
 -\sum_{t\in\Tset} \sum_{s\in\Sset}P_T(t)  Q'(s|x,t)W_{Y|X,S,T}(y|\tx,s,t) \Big|
\nonumber\\
=&\frac{1}{2}\max_{x,\tx,y} \Big|\sum_{t\in\Tset_n} \sum_{s\in\Sset} P_T(t)  Q(s|\tx,t)W_{Y|X,S,T}(y|x,s,t) 
-\sum_{t\in\Tset_n} \sum_{s\in\Sset} P_T(t) Q'(s|x,t)W_{Y|X,S,T}(y|\tx,s,t)\Big|
\nonumber\\
&+\frac{1}{2}\max_{x,\tx,y} \Big|\sum_{t\in\Tset_n} \sum_{s\in\Sset} P_T(t)  Q'(s|\tx,t)W_{Y|X,S,T}(y|x,s,t) 
-\sum_{t\in\Tset_n} \sum_{s\in\Sset} P_T(t)  Q(s|x,t)W_{Y|X,S,T}(y|\tx,s,t)\Big|
\nonumber\\
\geq&   \max_{x,\tx,y} \Big|\sum_{t\in\Tset_n} \sum_{s\in\Sset} P_T(t) Q_A(s|x,t)W_{Y|X,S,T}(y|\tx,s,t) 
-\sum_{t\in\Tset_n} \sum_{s\in\Sset} P_T(t)  Q_A(s|\tx,t) W_{Y|X,S,T}(y|x,s,t) \Big| \,.
\end{align}
Since we have assumed that $P_T(t)>\delta_0$ for all $t\in\Tset$, it follows that
\begin{align}
\sum_{s\in\Sset} Q_A(s|x,t)W_{Y|X,S,T}(y|\tx,s,t) =\sum_{s\in\Sset}  Q_A(s|\tx,t) W_{Y|X,S,T}(y|x,s,t) \,,
\end{align}
 for all $t\in\Tset$, $x,\tx\in\Xset$ and $y\in\Yset$. In other words, $Q_A(\cdot|\cdot,t)$ 
symmetrizes the channel $W_{Y|X,S,T}(\cdot|\cdot,\cdot,t)$ for all $t\in\Tset$. 
Therefore, by the definition of $\tLambda_n(p)$ in (\ref{eq:LambdaOig}), we have that
\begin{align}
\sum_{t,x,s} P_T(t)  p(x|t)Q_A(s|x,t)l(s)= \frac{1}{n} \sum_{i=1}^n \sum_{x,s}  p(x|\theta_i)Q_A(s|x,\theta_i)l(s)  \geq & \tLambda_n(p)
\end{align}
in contradiction to  (\ref{eq:A2sump}). The equality above holds because $T$ is distributed as the type of the parameter sequence $\theta^n$, hence averaging over time is the same as averaging according to $P_T$.
 It follows that 
the LHS of  (\ref{eq:decCNresA2}) must be positive.
This completes the proof of the auxiliary Lemma.
\end{proof}

We move to the main part of the proof.
To show that (\ref{eq:disDec}) holds for sufficiently small $\eta$, assume to the contrary that there exists $y^n$ such that 
$(y^n,\theta^n)$ is in $\Dset(m)\cap\Dset(\tm)\neq\emptyset$.
By the assumption in the lemma, the codewords $\{\enc(m,\theta^n)\}_{m\in [1:2^{nR}]}$  have the same conditional type. In particular, 
$P_{\tX|T}=P_{X|T}=p$. 

By Condition 1) of the decoding rule,
\begin{align}
&D(P_{T,X,S,Y}|| P_T\times  P_{X|T}\times P_{S|T} \times W_{Y|X,S,T}) \nonumber \\
=&\sum_{t,x,s,y} P_{T,X,S,Y}(t,x,s,y) 
  \log \frac{P_{T,X,S,Y}(t,x,s,y)}{ P_T(t)  p(x|t)  P_{S|T}(s|t)  W_{Y|X,S,T}(y|x,s,t)}  \leq \eta \,,
\label{eq:Drule1q}
\end{align}
and by Condition 2) of the decoding rule,
\begin{align}
I(X,Y;\tX|S,T) 
=& \sum_{t,x,\tx,s,y} P_{T,X,\tX,S,Y}(t,x,\tx,s,y)
\log 
\frac{P_{\tX|X,S,T,Y}(\tx|x,s,t,y)}{P_{\tX|S,T}(\tx|s,t)} 
\leq \eta \,,
\label{eq:Drule2bq}
\end{align}
where $T,X,\tX,S,Y$ are distributed according to the joint type of 
$\theta^n$, $f^n(m,\theta^n)$, $f^n(\tm,\theta^n)$, $s^n$, and $y^n$. 
 Adding (\ref{eq:Drule1q}) and (\ref{eq:Drule2bq}) yields
\begin{align}
\sum_{t,x,\tx,s,y} P_{T,X,\tX,S,Y}(t,x,\tx,s,y) \log 
\frac{P_{T,X,\tX,S,Y}(t,x,\tx,s,y)}{ P_T(t) p(x|t) P_{\tX,S|T}(\tx,s|t) W_{Y|X,S,T}(y|x,s,t)}  \leq 2\eta \,.
\end{align}
That is, $D(P_{T,X,\tX,S,Y}||  P_T \times p \times p \times P_{S|\tX,T}\times  W_{Y|X,S,T} )\leq 2\eta$. Therefore, by the log-sum inequality (see \eg \cite[Theorem 2.7.1]{CoverThomas:06b}),  
\begin{align}
&D(P_{T,X,\tX,Y}|| P_T \times p \times p  \times  V_{Y|X,\tX,T}  ) \nonumber\\
\leq& D(P_{T,X,\tX,S,Y}|| P_T\times  p\times p \times P_{S|\tX,T}\times  W_{Y|X,S,T} )\leq 2\eta
\,,
\end{align}
where $V_{Y|X,\tX,T}(y|x,\tx,t)=\sum_{s\in\Sset} W_{Y|X,S,T}(y|x,s,t)P_{S|\tX,T}(s|\tx,t)$.
Then, by Pinsker's inequality (see \eg \cite[Problem 3.18]{CsiszarKorner:82b}),
\begin{align}
\sum_{t,x,\tx,y} |P_{T,X,\tX,Y}(t,x,\tx,y)  -
P_T(t)  p(x|t)
 p(\tx|t)    V_{Y|X,\tX,T}(y|x,\tx,t)|   \leq c\sqrt{2\eta} \,,
\end{align}
where $c>0$ is a constant. By the same arguments, (\ref{eq:DcompA}) implies that
\begin{align}
\sum_{t,x,\tx,y} |P_{T,X,\tX,Y}(t,x,\tx,s)-
P_T(t) p(x|t)
 p(\tx|t)    V_{Y|X,\tX,T}'(y|x,\tx,t)|  \leq c\sqrt{2\eta} \,,
\end{align}
where $ V_{Y|X,\tX,T}'(y|x,\tx,t)=\sum_{s\in\Sset} W_{Y|X,S,T}(y|\tx,s,t)P_{\tS|X,T}(s|x,t)$. Now,
observe that inserting the sum over $t\in\Tset$ into the absolute value maintains the inequality, by the triangle inequality.
Furthermore, since $p(x|t)>\delta_1$, for  $x\in\Xset$, $t\in\Tset$,  we have that
\begin{align}
\max_{x,\tx,y} 
\Big| \sum_{t\in\Tset_n} P_T(t) V_{Y|X,\tX,T}(y|x,\tx,t) - \sum_{t\in\Tset_n} P_T(t) V_{Y|X,\tX,T}'(y|x,\tx,t) \Big| 
\leq \frac{2c\sqrt{2\eta}}{\delta^2} 
\,,
\label{eq:Vlowd}
\end{align}
Equivalently, the above can be expressed as
\begin{align}
\max_{x,\tx,y} 
\Big|\sum_{t,s} P_T(t) P_{S|\tX,T}(s|\tx,t) W_{Y|X,S,T}(y|x,s,t) 
-\sum_{t,s} P_T(t) P_{\tS|X,T}(s|x,t) W_{Y|X,S,T}(y|\tx,s,t) \Big| 
\leq
 \frac{2c\sqrt{2\eta}}{\delta_1^2 
} \,,
\label{eq:Vlowd1}
\end{align}

Now, we show that the state distributions $Q=P_{S|\tX,T}$ and $Q'=P_{\tS|X,T}$ satisfy the conditions of 
Lemma~\ref{lemm:A2}. Indeed, 
\begin{align}
&\max\bigg\{ \sum_{t,\tx,s} P_T(t) p(\tx|t)Q(s|\tx)l(s) ,\, \sum_{tx,s} P_T(t)  p(x|t)Q'(s|x)l(s) \bigg\} 
\nonumber\\
=& \max\bigg\{ \sum_{t,\tx,s} P_T(t)  p(\tx|t)P_{S|\tX,T}(s|\tx,t)l(s) ,\,
 \sum_{t,x,s} P_T(t)  p(x|t)P_{\tS|X,T}(s|x,t)l(s) \bigg\} 
\nonumber\\
=&\max\left\{ \sum_{s} P_S(s)l(s) ,\, \sum_{s} P_{\tS}(s)l(s)  \right\}
\nonumber\\
=&\max\left\{ l^n(s^n),\, l^n(\ts^n) \right\}\leq  \Lambda< \tLambda_n(p) \,,
\end{align}
where the last inequality is due to (\ref{eq:decLambda}). 
Thus, there exists $\xi>0$ such that (\ref{eq:decCNresA2}) holds with $Q=P_{S|\tX,T}$ and $Q'=P_{\tS|X,T}$, 
which contradicts (\ref{eq:Vlowd1}), if $\eta$ is sufficiently small such that 
$\frac{2c\sqrt{2\eta}}{\delta^2 }<\xi$.
\qed

\section{Proof of Lemma~\ref{lemm:codeBsets}}
\label{app:codeBsets}
Let $Z^n(m,\theta^n)$,  $m\in [1:2^{nR}]$, be statistically independent sequences, uniformly distributed over the conditional type class 
$\Tset^n(p)$.
Fix $a^n\in\Xset^n$ and $s^n\in\Sset^n$, and consider a joint type 
$P_{T,X,\tX,S}$, such that 
$P_{X|T}=P_{\tX|T}=p$.
We intend to show that $\{Z^n(m,\theta^n)\}$ satisfy each of the desired properties with double exponential high probability 
$(1-e^{-2^{\dE n}})$, $\dE>0$, implying that there exists a deterministic codebook that satisfies (\ref{eq:11ebn})-(\ref{eq:13ebn}) simultaneously.
 We begin with the following large deviations result by Csis\'{a}r and Narayan 
\cite{CsiszarNarayan:88p}. 
\begin{lemma}[see {\cite[Lemma A1]{CsiszarNarayan:88p}}]
\label{lemm:bookLD}
Let $\alpha,\beta\in [0,1]$, and consider a sequence of random vectors $ U^n(m)$, and functions $\varphi_m: \Xset^{nm}\rightarrow [0,1]$, for $m\in [1:\dM]$. If
\begin{align}
\E \left( \varphi_m( U^n(1)\,\ldots, U^n(m) ) \big|  U^n(1)\,\ldots,  U^n(m-1) \right) \leq \alpha \;\;\text{ a.s., }\; \text{for
$m\in [1:\dM]$ } \,,
\end{align}
then
\begin{align}
\prob{\sum_{m=1}^{\dM} \varphi_m(  U^n(1)\,\ldots,  U^n(m) )>\dM \beta  }\leq 
 \exp\{ -\dM (\beta-\alpha\log e) \} \,.
\end{align}
\end{lemma}

To show that (\ref{eq:11ebn}) holds, consider the indicator
\begin{align}
\varphi_{m}(Z^n(1,\theta^n),\ldots,Z^n(m,\theta^n))=
\begin{cases}
1 &\text{if $(\theta^n, Z^n(m,\theta^n), Z^n(\tm,\theta^n),s^n)\in\Tset^n(P_{T,X,\tX,S})$}
\\&\text{for some $\tm<m$ } 
\\
0 &\text{otherwise}
\end{cases}
\label{eq:indJ1}
\end{align}
By standard type class considerations (see \eg \cite[Theorem 1.3]{Kramer:08n}), we have that
\begin{align}
\E \left[ \varphi_{m}(Z^n(1,\theta^n),\ldots,Z^n(m,\theta^n))  \big|
Z^n(1,\theta^n),\ldots,Z^n(m-1,\theta^n) \right]
 \leq&  
2^{-n\left(I(\tX;T,X,S)-\frac{\eps}{4}-R \right)}
\leq 2^{-n\left(I(\tX;X,S|T)-\frac{\eps}{4}-R \right)}
 \,,
\label{eq:phiZb11}
\end{align}
where the last inequality holds since $I(\tX;T,X,S)\geq I(\tX;X,S|T)$.

Next, we use Lemma~\ref{lemm:bookLD}, and plug
\begin{align}
& \dM=2^{nR} \,,\; U^n(m)=Z^n(m,\theta^n) \,,
\nonumber\\
&
 \alpha= 2^{-n\left(I(\tX;X,S|T)-\frac{\eps}{4}-R \right)}  \,,\;   \beta=2^{n\left( \left[ R-I(\tX;X,S|T) \right]_{+} -R+\eps \right)} \,.
\end{align}
For sufficiently large $n$, we have that $\dM(\beta-\alpha\log e)\geq 2^{n\eps/2}$. Hence, by 
Lemma~\ref{lemm:bookLD},
\begin{align}
&\prob{ \sum_{m=1}^{2^{nR}} \varphi_{m}(Z^n(1,\theta^n),\ldots,Z^n(2^{nR},\theta^n))
>2^{n\left( \left[ R-I(\tX;X,S|T) \right]_{+} +\eps \right)  }} 
\leq 
 e^{- 2^{n\eps/2}}  \,.
\label{eq:dLm1n}
\end{align}
By the symmetry between $m$ and $\tm$ in the derivation above, the double exponential decay of the probability in (\ref{eq:dLm1n}) implies that there exists a codebook that satisfies (\ref{eq:11ebn}).

Similarly, to show  (\ref{eq:12ebn}), we replace the indicator of the type $P_{X,\tX,S|T}$ in (\ref{eq:indJ1}) by an indicator of the type $P_{\tX,S|T}$, and rewrite (\ref{eq:phiZb11})
with 
$I(\tX;S|T)$,
 to obtain
\begin{align}
\Pr\Big( |\{ \tm \,:\; (\theta^n,Z^n(\tm,\theta^n),s^n)\in\Tset^n(P_{T,\tX,S}) \}|>
2^{n\left( \left[ R-I(\tX;S|T) \right]_{+} +\eps_1 \right)}
 \Big) \;<  e^{- 2^{n\eps_1/2}}  \,,
\label{eq:LTxs1}
\end{align}
where $\eps_1>0$ is arbitrarily small.
If $I(\tX;S|T)>\eps$ and $R\geq\eps$, then choosing $\eps_1=\frac{\eps}{2}$, we have that
\begin{align}
\left[ R-I(\tX;S|T) \right]_{+}+\eps_1 \leq R-\frac{\eps}{2} \,,
\end{align}
 hence,
\begin{align}
\Pr\Big( |\{ \tm \,:\; (\theta^n,Z^n(\tm,\theta^n),s^n)\in\Tset^n(P_{T,\tX,S}) \}|>
2^{n\left(  R- \frac{\eps}{2} \right)}
 \Big) \;
<
e^{- 2^{n\eps/4}}  \,.
\end{align}

It remains to show that (\ref{eq:13ebn}) holds. Assume that
\begin{align}
I(X;\tX,S|T)-\left[ R-I(\tX;S|T) \right]_{+}>\eps \,.
\label{eq:encJassump1}
\end{align}
 Let $\Jset_m$ denote the set of indices 
$\tm<m$ such that $(\theta^n,Z^n(\tm,\theta^n),s^n)\in\Tset^n(P_{T,\tX,S})$, provided that their number does not exceed 
$2^{n\left(\left[ R-I(\tX;S|T) \right]_{+} +\frac{\eps}{8} \right)}$; else, let $\Jset_m=\emptyset$. Also, let
\begin{align}
\psi_m(Z^n(1,\theta^n),\ldots,Z^n(m,\theta^n))=\begin{cases}
1 &\text{if $(\theta^n,Z^n(m,\theta^n),Z^n(\tm,\theta^n),s^n)\in\Tset^n(P_{T,X,\tX,S})$}\\
  &\text{for some $\tm\in\Jset_m$}\,, \\
0 &\text{otherwise.}
\end{cases}
\end{align}
Then, choosing $\eps_1=\frac{\eps}{8}$ in (\ref{eq:LTxs1}) yields
\begin{align}
&\Pr \Big( \sum_{m=1}^{2^{nR}} \psi_m(Z^n(1,\theta^n),\ldots,Z^n(m,\theta^n))\neq\; 
|\{ m \,:\;  \nonumber\\&
 (\theta^n,Z^n(m,\theta^n),Z^n(\tm,\theta^n),s^n)\in\Tset^n(P_{T,X,\tX,S}) \;\text{for some $\tm<m$} \}|    
\Big) \;<  e^{- 2^{n\eps/16}} \,.
\label{eq:setsEquiv1}
\end{align}
Therefore, instead of bounding the set of messages, it is sufficient to consider the sum $\sum
\psi_m(Z^n(1,\theta^n),\ldots,Z^n(m,\theta^n))$.
Furthermore, by standard type class considerations (see \eg \cite[Theorem 1.3]{Kramer:08n}), we have that
\begin{align}
&\E \left( \psi_m(Z^n(1,\theta^n),\ldots,Z^n(m,\theta^n)) \big| Z^n(1,\theta^n),\ldots,Z^n(m-1,\theta^n) \right) \leq 
|\Jset_m|\cdot 2^{-n\left(I(X;\tX,S|T)-\frac{\eps}{8} \right)}
\nonumber\\
\leq& 2^{n\left(\left[ R-I(\tX;S|T) \right]_{+}-I(X;\tX,S|T)+\frac{\eps}{4} \right)}
< 2^{-3n\eps/4} \,,
\end{align}
where the last inequality is due to (\ref{eq:encJassump1}). Thus, by Lemma~\ref{lemm:bookLD},
\begin{align}
\prob{ \sum_{m=1}^{2^{nR}} \psi_m(Z^n(1,\theta^n),\ldots,Z^n(m,\theta^n))>
2^{n\left( R-\frac{\eps}{2} \right)} }<
e^{-2^{n\left(R-\frac{3\eps}{4}  \right)}}\leq e^{-2^{n\eps/4}} \,, 
\label{eq:b1j}
\end{align}
as we have assumed that $R\geq \eps$.
Equations (\ref{eq:setsEquiv1}) and (\ref{eq:b1j}) imply that the property in (\ref{eq:13ebn}) holds with double exponential probability $1-e^{-2^{\dE_1 n}}$, where $\dE_1>0$. 
\qed

\section{Proof of  Theorem~\ref{theo:PCavc}}
\label{app:PCavc}

\subsection{Achievability Proof}
Suppose that $L_n^*>\Lambda$ for sufficiently large $n$. 
%
Let $\eps>0$ be chosen later, and let $P_{X|T}$ be a conditional type over $\Xset$, for which $P_{X|T}(x|t)>0$ $\forall x\in\Xset$, 
$t\in\Tset$, and   
$\E \cost(X)\leq \plimit$,  with
\begin{align}
\tLambda_n(P_{X|T})>&\Lambda \,.
\end{align} 
As explained below, we may assume without loss of generality that for some $\delta_0>0$ that does not depend on $n$, we have that $P_T(t)>\delta_0$ for all $t\in\Tset$. 
Indeed, following our assumption in (\ref{eq:Symmassumption}), the asymptotic capacity formula $\liminf
\inC_n(\avc)$ does not change when we remove parameter values $t\in\Tset$ such that $P_T(t)\rightarrow 0$. Hence, coding can be limited  to the rest of the block with negligible rate decrease, thus  removing those parameters from consideration. 
Then, choose $\eta>0$ to be sufficiently small such that Lemma~\ref{lemm:disDec} guarantees that the decoder in Definition~\ref{def:Ldecoder} is well defined.
%
Now, Lemma~\ref{lemm:codeBsets} assures that there is a codebook $\{x^n(m,\theta^n)\}_{m\in [1:2^{nR}]}$ of conditional type $p$ that satisfies 
(\ref{eq:11ebn})-(\ref{eq:13ebn}).
Consider the following coding scheme.

\emph{Encoding}: To send $m\in [1:2^{nR}]$,  transmit $x^n(m,\theta^n)$.

\emph{Decoding}: Find a unique message $\hm$ such that $(y^n,\theta^n)$ 
belongs to $\Dset(\hm)$, as in Definition~\ref{def:Ldecoder}. If there is none, declare an error.
Lemma~\ref{lemm:disDec} guarantees that there cannot be two messages for which this holds.

\emph{Analysis of Probability of Error}: Fix $s^n\in\Sset^n$ with $l^n(s^n)\leq\Lambda$, let $q=P_{S|T}$ denote the conditional type of $s^n$ given $\theta^n$, and let 
$M$ denote the transmitted message.
Consider the error events
\begin{align}
\Eset_{1}=&\{ D(P_{T,X,S,Y}|| P_T\times P_{X|T}\times P_{S|T} \times \channel)> \eta \}
\\
\Eset_{2}=&\{ \text{Condition 2) of the decoding rule is violated} \}
\end{align}
and 
\begin{align}
\Fset_1=&\{ I_q(X;S|T)>\eps \} \,, \\
\Fset_2=&\{ I_q(X;\tX,S|T)>\left[ R-I(\tX;S|T) \right]_{+}+\eps \,,
\; 
\text{for some $\tm\neq M$}
\} \,,
\end{align}
where $(T,X,\tX,S)$ are dummy random variables, which are distributed as the joint type 
of $(\theta^n,x^n(M,\theta^n),x^n(\tm,\theta^n),$ $s^n)$. By the union of events bound,
\begin{align}
\err(\code|s^n,\theta^n)\leq& \prob{\Fset_1}+\prob{\Fset_2}
+\prob{\Eset_{1}\cap\Fset_1^c}+
\prob{\Eset_{2}\cap\Fset_2^c} \,,
\end{align}
where the conditioning on $S^n=s^n$ and $T^n=\theta^n$ is omitted for convenience of notation.
Based on Lemma~\ref{lemm:codeBsets}, 
the probabilities of the events $\Fset_1$ and $\Fset_2$ tend to zero as $n\rightarrow\infty$,
by (\ref{eq:12ebn}) and (\ref{eq:13ebn}), respectively.

Now, suppose that Condition 1) of the decoding rule is violated. 
Observe that the event $\Eset_{1}\cap\Fset_1^c$ implies that 
\begin{align}
&D(P_{T,X,S,Y}||P_{T,X,S}\times W_{Y|X,S,T})
\nonumber\\
=& D(P_{T,X,S,Y}|| P_T\times P_{X|T}\times P_{S|T} \times W_{Y|X,S,T})-I(X;S|T)
>\eta-\eps \,.
\end{align}
Then,  by standard large deviations considerations (see \eg \cite[pp. 362--364]{CoverThomas:06b}),
\begin{align}
\prob{\Eset_{1}\cap\Fset_1^c}  
\leq& \max_{P_{T,X,S,Y}\,:\; \Eset_{1}\cap\Fset_1^c\;\text{holds}
} 2^{-n (D(P_{T,X,S,Y} || P_{T,X,S}\times W_{Y|X,S,T})-\eps)} 
\nonumber\\
<& 2^{-n(\eta-2\eps)} \,,
\end{align}
which tends to zero as $n \rightarrow \infty$, for sufficiently small $\eps>0$, with $\eps<\frac{1}{2}\eta$.

Moving to Condition 2) of the decoding rule, let $\Dset_{2}$ denote the set of joint types 
$P_{T,X,\tX,S}$ such that
\begin{align}
&D(P_{T,X,S,Y}|| P_T P_{X|T}\times P_{S|T} \times W_{Y|X,S,T})\leq \eta \,, \\
&D(P_{\tX,\tS,Y}|| P_{\tX_1}\times P_{\tS|T} \times W_{Y|X,S,T})\leq \eta \,,\;
\text{for some  $\tS\sim \tq(s|t)$} \,,
\label{eq:D2a2}
\\
&I_q(X,Y;\tX|S,T)>\eta \,.
\label{eq:D2a3}
\end{align}
Then,  
by standard type class considerations (see \eg \cite[Theorem 1.3]{Kramer:08n}), 
\begin{align}
\cprob{\Eset_{2}\cap\Fset_2^c}{M=m} 
\leq&
\sum_{  \substack{P_{T,X,\tX,S}\in \Dset_{2} \,:\; \\   \Fset_2^c \;\text{holds}  }   }
|\{\tm\,:\;
 (\theta^n,x^n(m,\theta^n),x^n(\tm,\theta^n),s^n)\in \Tset^n(P_{T,X,\tX,S}) \}|
\nonumber\\
&\times 2^{-n\left( I_q(\tX;Y|X,S,T)-\eps  \right)} \,,
\end{align}
for every given $m\in [1:2^{nR}]$. Hence, by (\ref{eq:11ebn}),
\begin{align}
&\prob{\Eset_{2}\cap\Fset_2^c}\leq 
\sum_{  \substack{P_{T,X,\tX,S}\in \Dset_{2} \,:\;  \\   \Fset_2^c \;\text{holds}  }   } 
 2^{-n\left( I_q(\tX;Y|X,S,T)
-\left[ R-I_q(\tX;X,S|T) \right]_{+}-2\eps  \right)} \,.
\label{eq:LE2E02cB}
\end{align}

To further bound $\prob{\Eset_{2}\cap\Fset_2^c}$, consider the following cases.
Suppose that $R\leq I_q(\tX;S|T)$. Then, given $\Fset_2^c$, we have that
\begin{align}
I_q(X;\tX|S,T)\leq I_q(X;\tX,S|T)\leq \eps \,.
\end{align}
By (\ref{eq:D2a3}), it then follows that
\begin{align}
I_q(\tX;Y|X,S,T)=&I_q(\tX;X,Y|S,T)-I_q(\tX;X|S,T) 
\nonumber\\
\geq& \eta-\eps \,.
\label{eq:rule2aIcase1}
\end{align}
Returning to (\ref{eq:LE2E02cB}), we note that since the number of types is polynomial in $n$, the cardinality of the set of types $\Dset_{2}$ can be bounded by $2^{n\eps}$, for sufficiently large $n$. Hence, by (\ref{eq:LE2E02cB}) and (\ref{eq:rule2aIcase1}), we have that $\prob{\Eset_{2}\cap\Fset_2^c}\leq 2^{-n(\eta-4\eps)}$, which tends to zero as $n\rightarrow\infty$, for $\eps<\frac{1}{4}\eta$.

Otherwise, if $R> I_q(\tX;S|T)$, then given $\Fset_2^c$,
\begin{align}
R>&I_q(X;\tX,S|T)+I(\tX;S|T)-\eps 
\nonumber\\
=& I_q(\tX;X,S|T)+I(X;S|T)-\eps 
\nonumber\\
\geq& I_q(\tX;X,S|T)-\eps \,.
\end{align}
Thus,
\begin{align}
\left[ R-I_q(\tX;X,S|T) \right]_{+}\leq 
R-I_q(\tX;X,S|T)+\eps \,.
\end{align}
Hence, by (\ref{eq:LE2E02cB}) we have that
\begin{align}
\prob{\Eset_{2}\cap\Fset_2^c}\leq& \sum_{  \substack{P_{T,X,\tX,S}\in \Dset_{2}  \\   \Fset_2^c \;\text{holds}  }   } 2^{-n(I(\tX;X,S,Y|T)-R-3\eps )}
\nonumber\\
\leq& \sum_{  \substack{P_{T,X,\tX,S}\in \Dset_{2} \,:\; \\   \Fset_2^c \;\text{holds}  }   } 2^{-n(I_q(\tX;Y|T)-R-3\eps )} \,.
\end{align}
For $P_{T,X,\tX,S}\in \Dset_{2}$, we have by (\ref{eq:D2a2}) that
$P_{T,\tX,\tS,Y}$ is arbitrarily close to some 
$P_{T,X,\tS,\tY}$, where
\begin{align}
P_{T,X,\tS,\tY}(x,s,y)= P_T(t) P_{X|T}(x|t)\tq(s|t)W_{Y|X,S,T}(y|x,s,t) \,,
\end{align}
if $\eta>0$ is sufficiently small. In which case, 
\begin{align}
I_q(\tX;Y|T)\geq I_{\tq}(X;Y|T)-\delta \,,
\end{align}
where $\delta>0$ is arbitrarily small.
Therefore, provided that
\begin{align}
R <& \min_{q(s|t) \,:\; \E_q l(S)\leq\Lambda} I_q(X;Y|T)-\delta-5\eps
\,,
\end{align}
we have that $\prob{\Eset_{2}\cap\Fset_2^c}\leq 2^{-n(I_q(\tX;Y|T)-R-4\eps )}$ tends to zero as $n\rightarrow\infty$.
\qed

\subsection{Converse Proof}
We will use the following lemma, based on the observations of Ericson \cite{Ericson:85p}.
\begin{lemma} 
\label{lemm:nEricson}
Consider the AVC with fixed parameters free of state constraints, and let  $\code=(f,g)$ be a
$(2^{nR},n)$ deterministic code.
Suppose that the channels $W_{Y|X,S,T}(\cdot|\cdot,\cdot,\theta_i)$ are symmetrizable for all $i\in [1:n]$, and let $J_t(s|x)$, $t\in \Tset$, be a set of conditional state distributions that satisfy (\ref{eq:symmetrizable}).
If $R>0$, then
\begin{align}
&\err(\tq,\theta^n,\code) \geq \frac{1}{4} \,,\;
\intertext{for} 
&\tq(s^n|\theta^n)= \frac{1}{2^{n R}} \sum_{m=1}^{2^{nR}}  J_{\theta^n}(s^n|f^n(m,\theta^n)) 
\label{eq:PconvFtq}
\,,
\end{align}
where $J_{\theta^n}(s^n|x^n)=\prod_{i=1}^n J_{\theta_i}(s_i|x_{i})$.
\end{lemma}
For completeness, we give the proof below.
\begin{proof}[Proof of Lemma~\ref{lemm:nEricson}]
Denote the codebook size by  $\dM=2^{nR}$,  and the codewords by $x^n(m,\theta^n)=f^n(m,\theta^n)$.

Under the conditions of the lemma,
\begin{align}
& \err(\tq,\theta^n,\code)=\sum_{s^n\in\Sset^n} q(s^n|\theta^n) \frac{1}{\dM} \sum_{m=1}^{\dM} 
\sum_{y^n \,:\; g(y^n,\theta^n)\neq m} W^n(y^n|x^n(m,\theta^n),s^n,\theta^n)
\nonumber\\
=& \frac{1}{\dM^2} \sum_{\tm=1}^{2^{nR}} \sum_{s^n\in\Sset^n} J_{\theta^n}(s^n|x^n(\tm,\theta^n)) \sum_{m=1}^{\dM} 
\sum_{y^n \,:\; g(y^n,\theta^n)\neq m} W^n(y^n|x^n(m,\theta^n),s^n,\theta^n)
\end{align}
where have defined $W^n\equiv W_{Y^n|X^n,S^n,T^n}$ for short notation.
By switching between the summation indices $m$ and $\tm$, we obtain
\begin{align}
 \err(\tq,\theta^n,\code)
=& \frac{1}{2\dM^2} \sum_{m,\tm}\; \sum_{y^n \,:\; g(y^n,\theta^n)\neq m} 
\sum_{s^n\in\Sset^n} W^n(y^n|x^n(m,\theta^n),s^n,\theta^n) J_{\theta^n}(s^n|x^n(\tm,\theta^n))  
\nonumber\\
+& \frac{1}{2\dM^2} \sum_{m,\tm}\; \sum_{y^n \,:\; g(y^n,\theta^n)\neq \tm}
\sum_{s^n\in\Sset^n} W^n(y^n|x^n(\tm,\theta^n),s^n,\theta^n) J_{\theta^n}(s^n|x^n(m,\theta^n))    \,.
\end{align}
Now, as the channel is memoryless, 
\begin{align}
\sum_{s^n\in\Sset^n} W^n(y^n|x^n(\tm,\theta^n),s^n,\theta^n) J_{\theta^n}(s^n|x^n(m,\theta^n))
=&
\prod_{i=1}^n \sum_{s_i\in\Sset} W_{Y_i|X_i,S_i,T_i}(y_i|x_{i}(\tm,\theta^n),s_{i},\theta_i) J_{\theta_i}(s_i|x_{i}(m,\theta^n))
\nonumber\\ =&
\prod_{i=1}^n \sum_{s_i\in\Sset} W_{Y_i|X_i,S_i,T_i}(y_i|x_{i}(m,\theta^n),s_{i},\theta_i) J_{\theta_i}(s_i|x_{i}(\tm,\theta^n))
\nonumber\\ =&
\sum_{s^n\in\Sset^n} W^n(y^n|x^n(m,\theta^n),s^n,\theta^n) J_{\theta^n}(s^n|x^n(\tm,\theta^n)) \,,
\end{align}
where the second equality is due to (\ref{eq:symmetrizable}). Therefore, 
\begin{align}
 \err(\tq,\theta^n,\code)
\geq & \frac{1}{2\dM^2} \sum_{\tm\neq m}\; \sum_{s^n\in\Sset^n}
\Big[ \sum_{y^n \,:\; g(y^n,\theta^n)\neq m} 
 W^n(y^n|x^n(m,\theta^n),s^n,\theta^n) J_{\theta^n}(s^n|x^n(\tm,\theta^n))  
\nonumber\\&
+\sum_{y^n \,:\; g(y^n,\theta^n)\neq \tm} 
 W^n(y^n|x^n(m,\theta^n),s^n,\theta^n) J_{\theta^n}(s^n|x^n(\tm,\theta^n))  
\Big]
\nonumber\\
\geq& \frac{1}{2\dM^2} \sum_{\tm\neq m}\; \sum_{s^n\in\Sset^n}
\sum_{y^n\in\Yset^n} 
 W^n(y^n|x^n(m,\theta^n),s^n,\theta^n) J_{\theta^n}(s^n|x^n(\tm,\theta^n))  
\nonumber\\
=& \frac{\dM(\dM-1)}{2\dM^2}=\frac{1}{2}\left( 1-\frac{1}{\dM} \right) \,.
\end{align}
Assuming  the sum rate is positive, we have that $\dM\geq 2$, hence $\err(\tq,\theta^n,\code)\geq
\frac{1}{4}$.
\end{proof}

Now, we are in position to prove the converse part of Theorem~\ref{theo:PCavc}. 
Consider a sequence of $(2^{nR},n,\alpha_n)$ deterministic codes $\code_n$ over the AVC with fixed parameters under input constraint $\plimit$ and state constraint $\Lambda$, where $\alpha_n\rightarrow 0$ as
$n\rightarrow \infty$. 
In particular,  the conditional probability of error given a state sequence $s^n$ is bounded by
\begin{align}
\err(\code_n|s^n,\theta^n)\leq \alpha_n \,,\;\text{for $s^n\in\Sset^n$ with $l^n(s^n)\leq\Lambda$} \,.
\label{eq:MStateConverse1bDet}
\end{align}
Let $X^n=\encn(M,\theta^n)$ be the channel input sequence, and let $Y^n$ be the corresponding output. 

Consider using the same code over the compound channel with fixed parameters,  \ie where the jammer selects a state sequence at random according to a product distribution, $\oS^n\sim \prod_{i=1}^n q(\os_i|\theta_i)$, under the \emph{average} state constraint 
$\frac{1}{n}\sum_{i=1}^n \E_q l(S_i)\leq\Lambda-\delta$. 
Here, there is no state constraint with probability $1$, as the jammer may select a sequence $\oS^n$ with $l^n(\oS^n)>\Lambda$. 
Yet, the probability of error is bounded by
	\begin{align}
	\err(\oq,\theta^n,\code_n)
	\leq 
	\sum_{s^n\,:\; l^n(s^n)\leq\Lambda} \oq^{n}(s^n|\theta^n) \err(\code^\Gamma|s^n,\theta^n)
	+\prob{l^n(\oS^{n})>\Lambda} . 
	\end{align}
	The first sum is bounded by (\ref{eq:MStateConverse1bDet}), and the second term vanishes 
	by the law of large numbers, since
	$\oq\in\overline{\pSpace}_{\Lambda-\delta}(\Sset|\theta^{\infty})$. 
	It follows that the code sequence of the constrained AVC achieves the same rate $R$ over the compound channel $W_{Y|X,\oS,T}$. 
As in Appendix~\ref{app:PCcompound}, Fano's inequality implies that
for every jamming strategy $\oq^n(s^n|\theta^n)$,
\begin{align}
R\leq&  \min_{\oq(s|t) \,:\; \E_q l(S)\leq\Lambda } I_{\oq}(X;Y|T)+\eps_n  \,,
\end{align}
with $X\triangleq X_K$, $T\equiv \theta_K$, $Y\triangleq Y_K$, where $K$ is uniformly distributed over $[1:n]$.
Hence, $T$ is distributed according to the type of the parameter sequence $\theta^n$ (see (\ref{eq:ConvPT})).

Returning to the original AVC, suppose that $L_n^*>\Lambda$.
 It remains to show that $R>0$ implies that $\tLambda_n(P_{X|T})\geq \Lambda$. 
If the channels $W_{Y|X,S,T}(\cdot|\cdot,\cdot,\theta_i)$ is non-symmetrizable for some $i\in [1:n]$, then  $\tLambda_n(P_{X|T})=+\infty$, and there is nothing to show. Hence, consider the case where $W_{Y|X,S,T}(\cdot|\cdot,\cdot,\theta_i)$ are symmetrizable for all
$i\in [1:n]$.
Assume to the contrary that $R>0$ and $\tLambda_n(P_{X|T})<\Lambda$. 
Hence,   there exist  conditional state distributions $J_{\theta_i}(s|x)$ that symmetrize $W_{Y|X,S,T}(\cdot|\cdot,\cdot,\theta_i)$, such that
\begin{align}
\tLambda_n(P_{X|T})= \frac{1}{n}\sum_{i=1}^n \sum_{x,s} P_{X|T}(x|\theta_i) J_{\theta_i}(s|x)l(s) < \Lambda \,.
\label{eq:LJnConvH10}
\end{align}
Now,  consider the following jamming strategy. First, the jammer selects  a codeword $\tX^n$ from the codebook uniformly at random. Then, 
the jammer selects a sequence $\tS^n$ at random, according to the conditional distribution
\begin{align}
\cprob{\tS^n=s^n}{\tX=x^n}= J_{\theta^n}(s^n|x^n)\triangleq \prod_{i=1}^n J_{\theta_i}(s_i|x_{i})\,.
\label{eq:HtSn}
\end{align} 
At last, if $l^n(\tS^n)\leq \Lambda$, the jammer chooses the state sequence to be $S^n=\tS^n$. Otherwise, the jammer chooses $S^n$ to be some sequence of zero cost. Such jamming strategy satisfies the state constraint $\Lambda$ with probability $1$.

To contradict our assumption that $\tLambda(P_{X|T})<\Lambda$, we first show that 
$\E l^n(\tS^n)=\tLambda(P_{X|T})$.
Observe that for every $x^n\in\Xset^n$, 
\begin{align}
\E\, \left( l^n(\tS^n) | \tX^n=x^n \right)  =& 
\frac{1}{n} \sum_{i=1}^n \sum_{s\in\Sset} l (s) J_{\theta_i}(s|x_{i}) \,.
\end{align}
Since $\tX^n$ is distributed as $X^n$, we obtain
\begin{align}
\E\, l^n(\tS^n)=&  \sum_{s\in\Sset} l (s) \cdot \frac{1}{n}  \sum_{i=1}^n \E J_{\theta_i}(s|X_{i}) =
\frac{1}{n}  \sum_{i=1}^n \sum_{x,s} P_{X|T}(x|\theta_i) J_{\theta_i}(s|x) l (s) =
\tLambda_n(P_{X|T})<\Lambda \,.
\end{align}
Thus, by Chebyshev's inequality we have that for sufficiently large $n$, 
\begin{align}
\prob{ l^n(\tS^n)>\Lambda} \leq \delta_0 \,,
\end{align}
where $\delta_0>0$ is arbitrarily small.
Now, on the one hand, the probability of error is bounded by
\begin{align}
\err(q,\theta^n,\code_n)
\geq& \prob{g(Y^n,\theta^n)\neq M, l^n(\tS^n)\leq \Lambda}
\nonumber\\
=&   \sum_{s^n \,:\; l^n(s^n)\leq\Lambda} \tq(s^n|\theta^n) \err(\code_n|s^n,\theta^n) \,,
\label{eq:convEb1}
\end{align}
where $\tq(s^n|\theta^n)$ is as defined in (\ref{eq:PconvFtq}).
On the other hand, the sequence $\tS^n$ can be thought of as the state sequence of an AVC without a state constraint, hence, by 
Lemma~\ref{lemm:nEricson}, 
\begin{align}
\frac{1}{4}\leq &\err(\tq,\theta^n,\code_n)
\leq  \sum_{s^n \,:\; l^n(s^n)\leq\Lambda} \tq(s^n|\theta^n) \err(\code_n|s^n,\theta^n) +\prob{ l^n(\tS^n)> \Lambda}
\nonumber\\
\leq&  \sum_{s^n \,:\; l^n(s^n)\leq\Lambda} \tq(s^n|\theta^n) \err(\code_n|s^n,\theta^n)  +\delta_0 \,.
\label{eq:convEbf}
\end{align}
Thus, by (\ref{eq:convEb1})-(\ref{eq:convEbf}), the probability of error is bounded by 
$\err(q,\theta^n,\code_n)\geq \frac{1}{4}-\delta_0$. As this cannot be the case for a code with vanishing probability of error, we deduce that the assumption is false, \ie $R>0$ implies that $\tLambda_n(P_{X|T})\geq \Lambda$.

If $L_n^*<\Lambda$, then $\tLambda_n(P_{X|T})< \Lambda$ for all $P_{X|T}$ with $\E \cost(X)\leq\plimit$, and a positive rate cannot be achieved.
This completes the converse proof.
\qed

\section{Proof of Corollary~\ref{coro:LCavc01}}
\label{app:LCavc01}
Assume that the AVC $\avc$ with fixed parameters satisfies the conditions of Corollary~\ref{coro:LCavc01}. 
%
Looking into  the converse proof above, the following addition suffices.
We show that for every code $\code_n$ as in the converse proof above,
$\tLambda_n(P_{X|T})=\Lambda$ implies that $R=0$. Since there is only a polynomial number of types, we may consider $P_{X|T}(x|t)$ to be the conditional type of 
$f^n(m,\theta^n)$ given $\theta^n$, for all $m\in [1:2^{nR}]$ (see \cite[Problem 6.19]{CsiszarKorner:82b}).

Suppose that  $\tLambda_n(P_{X|T})=\Lambda$,  assume to the contrary that $R>0$,
 and let $J_i(s|x)$ be distributions that achieve the minimum in (\ref{eq:LambdaOig}), \ie
\begin{align}
\tLambda_n(p)=& \frac{1}{n} \sum_{i=1}^n
 \sum_{ x,s} P_{X|T}(x|\theta_i) J_i(s|x) l(s) =\Lambda \,.
\label{eq:tlambdaJConvc}
\end{align}
Based on the condition of the corollary, we may assume that $J_i(s|x)$ is a $0$-$1$ law, \ie 
\begin{align}
J_i(s|x)=\begin{cases}
1 &\text{if $s=G_i(x)$},\\
0 &\text{otherwise}
\end{cases} \,,
\end{align}
for some deterministic function $G_i:\Xset\rightarrow\Sset$.  

Recall that we have defined $X=X_K$, $Y=Y_K$ in the converse proof, where $K$ is a uniformly distributed variable over $[1:n]$.
Thus, by (\ref{eq:tlambdaJConvc}),
\begin{align}
 \E l(G_K(X))= \frac{1}{n} \sum_{i=1}^n \sum_{ x,s} p(x|\theta_i) J_i(s|x) l(s) =\Lambda \,.
\label{eq:tlambdaJConvEG}
\end{align}
Now,  consider the following jamming strategy. First, the jammer selects a codeword $\tX^n$ from the codebook uniformly at random. Then, given 
$\tX^n=x^n$, the jammer chooses the state sequence $S^n=\left( G_i(x_{i}) \right)_{i=1}^n$. 
Observe that 
\begin{align}
 l^n(S^n)   =&  \frac{1}{n} \sum_{i=1}^n l(G_i(x_{i})) 
=    \E l(G_K(X)) 
	=\Lambda \,,
\end{align}
where the last equality is 
due to 
(\ref{eq:tlambdaJConvEG}).
Thus, the state sequence satisfies the state constraint.
Now, observe that the jamming strategy $S^n=\left( G(\tX_{i}) \right)_{i=1}^n$ is equivalent to 
$S^n\sim\tq(s^n|\theta^n)$ as in (\ref{eq:PconvFtq}).
Thus, by Lemma~\ref{lemm:nEricson}, 
we have that 
$\err(\tq,\code_n)\geq \frac{1}{4}$, hence a positive rate cannot be achieved.
\qed

\section{Proof of Lemma~\ref{lemm:CiequivDet}}
\label{app:CiequivDet}
Suppose that $L_n^*>\Lambda$. The proof is similar to that of Lemma~\ref{lemm:Ciequiv}. 
We begin with the  property in the lemma below.
\begin{lemma}
\label{lemm:PsametDet}
Let $\omega_i^*$, $\lambda_i^*$, $\tlambda_i^*$, $i\in [1:n]$, be the parameters that achieve the saddle point in  (\ref{eq:Cieqiv2Det}), \ie
\begin{align}
\inR_{n}(\avc)= 
\frac{1}{n} \sum_{i=1}^n  
\inC_{\theta_i}(\omega_i^*,\lambda_i^*,\tlambda_i^*) \,.
\end{align}
Then, for every $i,j\in [1:n]$ such that $\theta_i=\theta_j$, we have that $\omega_i^*=\omega_j^*$, $\tlambda_i^*=\tlambda_j^*$, and $\lambda_i^*=\lambda_j^*$.
\end{lemma}

\begin{proof}[Proof of Lemma~\ref{lemm:PsametDet}]
For every $i\in [1:n]$, let $p_i,q_i$ denote input and state distributions such that $\E \cost(X_i)\leq \omega_i^*$,
$\tLambda_{\theta_i}(p_i)\geq \tlambda_i^*$,
 $\E l(S_i)\leq \lambda_i^*$ for $X_i\sim p_i$, $S_i\sim q_i$. Now, suppose that $\theta_i=\theta_j=t$, and define
\begin{align}
p'(x)=\frac{1}{2}[ p_i(x)+p_j(x) ] \,,\; q'(s)=\frac{1}{2}[ q_i(s)+q_j(s) ] \,.
\end{align}
Then, $\E \cost(X') = \frac{1}{2}[ \E \cost(X_i)+ \E \cost(X_j) ]$,  $\Lambda_t(p')= \frac{1}{2}[ \Lambda_t(p_i)+ \Lambda_t(p_j) ]$, and $\E l(S') = \frac{1}{2}[ \E l(S_i)+ \E l(S_j) ]$ 
for $X'\sim p'$, $S'\sim q'$. Furthermore, since the mutual information is concave-$\cap$ in the input distribution and convex-$\cup$ in the state distribution, we have that
\begin{align}
& \frac{1}{2}\left[ I_{q'}(X_i;Y_i|T_i=t)+I_{q'}(X_j;Y_j|T_j=t) \right] \leq I_{q'}(X';Y'|T'=t) \nonumber\\
& \frac{1}{2}\left[ I_{q_i}(X';Y'|T'=t)+I_{q_j}(X';Y'|T=t) \right] \geq I_q(X';Y'|T'=t) \,.
\end{align}
Therefore, the saddle point distributions must satisfy $p_i=p_j=p'$ and $q_i=q_j=q'$, hence $\omega_i^*=\omega_j^*$, $\tlambda_i^*=\tlambda_j^*$, and $\lambda_i^*=\lambda_j^*$.
\end{proof}

Next, it can be inferred from Lemma~\ref{lemm:PsametDet} that
\begin{align}
\inR_{n}(\avc)=& \min_{ \substack{ (\lambda_t)_{t\in\Tset} \,:\; \\  \sum\limits_{t\in\Tset} P_T(t) \lambda_t \leq \Lambda } } \;
\max_{ \substack{ (\omega_t)_{t\in\Tset}, (\tlambda_t)_{t\in\Tset} \,:\;\\   \sum\limits_{t\in\Tset} P_T(t) \omega_t\leq \plimit \\
\sum\limits_{t\in\Tset} P_T(t) \tlambda_t\geq \Lambda } }
\sum\limits_{t\in\Tset} P_T(t)  
\inC_{t}(\omega_t,\lambda_t)
\nonumber\\
=& \min_{ \substack{ (\lambda_t)_{t\in\Tset} \,,\; q(s|t) \,:\; \\  
\E_q[ l(S)|T=t ]\leq \lambda_t \\
\sum\limits_{t\in\Tset} P_T(t) \lambda_t \leq \Lambda } }
\max_{ \substack{ (\omega_t)_{t\in\Tset} \,,\; (\tlambda_t)_{t\in\Tset} \,,\; p(x|t) \,:\; \\  
\E[ \cost(X)|T=t ]\leq \omega_t, \tLambda(p,t)\geq \tlambda_t \\
\sum\limits_{t\in\Tset} P_T(t) \omega_t \leq \plimit \,,\; \sum\limits_{t\in\Tset} P_T(t) \tlambda_t \geq \Lambda  } } I_q(X;Y|T)
\nonumber\\
=& \min_{  q(s|t) \,:\; \E_q l(S) \leq \Lambda  }
\max_{ \substack{ p(x|t) \,:\; \E \cost(X)\leq \plimit \,,\\  \tLambda_n(p)\geq\Lambda  }  } I_q(X;Y|T) =\inC_{n}(\avc)
 \,,
\end{align}
where $P_T$ is the type of the parameter sequence $\theta^n$. The second equality follows from the definition of $\inC_{t}(\omega_t,\lambda_t,\tlambda_d)$ in (\ref{eq:CtolDet}), using the minimax theorem \cite{sion:58p} to switch between the order of the minimum and maximum.
In the third line, we eliminate the slack variables $\lambda_i$, $\omega_i$, and $\tlambda_i$, replacing $\E_q l(S_i)$, $\E\cost(X_i)$, and $\tLambda(p,\theta_i)$, respectively.
The last equality holds by the definition of $\inC_{n}(\avc)$ in (\ref{eq:Cieqiv3Det}).
\qed

\section{Analysis of Example~\ref{example:Fading}}
\label{app:Fading}
Consider the fading AVC in Example~\ref{example:Fading}. To show the direct part with random codes, set the conditional input distribution $X\sim \mathcal{N}(0,\omega(t))$ given $T=t$ in (\ref{eq:Cieqiv2}).
Then, for every $t\in\Tset$,
\begin{align}
I_q(X;Y|T=t)\geq \frac{1}{2} \log\left( 1+\frac{t^2 \omega(t) }{ \lambda'(t)+\sigma^2 }  \right) \,,
\end{align}
where we have denoted $\lambda'(t)\triangleq  \E(S^2|T=t)$. The last inequality holds since Gaussian noise is known to be the worst additive noise under variance constraint \cite[Lemma II.2]{DiggaviCover:01p}.
The direct part follows. As for the converse part, consider a jamming scheme where the state is drawn according to
the conditional distribution $S\sim \mathcal{N}(0,\lambda(t))$ given $T=t$. 
Then, the proof follows from Shannon's classic result on the Gaussian channel $Y=t X+V$ with $V\sim\mathcal{N}(0,\lambda(t)+\sigma^2)$. 

We move to the deterministic code capacity. 
 By Definition~\ref{def:symmetrizable}, the constant-parameter channel $W_{Y|X,S,T=t}$ is symmetrized by 
a conditional pdf
$\varphi(s|x)$ if 
\begin{align}
\label{eq:symmetrizableFading}
\int_{-\infty}^\infty  \varphi(s|x_2)f_{Z}(y-tx_1-s)ds=
\int_{-\infty}^\infty  \varphi(s|x_1)f_{Z}(y-tx_2-s)ds
 \,,\; 
\forall\, x_1,x_2,y\in\mathbb{R} \,,
\end{align}
where $f_{Z}(z)=\frac{1}{\sqrt{2\pi\sigma^2}} e^{-z^2/2\sigma^2}$. 
%
Equivalently, the constant-parameter channel is symmetrized by $\varphi_x(s)\equiv\varphi(s|x)$ if 
\begin{align}
\label{eq:symmetrizableFadingEq}
\int_{-\infty}^\infty  \varphi_0(s)f_{Z}(y-tx-s)ds=
\int_{-\infty}^\infty  \varphi_x(s)f_{Z}(y-s)ds
 \,,\; 
\end{align}
for all $x,y\in\mathbb{R}$.
By substituting $z=y-tx-s$ in the LHS, and $\bar{z}=y-s$ in the RHS, we have 
\begin{align}
\label{eq:symmetrizableEq2EqFading}
\int_{-\infty}^\infty \varphi_0(y-tx-z)f_{Z}(z)dz=
\int_{-\infty}^\infty \varphi_{x}(y-\bar{z})f_{Z}(\bar{z})d\bar{z} \,.
\end{align}
For every $x\in\mathbb{R}$, define the random variable $\oS(x)\sim\varphi_{x}$.
We note that the RHS is the convolution of the pdfs of the random variables $Z$ and 
$\oS(x) $, while the LHS is the convolution of the pdfs of the random variables $Z$ and 
$\oS(0)+x $. This is not surprising since the channel output $Y$ is  a sum of independent random variables, and thus the pdf of $Y$ is a convolution of pdfs.
It follows that $\varphi_0(y-tx)=\varphi_{x}(y)$, and
 by plugging $s$ instead of $y$, we have that $\varphi_{x}$ symmetrizes the constant-parameter channel $W_{Y|X,S,T=t}$ if and only if 
\begin{align}
\varphi_{x}(s)=\varphi_0 (s-tx) \,.
\label{eq:GPvarphiSymmEqFading}
\end{align} 
Then, the corresponding state cost satisfies
\begin{align}
\int_{-\infty}^\infty\int_{-\infty}^\infty  f_{X|T}(x|t) \varphi_{x}(s) s^2 
\, dx \, ds
=& \int_{-\infty}^\infty  \int_{-\infty}^\infty f_{X|T}(x|t) \varphi_0 (s-tx) s^2 \, ds
\, dx
\nonumber\\
=& \int_{-\infty}^\infty  \int_{-\infty}^\infty f_{X|T}(x|t) \varphi_0 (a) (a+tx)^2 \, da
\, dx
\nonumber\\
=& \int_{-\infty}^\infty  
\left[ \int_{-\infty}^\infty   (tx+a)^2
f_{X|T}(x|t) \, dx
\right]
 \varphi_0 (a)  \, da
\label{eq:GPscostEq1F}
\end{align}
where the second equality follows by the integral substitution of $a=s-tx$.
Observe that the bracketed integral can be expressed as 
\begin{align}
\int_{-\infty}^\infty  (tx+a)^2
f_{X|T}(x|t) \, dx =\E[(tX+a)^2|T=t]=t^2\E[ X^2|T=t]+a^2 \,.
\end{align}
Thus, by (\ref{eq:GPscostEq1F}), 
\begin{align}
\int_{-\infty}^\infty  \int_{-\infty}^\infty f_{X|T}(x|t) \varphi_{x}(s) s^2 
\, dx \, ds
=&t^2 \E[ X^2|T=t]+\int_{-\infty}^\infty  a^2  \varphi_0 (a)  \, da
\nonumber\\
\geq& t^2 \E[ X^2|T=t] \,.
\label{eq:GPCscostTrF}
\end{align}
Note that 
the last inequality holds for any $\varphi_{x}$ which symmetrizes the channel, and in particular 
 for $\hat{\varphi}_{x}(s)=\delta(s-tx)$, where $\delta(\cdot)$ is the Dirac delta function. 
In addition, since $\hat{\varphi}_0$ gives probability $1$ to $S=0$, we have that 
 (\ref{eq:GPCscostTrF}) holds with equality for $\hat{\varphi}_{x}$,
and thus, 
\begin{align}
\LambdaOig( F_{X|T} )=\frac{1}{n}\sum_{i=1}^n t^2 \E[ X^2|T=t]=\sum_{t\in\Tset} P_T(t)t^2 \E[ X^2|T=t] =
\E(T^2 \omega(T))
 \,,
\end{align}
with $\omega(t)\equiv \E[X^2|T=t]$. Hence,
\begin{align}
L_n^*=\max_{ \omega(t) \,:\; \E \omega(T)\leq\plimit }  \E(T^2 \omega(T))
 \,.
\end{align}

Having shown that the minimum in (\ref{eq:LambdaOig}) is attained by a $0$-$1$ law, we have by Corollary~\ref{coro:LCavc01} 
that the capacity of the fading AVC is 
$
\Cavc=\liminf
\inC_{n}(\avc)  
$, with 
\begin{align}
\inC_n(\avc)& =
\begin{cases}
\min\limits_{  F_{S|T} \,:\; \E S^2\leq \Lambda } \;
\max\limits_{ \substack{ F_{X|T} \,:\; \E\, X^2 \leq\plimit \,,\; \\ \E(T^2 X^2) \geq \Lambda } } \;   I_q(X;Y|T)  &\text{if 
$\max\limits_{ \omega(t) \,:\; \E \omega(T)\leq\plimit }  \E(T^2 \omega(T))> \Lambda$}\,,\\
  0 																 &\text{if $\max\limits_{ \omega(t) \,:\; \E \omega(T)\leq\plimit }  \E(T^2 \omega(T))\leq \Lambda$}
\end{cases}	
\,.
\label{eq:Cieqiv3DetF} 
\end{align}
To show the direct part, we only need to consider the case where $\max\limits_{ \omega(t) \,:\; \E \omega(T)\leq\plimit }  \E(T^2 \omega(T))> \Lambda$.
Then, set the conditional input distribution $X\sim \mathcal{N}(0,\omega(t))$ given $T=t$ in (\ref{eq:Cieqiv3DetF}).
As in the direct part with random codes,  
\begin{align}
I_q(X;Y|T=t)\geq \frac{1}{2} \log\left( 1+\frac{t^2 \omega(t) }{ \lambda'(t)+\sigma^2 }  \right) \,,
\end{align}
with 
$\lambda'(t)\triangleq  \E(S^2|T=t)$, 
since Gaussian noise is 
the worst additive noise under variance constraint \cite[Lemma II.2]{DiggaviCover:01p}.
The direct part follows. As for the converse part, for
the conditional distribution $S\sim \mathcal{N}(0,\lambda(t))$ given $T=t$, 
we have that
\begin{align}
I_q(X;Y|T=t)\leq \frac{1}{2} \log\left( 1+\frac{t^2 \omega'(t) }{ \lambda(t)+\sigma^2 }  \right) \,,
\end{align}
with $\omega'(t)\triangleq  \E(X^2|T=t)$, since the Gaussian distribution maximizes the differential entropy. The proof follows.
\qed

\section{Proof of Lemma~\ref{lemm:WaterProp}}
\label{app:WaterProp}
\subsection*{Part 1}
Since $\sum_{j'=1}^d P_{j'}^*=\plimit>0$, there must be some $j\in [1:d]$ such that 
$ P_j^*=\alpha-(N_j^*+\sigma_j^2)>0$, thus $\alpha>N_j^*+\sigma_j^2$. If $N_j^*=0$, then
it follows that $\beta\leq\sigma_j^2$, hence
\begin{align}
\alpha>N_j^*+\sigma_j^2=\sigma_j^2\geq  \beta \,.
\end{align}
Otherwise, $N_j^*=\beta-\sigma_j^2>0$, thus by the assumption $P_j^*>0$, we have that
\begin{align}
0< P_j^*= \alpha-(N_j^*+\sigma_j^2)= \alpha-\beta \,.
\end{align}

\subsection*{Part 2}
Assume to the contrary that $N_j^*=\beta-\sigma_j^2>0$ and $P_j^*=0$. The assumption $P_j^*=0$  implies that $\alpha\leq N_j^*+\sigma_j^2=\beta$, 
 in contradiction to part 1 of the Lemma. Hence, the assumption is false, and 
$N_j^*>0$ implies that $P_j^*>0$. 

\subsection*{Part 3 and Part 4}
By the definition of $N_j^*$ in  (\ref{eq:GPNjdef}), we have that $N_j^*+\sigma_j^2=\max(\beta,\sigma_j^2)$ for all $j\in [1:d]$. Thus,
\begin{align}
P_j^*+N_j^*+\sigma_j^2 =&\max(\beta,\sigma_j^2)+\left[ \alpha-\max(\beta,\sigma_j^2) \right]_{+}
=\max( \alpha,\beta,\sigma_j^2)= \max( \alpha,\sigma_j^2) \,,
\end{align}
where the last equality is due to part 1.
Part 4 immediately follows. \qed



\section{Proof of Lemma~\ref{lemm:GPscostP}}
\label{app:GPscostP}
Let $X^d$ be a zero mean random vector 
with the covariance matrix $K_X$. 
Observe that by (\ref{eq:GPsymmetrizable}), the AVGPC is symmetrized by a conditional pdf $\varphi_{x^d}(s^d)=\varphi(s^d|x^d)$ if
\begin{align}
\label{eq:GPsymmetrizableEq}
\int_{-\infty}^\infty\cdots \int_{-\infty}^\infty  \varphi_0(s^d)f_{Z^d}(y^d-x^d-s^d)ds^d=
\int_{-\infty}^\infty\cdots \int_{-\infty}^\infty  \varphi_{x^d}(s^d)f_{Z^d}(y^d-s^d)ds^d
 \,, 
\end{align}
for all $x^d,y^d\in\mathbb{R}^d$.
By substituting $z^d=y^d-x^d-s^d$ in the LHS, and $\bar{z}^d=y^d-s^d$ in the RHS, this is equivalent to 
\begin{align}
\label{eq:GPsymmetrizableEq2}
\int_{-\infty}^\infty\cdots \int_{-\infty}^\infty  \varphi_0(y^d-x^d-z^d)f_{Z^d}(z^d)dz^d=
\int_{-\infty}^\infty\cdots \int_{-\infty}^\infty  \varphi_{x^d}(y^d-\bar{z}^d)f_{Z^d}(\bar{z}^d)d\bar{z}^d \,.
\end{align}
For every $x^d\in\mathbb{R}^d$, define the random vector $\oS^d(x^d)\sim\varphi_{x^d}$.
We note that the RHS is the convolution of the pdfs of the random vectors $Z^d$ and 
$\oS^d(x^d) $, while the LHS is the convolution of the pdfs of the random vectors $Z^d$ and 
$\oS^d(0)+x^d $. This is not surprising since the channel output $Y^d$ is  a sum of independent random vectors, and thus the pdf of $Y^d$ is a convolution of pdfs.
It follows that $\varphi_0(y^d-x^d)=\varphi_{x^d}(y^d)$, and
 by plugging $s^d$ instead of $y^d$, we have that $\varphi_{x^d}$ symmetrizes the AVGPC if and only if 
\begin{align}
\varphi_{x^d}(s^d)=\varphi_0 (s^d-x^d) \,.
\label{eq:GPvarphiSymmEq}
\end{align} 
Then, the corresponding state cost satisfies
\begin{align}
&\int_{-\infty}^\infty \cdots \int_{-\infty}^\infty f_{X^d}(x^d) \varphi_{x^d}(s^d) \norm{s^d}^2 
\, dx^d \, ds^d
\nonumber\\
=& \int_{-\infty}^\infty \cdots \int_{-\infty}^\infty f_{X^d}(x^d) \varphi_0 (s^d-x^d) \norm{s^d}^2 \, ds^d
\, dx^d
\nonumber\\
=& \int_{-\infty}^\infty \cdots \int_{-\infty}^\infty f_{X^d}(x^d) \varphi_0 (a^d) \norm{a^d+x^d}^2 \, da^d
\, dx^d
\nonumber\\
=& \int_{-\infty}^\infty \cdots \int_{-\infty}^\infty 
\left[ \int_{-\infty}^\infty \cdots \int_{-\infty}^\infty \norm{x^d+a^d}^2
f_{X^d}(x^d) \, dx^d
\right]
 \varphi_0 (a^d)  \, da^d
\label{eq:GPscostEq1}
\end{align}
where the second equality follows by the integral substitution of $a^d=s^d-x^d$.
Observe that the bracketed integral can be expressed as 
\begin{align}
\int_{-\infty}^\infty \cdots \int_{-\infty}^\infty \norm{x^d+a^d}^2
f_{X^d}(x^d) \, dx^d =\E\norm{X^d+a^d}^2=\trace(K_X)+\norm{a^d}^2 \,.
\end{align}
Thus, by (\ref{eq:GPscostEq1}), 
\begin{align}
&\int_{-\infty}^\infty \cdots \int_{-\infty}^\infty f_{X^d}(x^d) \varphi_{x^d}(s^d) \norm{s^d}^2 
\, dx^d \, ds^d
\nonumber\\
=&\trace(K_X)+\int_{-\infty}^\infty \cdots \int_{-\infty}^\infty \norm{a^d}^2  \varphi_0 (a^d)  \, da^d
\nonumber\\
\geq& \trace(K_X) \,.
\label{eq:GPCscostTr}
\end{align}
Note that 
the last inequality holds for any $\varphi_{x^d}$ which symmetrizes the channel. 
Now, observe that (\ref{eq:GPvarphiSymmEq}) holds
 for $\hat{\varphi}_{x^d}(s^d)=\delta(s^d-x^d)$, where $\delta(\cdot)$ is the Dirac delta function, 
hence $\hat{\varphi}_{x^d}$ symmetrizes the channel.
In addition, since $\hat{\varphi}_0$ gives probability $1$ to $S^d=0$, we have that 
 (\ref{eq:GPCscostTr}) holds with equality for $\hat{\varphi}_{x^d}$,
and thus, $\LambdaOig( F_{X^d} )=\trace(K_X)$. 
\qed

\section{Proof of Theorem~\ref{theo:GPavcDet}}
\label{app:GPavcDet}
Consider the AVGPC under input constraint $\plimit$ and state constraint $\Lambda$.

\subsection*{Achievability Proof}
Assume that $\plimit>\Lambda$. We show that $\sigmaCavc\geq\sigmaICavc=\sigmarICav$.
 By \cite[Theorem 3]{Csiszar:92p}, 
if there exists an input distribution $F_{X^d}$ such that $\LambdaOig(F_{X^d})>\Lambda$,
then the capacity is given by
\begin{align}
\label{eq:GPdirCN}
\sigmaCavc=\max_{
\substack{ F_{X^d}\,:\; \sum_{j=1}^d P_j \leq \plimit  \\   \LambdaOig(F_{x^d})\geq \Lambda }
}  \;
\min_{F_{S^d}\,:\; \sum_{j=1}^d N_j \leq \Lambda } I(X^d;Y^d) \,,
\end{align}
where $P_j=\E X_j^2$ and $N_j=\E S_j^2$.

Consider the input distribution $F_{X^d}$ of a Gaussian vector $X^d\sim\mathcal{N}(\mathbf{0},K_X)$, where
the covariance matrix is given by $K_X=\diag(P_1^*,\ldots,P_d^*)$. 
By Lemma~\ref{lemm:GPscostP}, we have that 
\begin{align}
\LambdaOig(F_{X^d}) = \trace(K_X)=\sum_{j=1}^d P_j^*=\plimit.
\end{align}
Having assumed that $\plimit>\Lambda$, it follows that $\LambdaOig(F_{X^d})>\Lambda$, hence (\ref{eq:GPdirCN}) applies. Then, setting $X^d\sim\mathcal{N}(\mathbf{0},K_X)$ yields 
\begin{align}
\label{eq:GPdirCNgeq}
\sigmaCavc\geq& \min_{F_{S^d}\,:\; \sum_{j=1}^d N_j \leq \Lambda } I(X^d;Y^d) \\
\geq&  \min_{F_{S^d}\,:\; \sum_{j=1}^d N_j \leq \Lambda } \sum_{j=1}^d 
I(X_j;Y_j) 
\label{eq:sigmad12dirmaxmin}\\
\geq& \min_{F_{S^d}\,:\; \sum_{j=1}^d N_j \leq \Lambda } \sum_{j=1}^d 
\frac{1}{2}\log\left(1+\frac{P_j^*}{N_j+\sigma_j^2}  \right) \,,
\label{eq:sigmad4dirmaxmin}
\end{align}
where the second inequality holds as $X_1,\ldots,X_d$ are independent and since 
conditioning reduces entropy, and  the last inequality holds since Gaussian noise is known to be the worst additive noise under variance constraint \cite[Lemma II.2]{DiggaviCover:01p}.


From this point, we use the considerations given in \cite{HughesNarayan:88p}.
%
To prove the direct part, it remains to show that the assignment of $N_j=N_j^*$, for $j\in [1:d]$, is optimal in the RHS of
(\ref{eq:sigmad4dirmaxmin}),  where $N_j^*$ are as defined in (\ref{eq:GPNjdef})-(\ref{eq:GPbetadef}). An assignment of $N_1,\ldots,N_d$ is optimal if and only if it satisfies the KKT optimality conditions \cite[Section 5.5.3]{BoydVandenbergh:04b}, 
\begin{align}
& \sum_{j'=1}^d N_{j'}=\Lambda \,,\; N_j\geq 0 \,,\;    									\label{eq:nBasicCond} \\
& \frac{P_j^*}{(N_j+\sigma_j^2)\cdot(N_j+\sigma_j^2+P_j^*)} \leq \theta \,,		\label{eq:nIneqCond} \\
&  \left(  \theta-\frac{P_j^*}{(N_j+\sigma_j^2)\cdot(N_j+\sigma_j^2+P_j^*)}    \right)N_j=0 
\label{eq:nSlackCond} \,,
\end{align}
for $j\in [1:d]$, where $\theta>0$ is a Lagrange multiplier.

 We claim that the conditions are met by
\begin{align}
\theta=\theta^*\triangleq\frac{\alpha-\beta}{\alpha\beta} \,,\;\text{and }\;
 N_j=N_j^* \,,\;\text{for $j\in [1:d]$} \,.
\label{eq:WaterSol}
\end{align}
Condition (\ref{eq:nBasicCond}) is met by the definition of $N_j^*$, $j\in [1:d]$, in 
(\ref{eq:GPNjdef})-(\ref{eq:GPbetadef}).
%
Let $j\in [1:d]$ be a given channel index.
We consider the following cases. Suppose that $N_j^*=0$.
Then, Condition (\ref{eq:nSlackCond}) is clearly satisfied. 
Now, if $P_j^*=0$, then Condition (\ref{eq:nIneqCond}) is satisfied since $\alpha>\beta$ by part 1 of 
Lemma~\ref{lemm:WaterProp}.
Otherwise, $0<P_j^*=\alpha-(N_j^*+\sigma_j^2)=\alpha-\sigma_j^2$, and then
\begin{align}
\frac{P_j^*}{(N_j+\sigma_j^2)\cdot(N_j+\sigma_j^2+P_j^*)}=
\frac{\alpha-\sigma_j^2}{\sigma_j^2 \alpha}\leq \frac{\alpha-\beta}{\alpha\beta}=\theta^* \,,
\end{align}
where the last inequality holds since $N_j^*=0$ only if $\beta\leq\sigma_j^2$. 
Thus, Condition (\ref{eq:nIneqCond}) is satisfied.

Next, suppose that $N_j^*>0$, hence $N_j^*+\sigma_j^2=\beta$. By part 2 of Lemma~\ref{lemm:WaterProp}, this implies that $P_j^*>0$, \ie $P_j^*=\alpha-(N_j^*+\sigma_j^2)=\alpha-\beta$. Thus,
\begin{align}
\frac{P_j^*}{(N_j+\sigma_j^2)\cdot(N_j+\sigma_j^2+P_j^*)} =
\frac{\alpha-\beta}{\beta\cdot\alpha}=\theta^* \,,
\end{align}
and thus Condition (\ref{eq:nIneqCond}) is satisfied with equality, and  Condition (\ref{eq:nSlackCond}) is satisfied as well.

As the KKT conditions are satisfied under (\ref{eq:WaterSol}), we deduce that the assignment of $N_j=N_j^*$, $j\in [1:d]$, minimizes the RHS of (\ref{eq:sigmad4dirmaxmin}). Together with (\ref{eq:sigmad4dirmaxmin}), this implies that $\sigmaCavc\geq \sigmarICav$
for $\plimit>\Lambda$.

\subsection*{Converse Proof}
We use a similar technique as in \cite{CsiszarNarayan:91p} (see also \cite{Ericson:85p,BBT:60p}).
In general, the deterministic code capacity is bounded by the random code capacity, hence 
$\sigmaCavc\leq \sigmarCav=\sigmarICav$, by Theorem~\ref{theo:GPavcRand}. It remains to show that if 
$\plimit\leq\Lambda$, then the capacity is zero.
Suppose that $\plimit\leq\Lambda$, and assume to the contrary that there exists an achievable rate $R>0$.
 Then, there exists a sequence of $(2^{nR},n,\eps_n)$ codes 
$\code_n=(\fvec^d,g)$ for the AVGPC such that $\eps_n\rightarrow 0$ as $n\rightarrow\infty$, where
the size of the message set is at least $2$, \ie
$
\dM\triangleq 2^{nR} \geq 2 
$. 
 
Consider a jammer who chooses the state sequence from the codebook uniformly at random, \ie 
 $\Svec^d=\fvec^d(M')$, where $M'$ is uniformly distributed over $[1:\dM]$. This choice meets the state constraint, since the square norm of the state sequence is $\norm{\Svec^d}^2\leq\plimit\leq \Lambda$.  
The average probability of error is then bounded by
\begin{align}
\err(F_{\Svec^d},\code)= \frac{1}{\dM^2} \sum_{m=1}^{\dM} \sum_{m'=1}^{\dM} 
\int_{\Dset_e(m,m')} f_{\Zvec^d}(\zvec^d) d\zvec^d \,,
\end{align}
where $f_{\Zvec^d}(\zvec^d)=\prod_{j=1}^d \frac{1}{(2\pi\sigma_j^2)^{n/2}} 
e^{-\norm{\zvec_j}^2/2\sigma_j^2}$, and
\begin{align}
\Dset_e(m,m')=\{ \zvec^d \,:\; g(\fvec^d(m)+\fvec^d(m')+\zvec^d)\neq m \} \,.
\end{align}
By interchanging the summation variables $m$ and $m'$, we now have that
\begin{align}
&\err(F_{\Svec^d},\code)= \frac{1}{2\dM^2} \sum_{m,m'}  
\int_{\Dset_e(m,m')} f_{\Zvec^d}(\zvec^d) d\zvec^d 
+\frac{1}{2\dM^2} \sum_{m,m'}  
\int_{\Dset_e(m',m)} f_{\Zvec^d}(\zvec^d) d\zvec^d 
\nonumber\\
\geq & \frac{1}{2\dM^2} \sum_{m,m' \,:\; m\neq m'}  
\int_{\Dset_e(m,m')\cup \Dset_e(m,m')} f_{\Zvec^d}(\zvec^d) d\zvec^d \,.
\end{align}
Next, observe that for $m\neq m'$, $\Dset_e(m,m')\cup \Dset_e(m,m')=\mathbb{R}^{nd}$, and thus the probability of error is lower bounded by
\begin{align}
\err(F_{\Svec^d},\code)\geq \frac{\dM(\dM-1)}{2\dM^2} \geq \frac{1}{4} \,,
\end{align}
where the last inequality holds since $\dM\geq 2$.
Hence, the assumption is false and a positive rate cannot be achieved when $\plimit\leq\Lambda$.
This completes the proof of the converse part.
\qed

\section{Proof of Theorem~\ref{theo:sKGPavcRand}}
\label{app:KGPavcRand}
Consider the AVC with colored Gaussian noise. First, we show that the problem can be transformed into that of an AVC with fixed parameters.
Then, we derive a limit expression for the random code capacity, and prove the capacity characterization in Theorem~\ref{theo:sKGPavcRand} using the Toeplitz matrix properties in the auxiliary lemma below. To derive the deterministic code capacity, we use  similar symmetrizability and optimization arguments as in our proofs for the Gaussian product channel.

\begin{lemma} \cite[Section 2.3]{Ebert:66p} (see also {\cite{Gray:06n,Holsinger:64p} \cite[Section 8.5]{Gallager:68b}})
\label{lemm:Ebert}
Let $\Psi_Z(\omega)$ be the power spectral density of a zero mean stationary process $\{Z_i \}_{i=1}^{\infty}$.
Assume that $\Psi_Z: [-\pi,\pi]\rightarrow [0,\nu] $ is bounded and integrable, for some $\nu>0$, and denote the auto-correlation function by
\begin{align}
r_Z(\ell)=  \frac{1}{2\pi} \int_{-\pi}^{\pi} \Psi_Z(\omega) e^{j\omega} \, d\omega \,,\; \ell=0,1,2,\ldots
\end{align}
with $j=\sqrt{-1}$.
For a sequence $\Zvec$ of length $n$, let $\sigma_1^2,\ldots,\sigma_n^2$ denote the eigenvalues of the $n\times n$ covariance matrix $K_Z$, where $K_Z(i,j)=r_Z(|i-j|)$ for $i,j\in [1:n]$. Then, for every real, monotone non-increasing, and bounded function 
$G: [0,\nu]\rightarrow [0,\eta] $, 
\begin{align}
\lim_{n\rightarrow\infty} \frac{1}{n} \sum_{i=1}^{\infty} G(\sigma_i^2)
=\frac{1}{2\pi} \int_{-\pi}^\pi   G(\Psi_Z(\omega)) \, d\omega
\end{align} 
if the integral exists.
\end{lemma}

\subsection{Transformation to AVC with Fixed Parameters}
Let $K_Z$ denote the $n\times n$ covariance matrix of the noise sequence $\Zvec$. 
Consider the eigen decomposition of the covariance matrix $K_Z$, and denote the eigenvector and eigenvaule matrices  by $Q$ and $\Sigma$, respectively, \ie
\begin{align}
& K_Z=Q\Sigma Q^T \,,
\;\text{where  
 $QQ^T=I$ and 
 $\Sigma=\diag\{\sigma_1^2,\ldots,\sigma_n^2\}$} \,.
\end{align}
We claim that the capacity of the AVC with colored Gaussian noise is the same as the capacity of the following AVC,
\begin{align}
&\Yvec'=\Xvec'+\Zvec'+\Svec' \,,\; 
\label{eq:XYdiag}
\end{align}
where $\Xvec'=Q^T \Xvec $, $\Zvec'=Q^T \Zvec$, and $\Svec'=Q^T \Svec$. Since $Q$ is a unitary matrix, \ie
$Q^{-1}=Q^T$,
the input and state constraints remain the same, as 
$\norm{\Xvec'}^2=(\Xvec')^T \Xvec'=\Xvec^T Q Q^T \Xvec =\Xvec^T \Xvec=\norm{\Xvec}^2\leq n\plimit$, and similarly, $\norm{\Svec'}^2=\norm{\Svec}^2\leq n\Lambda$. Furthermore, the noise covariance matrix is now
\begin{align}
K_{Z'}
= Q^T K_Z Q=\Sigma=\diag\{ \sigma_1^2,\ldots,\sigma_n^2 \} \,.
\end{align}
This transformation can be thought of as a linear system, which is \emph{not} time invariant. Hence, the  noise of the transformed channel is a Gaussian process, but it is non-stationary.
Thereby, the input-output relation above specifies a time varying channel, $\{ F_{Y_1,\ldots,Y_n|X_1,\ldots,X_n,S_1,\ldots,S_n} \}_{n=1}^\infty$.
From operational perspective,  if there exists a $(2^{nR},n,\eps)$ code $\code=(\fvec,g)$ for the original AVC with colored Gaussian noise, then the code $\code'=(\fvec',g')$, given by 
$\fvec'(m)=Q^T \fvec(m)$ and $g'(\yvec')=
g(Q\yvec')$, is a
$(2^{nR},n,\eps)$ code for the transformed AVC in (\ref{eq:XYdiag}). Similarly,
if there exists a $(2^{nR},n,\eps)$ code $\code'=(\fvec',g')$ for the transformed AVC, then the code $\code=(\fvec,g)$, given by 
$\fvec(m)=Q \fvec'(m)$ and $g(\yvec)=g'(Q^T \yvec)$, is a
$(2^{nR},n,\eps)$ code for the original AVC. Thus, the original AVC and the transformed AVC have the same operational capacity.

Therefore, we can assume without loss of generality that the noise sequence has independent components $Z_i\sim\mathcal{N}(0,\sigma_i^2)$, $i\in [1:n]$. 
Assume, at first, that $\sigma_i^2\in \Tset$ for $i\in [1:n]$, with some set $\Tset$ of finite size, which does not grow with $n$,
and that $\sigma_i^2>\delta$, where $\delta>0$ is arbitrarily small.
Hence, observe that the channel in (\ref{eq:XYdiag}) is equivalent to a channel $W_{Y''|X'',S'',T''}$ with fixed parameters, specified by  
\begin{align}
Y''=X''+S''+Z''_t \,,\; \text{where }\; Z''_t \sim \mathcal{N}(0,t^2) 
\label{eq:sChanRP}
\end{align}
with the parameter sequence $\sigma_1, \sigma_2,\ldots$.
It is left to determine the random code capacity and deterministic code capacity of the Gaussian AVC with fixed parameters in (\ref{eq:sChanRP}). Although we previously assumed in Sections \ref{sec:Pchannels} and \ref{sec:Pmain} that the input, state, and output alphabets are finite, our results can be extended to the continuous case as well, using standard discretization techniques \cite{BBT:59p,Ahlswede:78p} \cite[Section 3.4.1]{ElGamalKim:11b}.

Now, consider the double water filling allocation, 
\begin{align}
b_i^*=&  \left[ \beta'-\sigma_i^2 \right]_{+}  \,,\; 
\label{eq:KGPbidef}
\\
a_i^*=&\left[ \alpha'-(b_i^*+\sigma_i^2) \right]_{+} \,,\; 
\label{eq:KGPaidef}
\end{align}
for $i\in [1:n]$,
 where $\beta'> 0$ and $\alpha'> 0$ are chosen to satisfy
$
\frac{1}{n}\sum_{i=1}^n \left[ \beta'-\sigma_i^2 \right]_{+}=\Lambda 
$ 
and
$
\frac{1}{n}\sum_{i=1}^n \left[ \alpha'-(b_i^*+\sigma_i^2) \right]_{+}=\plimit 
$, respectively. 
Define
\begin{align}
\KrICav\triangleq \frac{1}{2n}\sum_{i=1}^n\log\left( 1+\frac{a_i^*}{b_i^*+\sigma_i^2} \right) \,.
\label{eq:KrICavdef}
\end{align}

\subsection{Random Code Capacity}
Now that we have shown that the problem reduces to that of an AVC with fixed parameters, we have by
Corollary~\ref{coro:PrCavE} that the random code capacity is given by 
\begin{align}
\sKrCav= \liminf_{n\rightarrow\infty}  
\max_{ \substack{ P_1,\ldots,P_n \,:\;\\  \frac{1}{n} \sum_{i=1}^n P_i\leq \plimit } } \;
\min_{ \substack{ N_1,\ldots,N_n \,:\; \\ \frac{1}{n} \sum_{i=1}^n N_i \leq \Lambda } }
\frac{1}{n} \sum_{i=1}^n  
\inC_{\sigma_i}^{\rstarC}\hspace{-0.05cm}(P_i,N_i) \,,
\end{align}
where  $\inC_{\sigma}^{\rstarC}\hspace{-0.05cm}(P,N) $ is the random code capacity of the traditional AVC under input constraint 
$P$ and state constraint $N$.
Hughes and Narayan \cite{HughesNarayan:87p} showed that the random code capacity of such a channel, 
where the noise sequence  is i.i.d. $\sim\mathcal{N}(0,\sigma^2)$, is given by
\begin{align}
\inC_{\sigma}^{\rstarC}\hspace{-0.05cm}(P,N) =\frac{1}{2} \log\left( 1+\frac{P}{N+\sigma^2}   \right) \,.
\end{align}
Hence, for the AVC with colored Gaussian noise,
\begin{align}
\sKrCav= \liminf_{n\rightarrow\infty} 
\min_{ \substack{ N_1,\ldots,N_n \,:\; \\ \frac{1}{n} \sum_{i=1}^n N_i \leq \Lambda } } 
\max_{ \substack{ P_1,\ldots,P_n \,:\;\\  \frac{1}{n} \sum_{i=1}^n P_i\leq \plimit } }\;
\frac{1}{2n} \sum_{i=1}^n  \log\left( 1+\frac{P_i}{N_i+\sigma_i^2}   \right)
 \,.
\label{eq:sRandDir0}
\end{align}
Next, observe that this is the same min-max optimization as for the AVGPC in (\ref{eq:GPoptimi}), due to \cite{HughesNarayan:88p}, 
 with $d\leftarrow n$, $\plimit\leftarrow (n\plimit)$, $\Lambda\leftarrow (n\Lambda)$.
 Therefore, by Theorem~\ref{theo:GPavcRand} \cite{HughesNarayan:88p} and (\ref{eq:sRandDir0}),  
\begin{align}
\sKrCav=& 
\liminf_{n\rightarrow\infty}\KrICav
 \,.
\label{eq:sRandDir1}
\end{align}

Given a bounded power spectral density
$\Psi_Z: [-\pi,\pi]\rightarrow [0,\nu]$, define a function $G: [0,\nu]\rightarrow [0,\eta]$ by
\begin{align}
G(x)=\frac{1}{2}\log\left(  
1+\frac{\left[ \alpha'-\left[ \beta'+x \right]_{+} \right]_{+}}{\left[ \beta'-x \right]_{+}+x} 
\right)=\begin{cases}
\frac{1}{2}\log\left( \frac{\alpha'}{\beta'} \right)	 &\text{if $x<\beta'$}\\
\frac{1}{2}\log\left( \frac{\alpha'}{x} \right) 	     &\text{if $\beta'\leq x<\alpha'$}\\
0																								       &\text{if $x\geq\alpha'$}
\end{cases}
\label{eq:GcolAVC}
\end{align}
and observe that 
\begin{align}
\KrICav= \frac{1}{n}\sum_{i=1}^n G(\sigma_i^2) \,.
\end{align}
As $G(x)$ is non-increasing and bounded by $\eta=\frac{1}{2}\log[1+\plimit/\delta]$, we have by Lemma~\ref{lemm:Ebert} that
\begin{align}
\liminf_{n\rightarrow\infty}\KrICav =\frac{1}{2\pi} \int_{-\pi}^\pi   G(\Psi_Z(\omega)) \, d\omega \,.
\label{eq:sRandDir4}
\end{align}
Observing that the function defined in (\ref{eq:GcolAVC}) is also continuous, while $\Psi_Z(\omega)$ is bounded and integrable, 
it follows that the integral exists \cite[Theorem 6.11]{Rudin:76b}. 
Plugging (\ref{eq:GcolAVC}) into the RHS of (\ref{eq:sRandDir4}), we obtain
\begin{align}
\liminf_{n\rightarrow\infty}\KrICav= \frac{1}{2\pi}\int_{-\pi}^{\pi} \frac{1}{2}\log\left( 1+\frac{\left[ \alpha-\left[ \beta+\Psi_Z(\omega) \right]_{+} \right]_{+}}{\left[ \beta-\Psi_Z(\omega) \right]_{+}+\Psi_Z(\omega)} \right) \,d\omega
\label{eq:sRandDir5}
\end{align}
where $\beta$ and $\alpha$ satisfy (\ref{eq:sKGPbetadef}) and (\ref{eq:sKGPalphadef}), respectively. 
Since the covariance matrix of the stationary noise process is Toeplitz (see \eg \cite{Gray:06n}), the density of eigenvalues on the real line tends to the power spectral density \cite{GrenanderSzego:01b}.
Given that the power spectral density is bounded and integrable, we have that the sequence of eigenvalues 
$\sigma_1^2,\sigma_2^2,\ldots$ is summable \cite[Theorem 4.2]{Gray:06n}, and thus, bounded as well.
Hence, we can remove the assumption that the set of noise variances has finite cardinality, by quantization of the variances.
The random code characterization now follows from (\ref{eq:sRandDir1}) and (\ref{eq:sRandDir5}).

\subsection{Deterministic Code Capacity}
Moving to the deterministic code capacity, observe that for a constant-parameter Gaussian AVC,
where the noise sequence  is i.i.d. $\sim\mathcal{N}(0,\sigma^2)$, we have that $\tLambda( F_X ,\sigma )=\E X^2$, by
Lemma~\ref{lemm:GPscostP}, taking $d=1$. Therefore, for the Gaussian AVC with a parameter sequence $\sigma_1^2,\ldots,\sigma_n^2$,
\begin{align}
L_n^* 
= \max_{F_{X|T} \,:\; \frac{1}{n} \sum_{i=1}^n \E[ X^2|T=\sigma_i] \leq\plimit} \frac{1}{n} \sum_{i=1}^n \tLambda(F_{X|T=\sigma_i},\sigma_i)=\max_{F_{X|T} \,:\; \frac{1}{n} \sum_{i=1}^n \E[ X^2|T=\sigma_i] \leq\plimit} \frac{1}{n} \sum_{i=1}^n 
\E[ X_i^2 | T=\sigma_i] =\plimit \,,
\label{eq:GlnS}
\end{align}
where the first equality holds by the definition of $L_n^*$ in (\ref{eq:1Lstar}) and by  (\ref{eq:LambdaOigEq}).
It can further be seen from the proof of Lemma~\ref{lemm:GPscostP} in Appendix~\ref{app:GPscostP} that the Gaussian channel
$Y=X+S+Z_{\sigma}$
 is symmetrized by a distribution $\varphi(s|x)$ that gives probability $1$ to $S=x$, and that the minimum in the formula of 
$\tLambda(F_X,\sigma)$ in (\ref{eq:LambdaOig1}) is attained with this distribution.
 
Therefore, by Corollary~\ref{coro:PrCavEDet}, the capacity of the AVC with colored Gaussian noise is given by the limit inferior of
\begin{align}
\inR_n(\avc)=& \begin{cases}
\min\limits_{ \substack{ N_1,\ldots,N_n \,:\; \\ \frac{1}{n} \sum_{i=1}^n N_i \leq \Lambda } } \;
\max\limits_{ \substack{ P_1,\ldots,P_n, \tlambda_1,\ldots\tlambda_n \,:\;\\  
\frac{1}{n} \sum_{i=1}^n P_i\leq \plimit \,,
\frac{1}{n} \sum_{i=1}^n \tlambda_i\geq \Lambda } }
\frac{1}{n} \sum\limits_{i=1}^n
\inC_{\sigma_i}(P_i,\tlambda_i,N_i)  &\text{if $L_n^*> \Lambda$}\,,\\
  0 																 &\text{if $L_n^*\leq \Lambda$}
\end{cases}	
\label{eq:CtolDet11}
\intertext{where }
\inC_\sigma(P,\Delta,N)=&
\min\limits_{  F_{S''} \,:\; \E S''^2 \leq N } \;
\max\limits_{ \substack{ F_{X''} \,:\; \E\, X''^2 \leq P \,,\; \\ \tLambda_\sigma( F_{X''}, \sigma )\geq \Delta } } \;   I_q(X'';Y''|T''=\sigma) \,.
\label{eq:CtolDet12}
\end{align}

Consider the direct part. Suppose that $\plimit>\Lambda$, hence $L_n^*>\Lambda$ (see (\ref{eq:GlnS})), and set $P_i=\tlambda_i=a_i^*$ for $i\in [1:n]$. This choice of parameters satisfies the optimization constraints in (\ref{eq:CtolDet11}), as
$\sum_{i=1}^n P_i=\plimit$, and also $\sum_{i=1}^n \tlambda_i=\plimit>\Lambda$. Therefore, 
\begin{align}
\inR_n(\avc)\geq & 
\min\limits_{ \substack{ N_1,\ldots,N_n \,:\; \\ \frac{1}{n} \sum_{i=1}^n N_i \leq \Lambda } } \; 
\frac{1}{n} \sum\limits_{i=1}^n
\inC_{\sigma_i}(a_i^*,a_i^*,\lambda_i)  	
=	\min\limits_{ \substack{ N_1,\ldots,N_n , F_{S''^n} \,:\; \\  \E S_i''^2\leq N_i \,,\;
 \frac{1}{n} \sum_{i=1}^n N_i \leq \Lambda }} 
\frac{1}{n} \sum\limits_{i=1}^n I_q(X_i'';Y_i''|T_i''=\sigma_i) \,,
	\nonumber\\
	\geq & \min\limits_{  N_1,\ldots,N_n \,:\;   \sum_{i=1}^n N_i \leq n\Lambda  }
\frac{1}{n} \sum\limits_{i=1}^n \frac{1}{2} \log \left(1+\frac{a_i^*}{N_i+\sigma_i^2} \right)
\label{eq:CtolDet13}
\end{align}
where the the last inequality holds since Gaussian noise is known to be the worst additive noise under variance constraint \cite[Lemma II.2]{DiggaviCover:01p}.
Next, observe that this is the same minimization 
as in (\ref{eq:sigmad4dirmaxmin}), in the proof of the direct part 
for the AVGPC, with
$d\leftarrow n$, $\plimit\leftarrow (n\plimit)$, $\Lambda\leftarrow (n\Lambda)$
(see proof of Theorem~\ref{theo:GPavcDet}  in Appendix~\ref{app:GPavcDet}). Therefore, the minimum is attained with $N_i=b_i^*$, and
 the RHS of (\ref{eq:sRandDir1}) is achievable with deterministic codes as well, provided that $\plimit>\Lambda$.

The converse part is straightforward. Since the deterministic code capacity is always bounded by the random code capacity, we have that
$\sKCavc\leq\sKrCav=\sKrICav$. If $\plimit\leq \Lambda$, then $L_n^*\leq\Lambda$ by (\ref{eq:GlnS}), hence $\KCavc=\liminf \inR_n(\avc)=0$ by the second part of Corollary~\ref{coro:PrCavEDet}. 
\qed

\end{appendices}


\ifdefined\bibstar\else\newcommand{\bibstar}[1]{}\fi

\end{document}